\documentclass[10pt]{article}  %%Formatting for STOC: 11pt
\usepackage[T1]{fontenc}
\usepackage[utf8]{inputenc}
\usepackage[a4paper,width=150mm,top=25mm,bottom=25mm]{geometry}
\usepackage{macros}
\usepackage{thmtools,thm-restate}

\bibliography{refs}

\newcommand{\betasecond}{\beta_{\mathsf{2nd}}}
\newcommand{\betaRS}{\beta_{\mathsf{RS}}}
\newcommand{\betaSL}{\beta_{\mathsf{SL}}}
\newcommand{\betashatter}{\beta_{\mathsf{shatter}}}
\newcommand{\betaunique}{\beta_{\mathsf{unique}}}
\newcommand{\CW}{\mathsf{CW}}

\newcommand{\GOE}{\mathsf{GOE}}
\newcommand{\nul}{\mathsf{null}}
\newcommand{\pl}{\mathsf{pl}}
\newcommand{\ov}{\mathsf{ov}}
\newcommand{\UC}{\mathsf{UC}}
\newcommand{\Val}{\mathsf{Val}}
\newcommand{\romI}{\textup{I}}
\newcommand{\romII}{\textup{II}}
\newcommand{\tkappa}{{\tilde\kappa}}
\newcommand{\tlambda}{{\tilde\lambda}}
\newcommand{\tOmega}{{\tilde \Omega}}
\newcommand{\hOmega}{{\hat \Omega}}

\newcommand{\Ct}{\mathsf{Ct}}
\newcommand{\Cy}{\mathsf{Cy}}
\newcommand{\PM}{\mathsf{PM}}
\newcommand{\op}{\mathsf{op}}
\newcommand{\Ev}{\mathsf{Ev}}

\title{On Zeros and Algorithms for Disordered Systems: Mean-Field Spin Glasses}
\author{Ferenc Bencs\thanks{Email: \texttt{ferenc.bencs@gmail.com},
Centrum Wiskunde \& Informatica. Supported by the Netherlands Organisation of Scientific Research (NWO): VI.Veni.222.303.} \and Brice Huang\thanks{Email: \texttt{bmhuang@stanford.edu},
Stanford University. Supported by a Stanford Science Fellowship and NSF Mathematical Sciences Postdoctoral Research Fellowship.} \and Daniel Z. Lee\thanks{Email: \texttt{lee\_d@mit.edu},
Massachusetts Institute of Technology. Supported by the NSF CAREER grant CCF-2443045, and the Reed Fund at MIT.} \and Kuikui Liu\thanks{Email: \texttt{liukui@mit.edu},
Massachusetts Institute of Technology. Supported by the NSF CAREER grant CCF-2443045, and the Reed Fund at MIT.} \and Guus Regts\thanks{Email: \texttt{guusregts@gmail.com},
Korteweg de Vries Institute for Mathematics, University of Amsterdam. Supported by the Netherlands Organisation of Scientific Research (NWO): VI.Vidi.193.068.}}
\date{}

\begin{document}
\maketitle

\begin{abstract}
Spin glasses are fundamental probability distributions at the core of statistical physics, the theory of average-case computational complexity, and modern high-dimensional statistical inference. In the mean-field setting, we design deterministic quasipolynomial-time algorithms for estimating the partition function to arbitrarily high accuracy for all inverse temperatures in the second moment regime. In particular, for the Sherrington--Kirkpatrick model, our algorithms succeed for the entire replica-symmetric phase. To achieve this, we study the locations of the zeros of the partition function. Notably, our methods are conceptually simple, and apply equally well to the spherical case and the case of Ising spins.

% \dan{
%   We design a deterministic quasi-polynomial time algorithm for estimating the partition function of mean field spin glass models up to the second moment threshold. For the Sherrington-Kirkpatrick model in particular, the algorithm works in the entire replica-symmetric phase. In general, the second moment threshold is within a constant factor of the replica symmetric threshold. This greatly extends the range in which efficient counting or sampling algorithms are currently known. 

%   Conceptually, our methods are simple and represent the first use of the Barvinok zero-freeness framework in the average case setting. We believe that extensions to the ideas developed here will be useful for further algorithmic applications, and motivate interesting and tractable open questions on the geometry of random analytic functions.
% }
\end{abstract}

%\newpage
{\small
\tableofcontents
}

%%%%%%%%%%%%%%%%%%%%%
% Formating for STOC
%%%%%%%%%%%%%%%%%%%%%%%%
\setlength{\parskip}{0.5em}  % space between paragraphs
%%%%%%%%%%%%%%%%%%%%%%%%%

%%%REmove newpage
% \newpage

\section{Introduction}

We study \emph{spin glasses}, which are probabilistic models of inhomogeneous materials (e.g. disordered magnetic alloys) in statistical mechanics. These and related models have also formed the bedrock for the study of average-case algorithm design, as well as the computational complexity of high-dimensional statistical inference problems; see e.g. \cite{Bov06, MS24} and references therein. For concreteness, we consider \emph{mixed $p$-spin models (with Gaussian interactions)}, which are specified by the random Hamiltonian
\begin{align}\label{eq:pspin-hamiltonian}
    \mathcal{H}_{\bm{G}}(\sigma) \defeq \sum_{p=2}^{p_{\max}} \frac{\upgamma_{p}}{n^{\frac{p-1}{2}}} \sum_{i_{1},\dots,i_{p}=1}^{n} \bm{G}_{i_{1},\dots,i_{p}} \prod_{j=1}^{p} \sigma_{i_{j}}, \qquad \forall \sigma \in \R^{n},
\end{align}
where $\{\upgamma_{p}\}_{p=2}^{p_{\max}}$ are deterministic nonnegative real coefficients, $p_{\max} \geq 2$ is an arbitrary integer, and $\bm{G}_{i_{1},\dots,i_{p}} \sim \mathcal{N}(0,1)$ are independent across all $2 \leq p \leq p_{\max}$ and $i_{1},\dots,i_{p} \in [n]$. These models are \emph{mean-field} in the sense that each spin $\sigma_{i}$ interacts directly with every other spin $\sigma_{j}$ in $\mathcal{H}_{\bm{G}}(\sigma)$. For a background probability measure $\varrho$ on $\R^{n}$ and an \emph{inverse temperature} $\beta \geq 0$, define the corresponding (random) \emph{Gibbs measure} on $\R^{n}$ by
\begin{align*}
    d\mu_{\bm{G},\beta}(\sigma) \propto \exp\wrapp{\beta \cdot \mathcal{H}_{\bm{G}}(\sigma)} \,d\varrho(\sigma),
\end{align*}
where the normalization constant is given by the \emph{partition function}
\begin{align*}
    Z_{\bm{G}}(\beta) \defeq \E_{\sigma \sim \varrho}\wrapb{\exp\wrapp{\beta \cdot \mathcal{H}_{\bm{G}}(\sigma)}}.
\end{align*}
In this paper, we will consider the cases where $\varrho$ is the uniform measure over either the discrete Boolean hypercube $\mathcal{C}_{n} \defeq \{\pm1\}^{n}$ or the (rescaled) Euclidean sphere $\mathcal{S}_{n} \defeq \wrapc{\sigma \in \R^{n} : \norm{\sigma}_{2}^{2} = n}$. Note that by setting $\varrho = \mathsf{Unif}(\mathcal{C}_{n})$, $p_{\max} = 2$ and $\upgamma_2=1/\sqrt{2}$, the above recovers perhaps the most intensely studied mean-field spin glass: the \emph{Sherrington--Kirkpatrick model} \cite{SK75}.

Given these data, two natural and intimately related computational tasks are as follows:
\begin{itemize}[label={}]
    \item \textbf{Counting:} For an error tolerance $\eta > 0$ and a failure probability $0 < \delta < 1$, compute an \emph{estimate} $\widehat{Z}$ satisfying $e^{-\eta} \cdot Z_{\bm{G}}(\beta) \leq \widehat{Z} \leq e^{\eta} \cdot Z_{\bm{G}}(\beta)$ with probability at least $1 - \delta$.
    \item \textbf{Sampling:} For $\eta > 0$, sample a random configuration $\widehat{\sigma}$ such that $\diverge\wrapp{\Law\wrapp{\widehat{\sigma}}, \mu_{\bm{G},\beta}} \leq \eta$ for some prescribed (pseudo-)metric $\diverge(\cdot,\cdot)$ on the space of probability measures.
\end{itemize}
These are fundamental and important algorithmic primitives in high-dimensional statistics and computational physics. Over the last few decades, a deep and beautiful connection has emerged between the computational complexity of these tasks for a multitude of statistical physics models, and the presence/absence of structural phase transitions in these models; see e.g. \cite{Bar16book, PR22} and references therein. In this vein, our contributions are two-fold:
\begin{enumerate}
    \item Inspired by the seminal Lee--Yang theory of phase transitions \cite{LY52} (see \cref{subsec:prior} for further background), we control the locations of the \emph{zeros} of $Z_{\bm{G}}(\beta)$ in the complex plane. To the best of our knowledge, this perspective on mean-field spin glasses is novel. 

    \item Leveraging our understanding of the zeros, we design deterministic quasipolynomial-time algorithms which \emph{provably compute} the partition function $Z_{\bm{G}}(\beta)$ to high accuracy for nearly all $\beta$ below a certain natural threshold $\betasecond$ (see \cref{def:2nd-moment-regime}). For the Sherrington--Kirkpatrick model, $\betasecond$ is known to coincide with the \emph{replica-symmetry breaking threshold} $\betaRS$, a well-known phase transition also believed to separate the computationally tractable and intractable regimes of $\beta$ for counting and sampling. As we discuss in \cref{subsec:compare-thresholds,subsec:prior}, $\betasecond$ far exceeds the current range of $\beta$ in which efficient counting/sampling algorithms are known.

    At a high level, our algorithms are based on Barvinok's interpolation method, which additively approximates $\log Z_{\bm{G}}(\beta)$ using a truncation of its Taylor series. To apply this to the average-case setting (e.g. models with disorder), we construct a simple and general family of smooth interpolating curves which may be of independent interest.
\end{enumerate}
One advantage of our approach is its conceptual simplicity, as it relies only on the second moment method and elementary complex analysis. Moreover, our techniques can simultaneously treat the cubical and spherical cases in a completely unified manner. 

\subsection{Main Results}
We begin by stating the range of inverse temperatures $\beta$ in which we obtain efficient algorithms for estimating $Z_{\bm{G}}(\beta)$. For this, as is customary in the theory of mean-field spin glasses, we first define the \emph{mixture function} associated with the parameters $\{\upgamma_{p}\}_{p=2}^{p_{\max}}$ by
\begin{align*}
    \xi(s) \defeq \sum_{p=2}^{p_{\max}} \upgamma_{p}^{2} \cdot s^{p}.
\end{align*}
This is relevant because under the choice of normalization in \cref{eq:pspin-hamiltonian}, the collection of random variables $\wrapc{\mathcal{H}_{\bm{G}}(\sigma)}_{\sigma \in \R^{n}}$ appearing in the Gibbs measure and the partition function is a Gaussian process with covariance
\begin{align}\label{eq:hamiltonian-covariance}
    \E_{\bm{G}}\wrapb{\mathcal{H}_{\bm{G}}(\tau) \cdot \mathcal{H}_{\bm{G}}(\sigma)} = n \cdot \xi\wrapp{\frac{\langle \tau,\sigma \rangle}{n}},
\end{align}
which depends only on the \emph{overlap} $\frac{\langle \tau,\sigma \rangle}{n}$.
\begin{definition}[Second Moment Regime]\label{def:2nd-moment-regime}
Define the following (context-dependent) entropy functional on $[-1,1]$:
\begin{align*}
    h(m) \defeq \begin{dcases}
        -\frac{1}{2} \cdot \wrapp{(1-m)\log(1-m) + (1+m)\log(1+m)}, &\quad\text{when } \varrho = \mathsf{Unif}(\mathcal{C}_{n}) \quad \text{ (Hypercube)} \\[5pt]
        \frac{1}{2} \log \wrapp{1 - m^{2}}, &\quad\text{when } \varrho = \mathsf{Unif}(\mathcal{S}_{n}) \quad \text{ (Sphere)}
    \end{dcases}
\end{align*}
For a mixture function $\xi$, we say $\beta \geq 0$ satisfies the \emph{second moment condition (with respect to $\varrho,\xi$)} if the function $\varphi : [-1,1] \to \R$ defined by
\begin{align*}
    \varphi(m) \defeq \beta^{2} \cdot \xi(m) + h(m)
\end{align*}
has $0$ as its unique global maximizer, and that $\varphi''(0) < 0$. We define the corresponding \emph{second moment threshold} by the quantity
\begin{align*}
    \betasecond = \betasecond(\varrho,\xi) &\defeq \sup\wrapc{\beta \geq 0 : \beta \text{ satisfies the second moment condition w.r.t. } \varrho,\xi} \\
    &\hspace{2pt} = \hspace{2pt} \sup\wrapc{\beta \geq 0 : \varphi''(0) < 0 \text{ and } \varphi(m) < 0, \, \forall m \in [-1,1] \setminus \{0\}}.
\end{align*}
We suppress the dependence on $\varrho,\xi$ when they are clear from context.
\end{definition}

This is the supremal (positive real) $\beta$ where the second moment method locates the in-probability limiting free energy $\frac{1}{n} \log Z_{\bm{G}}$.
In \cref{sec:techniques} we explain how the second moment method is related to our techniques; see also \cref{lem:2nd-moment-interpret}.
In \cref{subsec:compare-thresholds}, we compare the second moment threshold with other notable phase transition thresholds appearing in the literature.

\begin{theorem}[Algorithms]\label{thm:alg-pspin-glass}
For any $0 < \varepsilon,\eta < 1$, there exists $C = \Theta(1/\varepsilon)$ and $\delta = \delta(\varepsilon,n) > 0$ tending to $0$ as $n \to \infty$ such that
\begin{align*}
    \mathbb{P}_{\bm{G}}\wrapb{e^{-\eta} \cdot \exp\wrapp{\widehat{P}(\beta)} \leq Z_{\bm{G}}(\beta) \leq e^{\eta} \cdot \exp\wrapp{\widehat{P}(\beta)}} \geq 1 - \delta, \qquad \forall \beta \in \D(0,(1 - \varepsilon) \cdot \betasecond),
\end{align*}
where $\widehat{P}(\beta)$ denotes the Taylor polynomial of $\beta \mapsto \log Z_{\bm{G}}(\beta)$ of degree $C\log(n/\eta)$. In particular, there is a deterministic algorithm running in time $n^{C \log(n/\eta)}$ which, on input $\bm{G},\{\upgamma_{p}\}_{p=2}^{p_{\max}},\beta$, outputs an $e^{\pm\eta}$-multiplicative approximation to $Z_{\bm{G}}(\beta)$ with probability at least $1 - o_{n}(1)$, for all $\beta \in \D(0,(1-\varepsilon) \cdot \betasecond)$.
\end{theorem}
We highlight that as a special case, we obtain deterministic quasipolynomial-time algorithms for estimating the partition function of the Sherrington--Kirkpatrick model, where $\varrho = \mathsf{Unif}(\mathcal{C}_{n})$ and $\xi(s) = \frac{1}{2}s^{2}$, for all inverse temperatures up to the \emph{replica-symmetry breaking threshold} $\betaRS = \betasecond = 1$. This threshold is known to separate the high-temperature and low-temperature regimes in these models; see \cref{subsec:compare-thresholds} for a more in-depth discussion. %A couple remarks are in order.

Following Barvinok's interpolation method \cite{Bar15, Bar16, Bar16book, Bar17perm}, our algorithms are based on truncating the Taylor series of $\log Z_{\bm{G}}(\beta)$ to logarithmically many terms. The low-order coefficients are then computed exactly via brute force. The crucial property which permits this truncation is \emph{stability} (or \emph{zero-freeness}) of the partition function $Z_{\bm{G}}(\beta)$. This is our main technical contribution.

\begin{theorem}[Zeros]\label{thm:zeros-pspin-glass}
For any $0 < \varepsilon < 1$, there exists $\delta = \delta(\varepsilon,n) > 0$ tending to $0$ as $n \to \infty$ such that with probability at least $1 - \delta$ over the randomness of $\bm{G}$, the partition function $Z_{\bm{G}}(\beta)$ is nonzero on the disk $\D(0,(1 - \varepsilon) \cdot \betasecond)$.
\end{theorem}

Taking inspiration from \cite{EM18}, we prove these using \emph{Jensen's Formula}, a widely used tool in the analysis of random polynomials. Our approach continues a long line of research studying the connections between the stability of partition functions, phase transitions in large-scale systems in statistical physics, and computational complexity of counting and sampling. In our setting, since we are viewing $Z_{\bm{G}}(\beta)$ as a function of the interaction strength $\beta$, \cref{thm:zeros-pspin-glass} says that with high probability there are no \emph{Fisher zeros} \cite{Fis65} of the model inside the centered disk of radius equal to the second moment threshold. Constraining the locations of the \emph{Lee--Yang zeros} \cite{LY52}, i.e. the zeros of the partition function where the input variables are given by the \emph{external fields}, remains an interesting open problem. This is especially relevant for \emph{sampling} from the associated Gibbs measure, since one needs to be able to handle the presence of external fields in order to invoke standard reductions between counting and sampling \cite{JVV86}. We emphasize that due to the lack of self-reducibility, our work does not immediately yield an efficient algorithm which samples from the Gibbs distribution. Even for the Sherrington--Kirkpatrick model, this remains an outstanding open problem as of this writing, although for general mixed $p$-spin models, efficient sampling is expected to not always be possible up to $\betasecond$; see \cref{subsec:compare-thresholds}.

\subsection{Comparison with Other Thresholds}\label{subsec:compare-thresholds}
Using the powerful but nonrigorous ``replica method'', physicists identified a natural phase transition threshold $\betaRS$ such that in the \emph{replica symmetric (RS) phase} $0 \leq \beta < \betaRS$, the model exhibits ``high-temperature behavior'' (see e.g. \cite{SK75, TAP77, dAT78}). For the models we study in this paper, an equivalent description of this phase is that the ``quenched free energy'' $\frac{1}{n} \E_{\bm{G}}\wrapb{\log Z_{\bm{G}}(\beta)}$ and the ``annealed free energy'' $\frac{1}{n} \log \E_{\bm{G}}\wrapb{Z_{\bm{G}}(\beta)}$ agree in the large $n$ limit. Note that the former quantity is of greater interest from the physics perspective, while the latter quantity admits a simple and explicit formula. This has since been vindicated in a number of works in the mathematical physics literature (see e.g. \cite{ALR87, CN95}). 
The RS phase coincides with a natural range of $\beta$ for which one can hope to obtain high-accuracy estimates of the partition function $Z_{\bm{G}}(\beta)$ in polynomial time. This is the regime in which we operate. In the specific setting of the Sherrington--Kirkpatrick model, where $\xi(s) = \frac{1}{2}s^{2}$ and $\varrho = \mathsf{Unif}(\mathcal{C}_{n})$, it is well-known that $\betaRS = \betasecond = 1$ \cite{ALR87, CN95}. Thus, for this model, our algorithmic result is essentially sharp in terms of the range of $\beta$.

For the remainder of this discussion, we restrict our attention to the spherical case; a similar picture is predicted to hold for the discrete hypercube. In this setting,
\begin{align}\label{eq:RS-threshold}
    \betaRS \defeq \sup\wrapc{\beta \geq 0 : \beta^{2} \cdot \xi(m) + m + \log(1 - m) < 0, \, \forall m \in (0,1)}.
\end{align}
As described in \cite{HMP24, HMRW25}, the replica-symmetric phase can be further broken up into several, typically distinct, thresholds
\begin{align*}
    0 \leq \betaunique \leq \betaSL \leq \betashatter \leq \betaRS,
\end{align*}
which are particularly relevant for the related task of \emph{sampling} from the Gibbs measure $\mu_{\bm{G},\beta}$:
\begin{itemize}[label={}]
    \item \textbf{Uniqueness Threshold:} This is the threshold $\betaunique$ below which local Markov chains (e.g. \emph{Glauber dynamics} and \emph{Langevin diffusion}) for sampling from $\mu_{\bm{G},\beta}$ are believed to satisfy a \emph{Poincar\'{e} Inequality} with high probability (and hence, converge to $\mu_{\bm{G},\beta}$ in polynomial time even under a worst-case initialization \cite{CHS93}). Above this threshold, it is known that worst-case mixing is exponentially slow \cite{AJ24}. For a precise definition of this threshold in the spherical pure $p$-spin case, see \cite{BBM96}.
    \item \textbf{Stochastic Localization Threshold:} This is the threshold $\betaSL$ below which \emph{algorithmic stochastic localization} \cite{AMS22, HMP24} successfully samples from $\mu_{\bm{G},\beta}$, and above which algorithmic stochastic localization fails \cite{HMP24}.
    This is also the threshold below which \emph{simulated annealing} \cite{HMRW25} is proven to sample from $\mu_{\bm{G},\beta}$ (though simulated annealing may also succeed beyond $\betaSL$).
%    It is given by the expression
%    \begin{align*}
%        \betaSL \defeq \sup\wrapc{\beta \geq 0 : \beta^{2} \cdot \xi''(m) < \frac{1}{(1 - m)^{2}}, \, \forall m \in [0,1)}.
%    \end{align*} % NOTE: readers don't really need to know these definitions; refer to previous work
    \item \textbf{Shattering Threshold:} Below this threshold $\betashatter$ and above $\betaunique$, \emph{metastable states} are present in the Gibbs measure \cite{AJ24}, which obstruct local Markov chains from mixing rapidly from a worst-case initial distribution. Nonetheless, because these metastable states in aggregate only constitute a vanishingly small proportion of the Gibbs mass, it is has been predicted in the physics literature that local Markov chains still mix rapidly from a \emph{random initialization}; see e.g. \cite{CK93, CHS93, Ala25}. Mixing from a suitable \emph{warm start} has been proved for $\beta$ up to $\betaSL$ \cite{HMRW25}.
    Rigorous analysis of any algorithm that succeeds up to $\betashatter$ (which can be Glauber or Langevin dynamics from random initialization, or another algorithm entirely) remains an interesting open problem.
%    This threshold is given by the expression
%    \begin{align*}
%        \betashatter \defeq \sup\wrapc{\beta \geq 0 : \beta^{2} \cdot \xi'(m) < \frac{m}{1-m}, \, \forall m \in (0,1)}.
%    \end{align*} % NOTE: readers don't really need to know these definitions; refer to previous work

    Above this threshold $\betashatter$, $\mu_{\bm{G},\beta}$ \emph{shatters} into exponentially many well-separated clusters, each of which contains only an exponentially small fraction of the Gibbs mass \cite{CHS93, GJK23, AMS25}. There is preliminary evidence suggesting that efficient sampling from $\mu_{\bm{G},\beta}$ is impossible for all $\beta > \betashatter$ \cite{AMS22, AMS25}.
\end{itemize}
We refer interested readers to \cite{HMP24, HMRW25} for precise characterizations of $\betaSL,\betashatter$, analogous to \cref{def:2nd-moment-regime} and \cref{eq:RS-threshold}.

Because sampling from $\mu_{\bm{G},\beta}$ is intimately related to computing the partition function $Z_{\bm{G}}(\beta)$ \cite{JVV86}, it is interesting to compare these thresholds with $\betasecond$. For this, let us focus on the \emph{pure $p$-spin} case, where $\xi(s) = s^{p}$ for some positive integer $p \geq 2$. In this case, it is known (see \cite{HMP24, HMRW25}) that
\begin{align*}
    \betaunique \asymp \frac{1}{\sqrt{\log p}}, \qquad \betaSL, \betashatter \asymp 1 \text{ with } \frac{\betaSL}{\betashatter} \approx \frac{\sqrt{e}}{2}, \qquad \betaRS \asymp \sqrt{\log p}.
\end{align*}
Now recall that $h(m) = \frac{1}{2} \log\wrapp{1 - m^{2}}$ in the spherical case. Since the following numerical inequalities hold
\begin{align*}
    2 \cdot h(m) \leq m + \log(1-m) \leq h(m), \qquad \forall m \in [0,1],
\end{align*}
it is clear that $\betasecond \asymp \betaRS \asymp \sqrt{\log p}$ as well. Thus, \cref{thm:alg-pspin-glass} covers a significantly larger range than $0 \leq \beta \leq \betashatter$, which is the limit of what is believed to be achievable in the realm of efficient sampling \cite{AMS22, AMS25}. Note, however, that $\betasecond < \betaRS$ in general. We leave extending \cref{thm:alg-pspin-glass} to the full replica-symmetric phase to future work; see \cref{sec:conclusion} for further discussion.

\subsection{Prior Work}\label{subsec:prior}

\paragraph{Free Energy Formulas} There is an enormous literature devoted to deriving and proving explicit formulas for the limiting ``quenched free energy'':
\begin{align*}
    \lim_{n \to \infty} \frac{1}{n} \E_{\bm{G}}\wrapb{\log Z_{\bm{G}}(\beta)}.
\end{align*}
For the Sherrington--Kirkpatrick model, landmark results include the \emph{replica-symmetric formula} \cite{SK75}, the celebrated \emph{Parisi formula} \cite{Par79, Par80, Par83}, and subsequent breakthroughs which made the latter mathematically rigorous \cite{GT02, Gue03, Tal06a, Tal06b, Che13, Pan13a, Pan14}; we refer interested readers to \cite{Bov06, Tal11a, Tal11b, Pan13b} and references therein for more details. Note that such formulas entail only an $o(n)$-additive approximation to $\E_{\bm{G}}\wrapb{\log Z_{\bm{G}}(\beta)}$, which is of order $O(n)$. In contrast, following the theoretical computer science tradition, our goal is to design algorithms which estimate $\log Z_{\bm{G}}$ up to $\eta$-additive error for \emph{any} $\eta > 0$ prescribed as input. Moreover, our algorithms must work with high probability over the realization of $\bm{G}$, as opposed to in expectation.

\paragraph{The Plefka and Aizenman--Lebowitz--Ruelle Expansions} In \cite{Ple82}, an alternative but nonrigorous derivation of the \emph{Thouless--Anderson--Palmer (TAP) Equations} \cite{TAP77} for the Sherrington--Kirkpatrick model was proposed, based on truncating the Taylor expansion of $\beta \mapsto \log Z_{\bm{G}}(\beta)$ to second-order:
\begin{align*}
    \log Z_{\bm{G}}(\beta) ``\approx" \log Z_{\bm{G}}(0) + \E_{\sigma \sim \varrho}\wrapb{\frac{1}{\sqrt{2n}} \sigma^{\top}\bm{G}\sigma} \cdot \beta + \Var_{\sigma \sim \varrho}\wrapb{\frac{1}{\sqrt{2n}} \sigma^{\top}\bm{G}\sigma} \cdot \frac{\beta^{2}}{2}.
\end{align*}
This is known in the literature as the \emph{Plefka expansion}. However, we are unaware of any mathematically rigorous result prior to our work which validates the above approximation. \cref{thm:zeros-pspin-glass} vindicates the $O(\log n)$-order Plefka expansion with probability $1 - o_{n}(1)$ for all $\beta$ with $\abs{\beta} < \betasecond$.

Aizenman, Lebowitz and Ruelle~\cite{ALR87} introduced a cluster expansion approximation of the Sherrington--Kirkpatrick partition function and rigorously showed that it $o_n(1)$-additively approximates $\log Z_{\bm{G}}(\beta)$ for all (positive real) $\beta < \betasecond = 1$. An analogue of this expansion for the mixed $p$-spin models will play an important role in our proof; see \cref{subsec:overview-reweighting}. However, we emphasize that the cluster expansion of \cite{ALR87} only achieves an $o_n(1)$ error for \emph{some} $o_n(1)$, whereas \cref{thm:alg-pspin-glass} achieves any \emph{arbitrarily small} prescribed error.

\paragraph{Localization and Markov Chains} For the Sherrington--Kirkpatrick model, \cite{EKZ22} (see also \cite{BB19}) showed that Glauber dynamics for sampling from $\mu_{\bm{G},\beta}$ mixes in $O(n^{2})$-time for $\beta < \frac{1}{4}$. Subsequently, the mixing time was improved to $O(n\log n)$ in \cite{AJKPV22}, and the range of $\beta$ was extended to $\beta < 0.295$ \cite{AKV24}. In \cite{AMS22} (see also \cite{AMS23a} for the mixed $p$-spin setting), it was shown that an algorithmic version of Eldan's \emph{stochastic localization} \cite{Eld13} samples from the Gibbs measure up to $o(n)$ Wasserstein distance for all $\beta < 1 = \betaRS$, which is the tight range of $\beta$; it is believed that Glauber dynamics efficiently samples up to arbitrary total variation error in this regime as well. For more general mixed $p$-spin models on the hypercube, $O(n^{2})$-mixing for Glauber dynamics was established in \cite{ABXY24} under a more restrictive condition than the replica-symmetric phase. This was improved to $O(n\log n)$ in \cite{AJKPV24univ}. Besides \cite{AMS22, AMS23a}, all these sampling results can be converted into efficient algorithms for computing $(1 \pm \eta)$-multiplicative approximations to $Z_{\bm{G}}(\beta)$ via standard counting-to-sampling reductions \cite{JVV86}.

In the spherical setting, \cite{HMP24} used algorithmic stochastic localization to sample from the Gibbs measure with $o_{n}(1)$ total variation error for $\beta < \betaSL$. \cite{HMRW25} then showed that annealed Langevin dynamics samples from the Gibbs measure up to subexponentially small total variation error in the same regime. While the lack of self-reducibility obstructs standard sampling-to-counting reductions, we note that sampling up to arbitrarily small total variation error does yield an algorithm for estimating the partition function $Z_{\bm{G}}(\beta)$ in a black-box manner. Indeed, we have
\begin{align*}
    \wrapp{\log Z_{\bm{G}}}'(\beta) = \E_{\sigma \sim \mu_{\bm{G},\beta}}\wrapb{\mathcal{H}_{\bm{G}}(\sigma)},
\end{align*}
which can be approximated using the Monte Carlo method given sampling access to the Gibbs measure $\mu_{\bm{G},\beta}$. This can then be integrated over $\beta$.

\paragraph{Approximate Counting from Zero-Freeness} Since the introduction of the $\#\P$ complexity class in theoretical computer science \cite{Val79perm, Val79}, there has been significant interest in designing provable $\FPTAS$ for \emph{approximate counting} tasks, that is, efficient deterministic algorithms which estimate the partition function up to arbitrarily small relative error. Towards this, Barvinok proposed to truncate the Taylor series of the logarithm of the partition function. This is often referred to as Barvinok's \emph{interpolation method}, and it has since been used to design quasipolynomial-time approximation approximation algorithms for a wide variety of counting problems, including the permanent \cite{Bar16, Bar17perm}, graph homomorphisms \cite{Bar15, BS16, BS17hom}, etc. We highlight \cite{EM18}, which applies the interpolation method to compute permanents of \emph{random} complex matrices. Our work draws heavily upon theirs. For some of these counting problems these algorithms have been improved to run in polynomial time~\cite{PR17}, thereby obtaining a genuine $\FPTAS$.

Now, the radius of convergence is governed by the location of the nearest zero of the partition function. Building on the seminal Lee--Yang theory \cite{LY52}, it was gradually realized that there is a deep connection between the absence of phase transitions and the computational complexity of estimating the partition function. This has been formalized in a number of works; see e.g. \cite{PR19, BGGS20, BCSV23, deBoeratal24} for the hardcore gas model, \cite{LSS19, LSS19b, LSS20, GLL20, BDPR21} for the Ising and Potts models, \cite{BCR21, GLLZ21} for Holant-type problems, \cite{HMS20} for quantum many-body systems, etc. We refer interested readers to \cite{Bar16book, PR22} for a more comprehensive discussion.

\subsection{Acknowledgements}
K.L. is grateful to Sidhanth Mohanty, Francisco Pernice, and Amit Rajaraman for enlightening conversations surrounding spin glasses. B.H. is grateful to Byron Chin and Mark Sellke for motivating conversations. We would also like to thank the Centrum Wiskunde \& Informatica (CWI) Amsterdam for hosting a workshop organized by F.B. and G.R. where this collaboration began.

\section{Techniques and Overview of the Paper}\label{sec:techniques}

\subsection{Counting Zeros via Jensen's Formula}

\cref{thm:alg-pspin-glass} will follow easily from the zero-freeness established in \cref{thm:zeros-pspin-glass}. Indeed, following Barvinok (see e.g. \cite{Bar16book}), the algorithm just needs to compute the Taylor coefficients of $\widehat{P}(\beta)$; we provide a formal proof in \cref{sec:algos}.
In the rest of this overview, we discuss the proof of \cref{thm:zeros-pspin-glass}, which occupies most of this paper.
We will set out to prove
\begin{align}\label{eq:overview-markov-goal}
    \E_{\bm{G}}\wrapb{\#\{\omega \in \D(0,r) : Z_{\bm{G}}(\omega) = 0\}} = o_n(1),
\end{align}
from which \cref{thm:zeros-pspin-glass} follows by Markov's inequality.
At the heart of our methodology is the following tool from complex analysis and the study of zeros of random polynomials, which was used in a similar context in \cite{EM18}.
\begin{theorem}[Jensen's Formula]\label{thm:jensen-zeros}
Let $\Omega\subset\C$ be an open set that contains the disk $\overline{\D(0,R)}$ for some $R>0$. Let $f : \Omega \to \C$ be an analytic function such that $f(0)\neq 0$. Then
\begin{align*}
    \sum_{\omega \in \D(0,R) : f(\omega) = 0} \log\wrapp{\frac{R}{\abs{\omega}}} = \E_{\theta}\wrapb{\log \abs{\frac{f(Re^{i\theta})}{f(0)}}},
\end{align*}
where the expectation is taken with respect to $\theta \sim \mathsf{Unif}[-\pi,\pi]$.
\end{theorem}
Throughout this paper we set
\begin{align*}
    R = \wrapp{1 - \frac{\varepsilon}{2}} \betasecond \qquad \text{and} \qquad
    r = \wrapp{1 - \varepsilon} \betasecond.
\end{align*}
The starting point of our proof is the following estimate on the expected number of zeros of a random analytic function on $\D(0,r)$, in terms of the function's second moment.
This is proved in \cref{subsec:jensen-formula-second-moment}.
\begin{lemma}\label{lem:jensen-formula-second-moment}
Let $\Omega$ be as in \cref{thm:jensen-zeros}, and $f_{\bm{G}} : \Omega \to \C$ be a random $\bm{G}$-measurable function, which is analytic  with $f_{\bm{G}}(0) \neq 0$ deterministically. Then,
\begin{align*}
    \E_{\bm{G}}\wrapb{\#\{\omega \in \D(0,r) : f_{\bm{G}}(\omega) = 0\}}
    \le \frac{1}{2\varepsilon} \E_{\theta} \log \E_{\bm{G}} \wrapb{
        \abs{\frac{f_{\bm{G}}(Re^{i\theta})}{f_{\bm{G}}(0)}}^2
    }
\end{align*}
\end{lemma}
The most natural choice of $f_{\bm{G}}(\beta)$ to use in \cref{lem:jensen-formula-second-moment} is $f_{\bm{G}}(\beta) = Z_{\bm{G}}(\beta)$.
With this choice, we will see that \cref{lem:jensen-formula-second-moment} reduces the mixed $p$-spin glass to a significantly simpler mean-field model where the interactions are deterministic, the \emph{mixed $p$-spin Curie--Weiss model}.
The partition function of this model is

\begin{align}\label{eq:pspin-curie-weiss}
    Z_{\CW}(\beta) \defeq \E_{\sigma \sim \varrho}\wrapb{\exp\wrapp{\beta \cdot n \cdot \xi\wrapp{\frac{\langle \sigma, \allone \rangle}{n}}}},
\end{align}
where again, we suppress the dependence on $\xi$. In the corresponding Gibbs distribution, higher probability mass is placed on configurations $\sigma$ which align with $\allone$ or $-\allone$, i.e. all interactions are ferromagnetic in nature. In the case where $\varrho = \mathsf{Unif}(\mathcal{C}_{n})$ and $\xi(s) = \frac{1}{2}s^{2}$, this is the partition function of the standard Curie--Weiss model in statistical physics, which may be viewed as the ferromagnetic Ising model on the complete graph.

Formally, the following lemma, proved in \cref{subsec:jensen-formula-second-moment}, evaluates the right-hand side of \cref{lem:jensen-formula-second-moment} for $f_{\bm{G}}(\beta) = Z_{\bm{G}}(\beta)$.
\begin{lemma}\label{lem:no-reweight-reduction-to-curie-weiss}
We have that
\begin{align*}
    \E_{\theta}\log \E_{\bm{G}}\wrapb{\abs{\frac{Z_{\bm{G}}(Re^{i\theta})}{Z_{\bm{G}}(0)}}^2}
    = \log Z_{\CW}\wrapp{R^2}.
\end{align*}
\end{lemma}
Unfortunately, this is not enough to prove \cref{thm:zeros-pspin-glass}. In the second moment regime we study,  $Z_{\CW}$ is bounded by a constant, but is usually not $1+o_n(1)$. Indeed, for the Sherrington--Kirkpatrick model (i.e. the mixed $p$-spin glass with $\varrho = \mathsf{Unif}(\mathcal{C}_{n})$ and $\xi(s) = \frac{1}{2}s^{2}$), the corresponding $Z_{\CW}$ is the standard Curie--Weiss partition function. It is well-known that $Z_{\CW}$ has a phase transition at $\betasecond = 1$, and $Z_{\CW}\wrapp{\beta^{2}} = (1 - \beta^{2})^{-1/2} + o_n(1)$ for positive real $\beta < 1$; see e.g. Lemma 2.1 in \cite{Tal98}, or \cref{lem:CW-bound} below. For general mixture functions, we show in \cref{prop:curie-weiss-bound} that
\begin{align}\label{eq:CW-partition-fn-bound}
    Z_{\CW}\wrapp{|\beta|^2}
    = \frac{1}{\sqrt{1-|\beta|^2 \xi''(0)}} + o_n(1)
    = \frac{1}{\sqrt{1-2|\beta|^2\upgamma_{2}^{2}}} + o_n(1)
\end{align}
for all $|\beta| < \betasecond$.
This generalizes earlier results for the pure $p$-spin case \cite{MSB21}.
Combining \cref{eq:CW-partition-fn-bound} with \cref{lem:jensen-formula-second-moment,lem:no-reweight-reduction-to-curie-weiss}, where we choose $f_{\bm{G}}(\beta) = Z_{\bm{G}}(\beta)$, shows 
\begin{align*}
    \E_{\bm{G}}\wrapb{\#\{\omega \in \D(0,r) : Z_{\bm{G}}(\omega) = 0\}}
    \le O_\varepsilon(1),
\end{align*}
where $O_\varepsilon(1)$ denotes an $\varepsilon$-dependent constant.
This falls short of the $o_n(1)$ desired in \cref{eq:overview-markov-goal}.
\begin{remark}
An exception to the above discussion is the case where quadratic interactions are absent, i.e. $\upgamma_2 = 0$. In this case \cref{eq:CW-partition-fn-bound} shows $Z_{\CW}\wrapp{|\beta|^2} = 1 + o_n(1)$. Then, \cref{lem:jensen-formula-second-moment,lem:no-reweight-reduction-to-curie-weiss} directly imply \cref{eq:overview-markov-goal}, which yields \cref{thm:zeros-pspin-glass}.
\end{remark}
We will be able to obtain the $o_n(1)$ bound in \cref{eq:overview-markov-goal} by applying \cref{lem:jensen-formula-second-moment} with $f_{\bm{G}}(\beta) = X_{\bm{G}}(\beta)$, where $X_{\bm{G}}(\beta)$ is a suitably reweighted version of $Z_{\bm{G}}(\beta)$.
We explain this in the next subsection.

\subsection{Reweighting the Partition Function\label{subsec:overview-reweighting}}

The fact that the bound in \cref{lem:jensen-formula-second-moment} with $f_{\bm{G}}(\beta) = Z_{\bm{G}}(\beta)$ loses a constant factor is related to a common issue faced by the \emph{second moment method}. We describe this first in the simpler setting of a positive real-valued random variable $Z$. The second moment method is a strategy for proving that $Z>0$ with high probability, in which we aim to show that
\begin{align}\label{eq:second-moment-method-goal}
    \E\wrapb{Z^{2}} = (1+o_n(1)) \cdot \E[Z]^2.
\end{align}
If this holds, the Cauchy--Schwarz inequality implies $\mathbb{P}(Z>0) = 1-o_n(1)$. However, in many applications \cref{eq:second-moment-method-goal} does not hold, and instead we only have the weaker estimate
\begin{align}\label{eq:second-moment-method-weaker}
    \E\wrapb{Z^{2}} = O(1) \cdot \E[Z]^2.
\end{align}
Typically, this is because the random variable $Z$ has $O(1)$-multiplicative fluctuations. From \cref{eq:second-moment-method-weaker} we can only conclude $\mathbb{P}(Z>0) = \Omega(1)$. In such settings, one way to recover a high-probability bound is to identify a random variable $A>0$ such that
\begin{align*}
    X \defeq A \cdot Z = 1+o_n(1) \qquad \text{with high probability},
\end{align*}
and then aim to show \cref{eq:second-moment-method-goal} with $X$ in place of $Z$.
That is, $A$ is a reweighting factor that ``explains the fluctuations'' of $Z$.
This method is well-known in the combinatorics literature as the \emph{small subgraph conditioning method} \cite{RW92,RW94}.

In our setting, the random variable $Z_{\bm{G}}(\beta)$ likewise has $O(1)$-multiplicative fluctuations, which cause a direct application of \cref{lem:jensen-formula-second-moment} to lose a constant factor. To remedy this, we will guess a reweighting factor $A_{\bm{G}}(\beta)$ such that we (heuristically) expect
\begin{align}\label{eq:reweighted-partition-function}
    X_{\bm{G}}(\beta) \defeq A_{\bm{G}}(\beta) \cdot Z_{\bm{G}}(\beta)
    = 1 + o_n(1),
    \qquad \forall \beta \in \overline{\D(0,R)}
\end{align}
with high probability.
Our choice of $A_{\bm{G}}(\beta)$ takes the form of a cluster expansion in the entries of the quadratic part of $\bm{G}$, and is motivated by the cluster expansion of \cite{ALR87}.
This follows the work \cite{ALS22}, which applies a similar cluster expansion-reweighted second moment to the (positive real-valued) partition function of the symmetric binary perceptron.

We will choose $A_{\bm{G}}(\beta)$ to be analytic in $\beta$ on an open neighborhood $\Omega \supseteq \overline{\D(0,R)}$, so that $X_{\bm{G}}$ is also analytic in $\Omega$ and its zeros are a superset of those of $Z_{\bm{G}}$.
We then show that for $f_{\bm{G}}(\beta) = X_{\bm{G}}(\beta)$, the upper bound obtained in \cref{lem:jensen-formula-second-moment} is $o_n(1)$.
(In order to make this calculation tractable, it will also be important that $A_{\bm{G}}(\beta)$ is complex analytic in $\bm{G}$; see \cref{rem:why-need-A-analytic}.)
The main point of the analysis is that the reweighting factor $A_{\bm{G}}(\beta)$ contributes a factor of $\sqrt{1 - |\beta|^2 \xi''(0)}$ to the second moment.
This cancels the Curie--Weiss partition function \cref{eq:CW-partition-fn-bound}, leaving
\begin{align*}
    \E_{\bm{G}}\wrapb{\abs{\frac{X_{\bm{G}}(\beta)}{X_{\bm{G}}(0)}}^2} = 1+o_n(1),
    \qquad \forall \beta \in \overline{\D(0,R)}.
\end{align*}
Then \cref{thm:zeros-pspin-glass} readily follows from \cref{lem:jensen-formula-second-moment}.
Finally, we note that while our choice of $A_{\bm{G}}(\beta)$ is motivated by spin glass considerations, for the formal proof it is enough to control the upper bound from \cref{lem:jensen-formula-second-moment}, and this is purely a second moment calculation.
In particular, we do not need to prove that $A_{\bm{G}}(\beta)$ satisfies \cref{eq:reweighted-partition-function}.

\subsection{Reduction from Zero-Counting to Second Moment}\label{subsec:jensen-formula-second-moment}

In this subsection, we prove \cref{lem:jensen-formula-second-moment,lem:no-reweight-reduction-to-curie-weiss}. \cref{lem:no-reweight-reduction-to-curie-weiss} will not play a formal role in the proofs of our main results. We include its proof to explain what \cref{lem:jensen-formula-second-moment} obtains with the natural choice $f_{\bm{G}}(\beta) = Z_{\bm{G}}(\beta)$, and to provide a simple demonstration of the second moment calculations that arise in our method.

\begin{proof}[Proof of \cref{lem:jensen-formula-second-moment}]
Note that (deterministically)
\begin{align*}
    \#\{\omega \in \D(0,r) : f_{\bm{G}}(\omega) = 0\}
    &\le \wrapp{\log \frac{R}{r}}^{-1}
    \sum_{\omega \in \D(0,r) : f_{\bm{G}}(\omega) = 0}
    \log \wrapp{\frac{R}{|\omega|}} \\
    &\le \frac{1}{\varepsilon}
    \sum_{\omega \in \D(0,R) : f_{\bm{G}}(\omega) = 0}
    \log \wrapp{\frac{R}{|\omega|}}
    =
    \frac{1}{\varepsilon}
    \E_{\theta}\wrapb{\log \abs{\frac{f_{\bm{G}}(Re^{i\theta})}{f_{\bm{G}}(0)}}},
\end{align*}
where the final equality is by \cref{thm:jensen-zeros}.
Taking expectation over $\bm{G}$ we obtain
\begin{align*}
    \E_{\bm{G}}\wrapb{\#\{\omega \in \D(0,r) : f_{\bm{G}}(\omega) = 0\}}
    \le \frac{1}{\varepsilon}
    \E_{\bm{G}}\E_{\theta}\wrapb{\log \abs{\frac{f_{\bm{G}}(Re^{i\theta})}{f_{\bm{G}}(0)}}}
    \le \frac{1}{2\varepsilon}
    \E_{\theta}\log \E_{\bm{G}}\wrapb{\abs{\frac{f_{\bm{G}}(Re^{i\theta})}{f_{\bm{G}}(0)}}^2},
\end{align*}
by Jensen's inequality.
\end{proof}

\begin{proof}[Proof of \cref{lem:no-reweight-reduction-to-curie-weiss}]
We first compute the inner expectation by leveraging the independence of the entries of $\bm{G}$. For notational convenience, we write $\beta = Re^{i\theta} \in \C$, which we treat as fixed for the moment. Observe that
\begin{align*}
    \E_{\bm{G}}\wrapb{\abs{\frac{Z_{\bm{G}}(\beta)}{Z_{\bm{G}}(0)}}^{2}} &= \E_{\tau,\sigma \sim \varrho}\wrapb{\E_{\bm{G}}\wrapb{\exp\wrapp{\beta \cdot \mathcal{H}_{\bm{G}}(\tau) + \overline{\beta} \cdot \mathcal{H}_{\bm{G}}(\sigma)}}} \\
    &= \E_{\tau,\sigma \sim \varrho}\wrapb{\exp\wrapp{\frac{1}{2} \cdot \E_{\bm{G}}\wrapb{\wrapp{\beta \cdot \mathcal{H}_{\bm{G}}(\tau) + \overline{\beta} \cdot \mathcal{H}_{\bm{G}}(\sigma)}^{2}}}} \tag{Moment Generating Function of a Gaussian} \\
    &= \E_{\tau,\sigma \sim \varrho}\wrapb{\exp\wrapp{\frac{n}{2} \cdot \wrapp{\beta^{2} + \overline{\beta}^{2}} \cdot \xi(1) + n \cdot \abs{\beta}^{2} \cdot \xi\wrapp{\frac{\langle \tau,\sigma \rangle}{n}}}} \tag{Using \cref{eq:hamiltonian-covariance}} \\
    &= \exp\wrapp{n \cdot \xi(1) \cdot \re \beta^{2}} \cdot Z_{\CW}\wrapp{\abs{\beta}^{2}}. \tag{Using symmetry of the hypercube/sphere}
\end{align*}
Since $\abs{\beta}^{2} = R^2$, plugging this back in yields the desired identity, up to an extra additive error term of
\begin{align*}
    n \cdot \xi(1) \cdot \E_{\theta} \re\wrapp{R^2e^{2i\theta}}.
\end{align*}
This is clearly zero, concluding the proof.
\end{proof}

\subsection{Preliminaries}
Throughout this paper, we use $c,C,K,L > 0$ to denote various constants that may change depending on the context; unless otherwise specified, these constants are universal, i.e. independent of all other parameters (especially $n$). All logarithms are base $e$. As a complex function, we take the branch of the logarithm that is real on the positive real line. We write $\D(z,r)$ for the open disk of radius-$r$ around $z$ in the complex plane $\C$; we abbreviate $\D$ for $\D(0,1)$.

We will also need a very crude upper bound on the magnitude of $Z_{\bm{G}}(\beta)$.
\begin{proposition}\label{prop:logZ-crude}
There exists a constant $C > 0$ depending only on $\xi$ such that
\begin{align*}
    \mathbb{P}_{\bm{G}}\wrapb{\abs{Z_{\bm{G}}(\beta)} \leq \exp\wrapp{C \cdot \max\wrapc{\abs{\beta},1} \cdot n}, \, \forall \beta \in \C} \geq 1 - o_{n}(1).
\end{align*}

\end{proposition}
This is well-established in the literature. For instance, one can leverage bounds on the \emph{tensor injective norm} of $\bm{G}$ (see e.g. \cite{TS14, NDT15, DM24}); in the case $p_{\max} = 2$, this just means that a random matrix with independent $O(1/n)$-variance centered Gaussian entries has $O(1)$-operator norm with high probability. An overpowered alternative would be to appeal to known free energy formulas (see the discussion in \cref{subsec:prior}).

\begin{lemma}\label{lem:beta-second-ub}
We have $\betasecond^2 \cdot \xi''(0) \le 1$.
\end{lemma}
\begin{proof}
In either case of the definition of $h(m)$ in \cref{def:2nd-moment-regime}, $h''(0) = -1$.
Let $\varphi$ be defined as in \cref{def:2nd-moment-regime} with $\beta = \betasecond$.
The definition of $\betasecond$ implies
\begin{align*}
    0 \ge \varphi''(0)
    = \betasecond^2 \cdot \xi''(0) - 1,
\end{align*}
from which the result follows.
\end{proof}

\section{Evaluation of Reweighted Second Moment}\label{sec:reweighted-second-moment}
In this section, we formally define the reweighting factor $A_{\bm{G}}(\beta)$. We will then prove \cref{prop:reweighted-second-moment} below, which controls the upper bound from \cref{lem:jensen-formula-second-moment} with $f_{\bm{G}}(\beta) = X_{\bm{G}}(\beta)$ defined in \cref{eq:reweighted-partition-function}, assuming as input \cref{prop:curie-weiss-bound,prop:planted-cluster-expansion-moment} which are proved in \cref{sec:curie-weiss-bound,sec:planted-cluster-expansion-moment}. \cref{thm:zeros-pspin-glass} is immediate from \cref{prop:reweighted-second-moment}, and its proof appears at the end of this section. We recall that we write \begin{align*}
    R = \wrapp{1 - \frac{\varepsilon}{2}} \betasecond
\end{align*}
throughout.
    \begin{proposition}\label{prop:reweighted-second-moment}
Let $X_{\bm{G}}(\beta) = A_{\bm{G}}(\beta) \cdot Z_{\bm{G}}(\beta)$, where $A_{\bm{G}}$ is defined in \cref{eq:def-A-bG} below. Then,
\begin{align*}
    \E_{\bm{G}}\wrapb{\abs{\frac{X_{\bm{G}}(\beta)}{X_{\bm{G}}(0)}}^2} \le 1+o_n(1),
    \qquad \forall \beta \in \overline{\D(0,R)},
\end{align*}
where the error term $o_n(1)$ is uniform over $\beta$ in this domain.
\end{proposition}

\subsection{Definition of Reweighting Factor}

We begin by defining $A_{\bm{G}}(\beta)$. Note that
\begin{align}
  \nonumber
  \nabla^2 \mathcal{H}_{\bm{G}}(\bm{0})
  &= \frac{\upgamma_2}{\sqrt{n}} (\bm{G}_{i,j} + \bm{G}_{j,i} \,:\, 1\le i,j\le n ) \\
  \label{eq:nabla2H-expansion}
  &= \frac{\xi''(0)^{1/2}}{\sqrt{2n}} (\bm{G}_{i,j} + \bm{G}_{j,i} \,:\, 1\le i,j\le n ).
\end{align}
We will write throughout this paper
\begin{align}
    \label{eq:def-bM}
    \bm{M} = \nabla^2 \mathcal{H}_{\bm{G}}(\bm{0}).
\end{align}
Let $\binom{[n]}{2}$ denote the set of unordered pairs $\{i,j\}$ with $1\le i<j \le n$. We will think of elements of $\binom{[n]}{2}$ as edges on the complete graph $K_n$. For $e = \{i,j\} \in \binom{[n]}{2}$, we write $\bm{M}_e = \bm{M}_{i,j}$. Let
\begin{align}
    \label{eq:def-zeta}
    \zeta = \zeta(\beta) = \beta \xi''(0)^{1/2}.
\end{align}
Note that for all $\beta \in \overline{\D(0,R)}$,
\begin{align}
    \label{eq:zeta-ub}
    |\zeta| \le R\xi''(0)^{1/2}
    = \wrapp{1 - \frac{\varepsilon}{2}} \betasecond \cdot \xi''(0)^{1/2}
    \le 1 - \frac{\varepsilon}{2},
\end{align}
by \cref{lem:beta-second-ub}. Let $K = \lfloor \log \log n \rfloor$. Then, define
\begin{align}
    \label{eq:uc}
    \UC(n) = \wrapc{
        \Gamma \subseteq \binom{[n]}{2} :
        \text{
          $\Gamma$ is the edge set of a vertex-disjoint union of cycles
          and $|\Gamma| \le K$
        }
    }.
\end{align}
For $\Gamma \in \UC(n)$, let $c(\Gamma)$ denote the number of connected components of $\Gamma$ viewed as a graph. Finally define the reweighting factor
\begin{align}
    \label{eq:def-A-bG}
    A_{\bm{G}}(\beta)
    &=
    (1-\zeta^2)^{-1/2} \exp\wrapp{
        -\frac{n\beta^2}{2} \xi(1) + \frac{n\zeta^2}{4}
        - \frac{\zeta^2}{2} - \frac{\zeta^4}{8}
    } \\
    \nonumber
    &\qquad \times
    \exp\wrapp{
        -\frac{\beta}{2} \tr (\bm{M})
        -\frac{\beta^2}{2}
        \sum_{1\le i<j\le n}
        \bm{M}_{i,j}^2
    }\wrapp{
        \sum_{\Gamma \in \UC(n)}
        (-1)^{c(\Gamma)}
        \prod_{e\in \Gamma} (\beta \bm{M}_e)
    }.
\end{align}
We explain the motivation for this definition in \cref{sec:reweight-factor-motivation}. In light of \cref{eq:zeta-ub}, it is possible to pick a branch of the square root such that $(1-\zeta^2)^{-1/2}$ is analytic for $\beta$ in an open neighborhood $\Omega \supseteq \overline{\D(0,R)}$. Thus $A_{\bm{G}}(\beta)$ is analytic in $\Omega$.

\subsection{Gaussian Change of Measure Formula}

In preparation for the proof of \cref{prop:reweighted-second-moment}, we introduce a Gaussian change of measure formula, which will be useful for evaluating the exponentially-tilted Gaussian moments that arise in this proof.
If the tilting vector $\bm{v}$ is real, this formula is standard, and holds more generally for any measurable $f$ with suitably bounded growth.
However, the case of complex $\bm{v}$ will require more care.
\begin{fact}\label{fac:planting}
Let $m\in \N$ and $\bm{g} \sim \mathcal{N}(0,I_m)$. Consider a function $f : \C^m \rightarrow \C$ of the form
\begin{align}\label{eq:def-f}
    f(\bm{g}) = \exp(\langle \bm{A} \bm{g}, \bm{g}\rangle  + \langle \bm{w}, \bm{g} \rangle) p(\bm{g}),
\end{align}
where $\bm{A} \in \R^{m\times m}$ is symmetric with $\bm{A} \prec \frac12 I_m$ in the positive semidefinite ordering, $\bm{w} \in \C^m$, $p$ is a multivariate polynomial (with possibly complex coefficients) , and $\bm{A}, \bm{w}, p$ are deterministic. Then, for any deterministic $\bm{v} \in \C^m$,
\[
    \E \wrapb{e^{\langle \bm{g}, \bm{v} \rangle} f(\bm{g})}
    = \E \wrapb{e^{\langle \bm{g}, \bm{v} \rangle}}
    \E \wrapb{f(\bm{g}+\bm{v})}.
\]
\end{fact}
\begin{proof}
We may work in a coordinate basis in which $\bm{A}$ is diagonal, and write $\bm{A} = \diag(a_1,\ldots,a_m)$. For $0\le t\le m$, let $\bm{v}_{\le t} = (\bm{v}_1,\ldots,\bm{v}_t,0,\ldots,0) \in \C^m$ denote the vector whose first $t$ entries match those of $\bm{v}$, and whose last $m-t$ entries are $0$. We will argue by induction over $t$ that
\begin{equation}
\label{eq:planting-induction}
\E \wrapb{e^{\langle \bm{g}, \bm{v}_{\le t} \rangle} f(\bm{g})}
= \E \wrapb{e^{\langle \bm{g}, \bm{v}_{\le t} \rangle}}
\E \wrapb{f(\bm{g}+\bm{v}_{\le t})}.
\end{equation}
The base case $t=0$ is trivial. Suppose we have proved \cref{eq:planting-induction} up to $t$, for some $t<m$. Let
\[
    f_t(x) = \E \wrapb{
      e^{\langle \bm{g}, \bm{v}_{\le t} \rangle} f(\bm{g}_1,\ldots,\bm{g}_{t},x,\bm{g}_{t+2},\ldots,\bm{g}_m)
    }.
\]
Then, for $\varphi(x) = \frac{1}{\sqrt{2\pi}} e^{-x^2/2}$ the standard Gaussian density,
\begin{align*}
    \E \wrapb{e^{\langle \bm{g}, \bm{v}_{\le t+1} \rangle} f(\bm{g})}
    &= \E \wrapb{e^{\bm{g}_{t+1}\bm{v}_{t+1}} f_t(\bm{g}_{t+1})}
    = \int_{\R} e^{\bm{v}_{t+1} x} f_t(x) \varphi(x)\,dx \\
    &= e^{\bm{v}_{t+1}^2/2} \int_{\R} f_t(x) \varphi(x-\bm{v}_{t+1})\,dx
    = \E \wrapb{e^{\bm{g}_{t+1}\bm{v}_{t+1}}}
    \int_{\R} f_t(x) \varphi(x-\bm{v}_{t+1})\,dx.
\end{align*}
Since $x \mapsto f_t(x) \varphi(x - \bm{v}_{t+1})$ is complex analytic, its integral around the contour
\[
    -T \rightarrow T \rightarrow T + i\imag(\bm{v}_{t+1}) \rightarrow -T + i\imag(\bm{v}_{t+1}) \rightarrow -T
\]
is $0$ for any $T$. Sending $T\to\infty$ then shows
\[
    \int_{\R} f_t(x) \varphi(x-\bm{v}_{t+1})\,dx
    = \int_{\R + i\imag(\bm{v}_{t+1})} f_t(x) \varphi(x-\bm{v}_{t+1})\,dx
    = \E \wrapb{f_t(\bm{g}_{t+1}+\bm{v}_{t+1})},
\]
where the assumed hypotheses on $f$ and $\bm{A}$ imply the contribution from the ends of the contour integral go to $0$ as $T\to\infty$.
Thus, for $\bm{e}_{t+1} \in \R^m$ the $(t+1)$-th basis vector,
\[
    \E \wrapb{e^{\langle \bm{g}, \bm{v}_{\le t+1} \rangle} f(\bm{g})}
    = \E \wrapb{e^{\bm{g}_{t+1}\bm{v}_{t+1}}}
    \E \wrapb{f_t(\bm{g}_{t+1}+\bm{v}_{t+1})}
    = \E \wrapb{e^{\bm{g}_{t+1}\bm{v}_{t+1}}}
    \E \wrapb{e^{\langle \bm{g}, \bm{v}_{\le t} \rangle} f(\bm{g} + \bm{e}_{t+1}\bm{v}_{t+1})}.
\]
The induction hypothesis that \cref{eq:planting-induction} holds for $t$ then implies \cref{eq:planting-induction} holds for $t+1$. This completes the induction. Finally, the conclusion \cref{eq:planting-induction} for $t=m$ implies the result.
\end{proof}
\begin{remark}\label{rem:why-need-A-analytic}
If $\bm{v}$ is real, the contour integration step in the above proof is not necessary. Instead, the above formula is a simple consequence of the fact that
\begin{align*}
    \E \wrapb{e^{\langle \bm{g}, \bm{v} \rangle} f(\bm{g})}
    \Big/ \E \wrapb{e^{\langle \bm{g}, \bm{v} \rangle}}
    = \E \wrapb{e^{\langle \bm{g}, \bm{v} \rangle - \frac12 \|\bm{v}\|^2} f(\bm{g})},
\end{align*}
and the density reweighting on the right-hand side amounts to shifting the mean of $\bm{g}$ by $\bm{v}$. Due to the contour integration step, for complex $\bm{v}$ we need that $f$ is complex analytic in $\bm{g}$. As a result, our method will require that $A_{\bm{G}}(\beta)$ is complex analytic in $\bm{G}$; this is one of the factors motivating our choice of $A_{\bm{G}}(\beta)$.
\end{remark}
We will apply \cref{fac:planting} through the following corollary, where the disorder vector $\bm{G}$ takes the role of $\bm{g}$.
\begin{corollary}\label{cor:planted-hamiltonian}
Identify the Hamiltonian $\mathcal{H}_{\bm{G}}$ with the Gaussian vector
\begin{align*}
    \bm{G} &= (\bm{G}_{i_1,\ldots,i_p} : 2\le p\le p_{\max}, 1\le i_1,\ldots,i_p \le n) \in \C^m, &
    m &= n^2 + n^3 + \cdots + n^{p_{\max}}.
\end{align*}
For $\beta \in \C$ and $\sigma,\tau \in \mathcal{S}_{n}$, define a planted Hamiltonian (which can be identified with a mean-shifted Gaussian vector $\bm{G}^{\pl,\beta,\sigma,\tau}$)
\begin{align}
    \label{eq:planted-law}
    \mathcal{H}^{\pl,\beta,\sigma,\tau}_{\bm{G}}(x)
    = \mathcal{H}^{\nul}_{\bm{G}}(x)
    + \beta \xi\wrapp{\frac{\langle x, \sigma \rangle}{n}}
    + \bar\beta \xi\wrapp{\frac{\langle x, \tau \rangle}{n}},
\end{align}
where $\mathcal{H}^{\nul}_{\bm{G}}$ denotes a draw from the original law \cref{eq:pspin-hamiltonian} of $\mathcal{H}_{\bm{G}}$.
Then, for any $f : \C^m \rightarrow \C$ of the form \cref{eq:def-f}, where $\bm{A}$, $\bm{w}$, $p$ therein are as in \cref{fac:planting},
\begin{align*}
    \E \wrapb{
        \exp\wrapp{
            \beta \mathcal{H}_{\bm{G}}(\sigma)
            + \bar\beta \mathcal{H}_{\bm{G}}(\tau)
        }
        f(\mathcal{H}_{\bm{G}})
    }
    =
    \E \wrapb{
        \exp\wrapp{
            \beta \mathcal{H}_{\bm{G}}(\sigma)
            + \bar\beta \mathcal{H}_{\bm{G}}(\tau)
        }
    }
    \E \wrapb{f\wrapp{\mathcal{H}^{\pl,\beta,\sigma,\tau}_{\bm{G}}}}.
\end{align*}
\end{corollary}

\subsection{Second Moment Calculation}

We next introduce the two inputs that will be proved in future sections and prove \cref{prop:reweighted-second-moment}.
Let $\ov(\varrho) \in \mathcal{P}([-1,1])$ denote the distribution of the overlap $\langle \sigma, \tau \rangle / n$ for independent $\sigma,\tau \sim \varrho$.
\begin{proposition}[Proved in \cref{sec:curie-weiss-bound}]\label{prop:curie-weiss-bound}
Let $\alpha = 0.01$.
The following estimates hold for all $\beta \in \overline{\D(0,R)}$.
\begin{align}
    \label{eq:curie-weiss-bound-orthogonal}
    \E_{q\sim\ov(\varrho)} \wrapb{
        \indic\wrapb{|q| \le n^{-1/2 + \alpha}}
        \exp\wrapp{
            n|\beta|^2\xi(q)
        }
    }
    &= \frac{1}{\sqrt{1 - |\beta|^2 \xi''(0)}} + o_n(1), \\
    \label{eq:curie-weiss-bound-not-orthogonal}
    \E_{q\sim\ov(\varrho)} \wrapb{
        \indic\wrapb{|q| > n^{-1/2 + \alpha}}
        \exp\wrapp{
            n|\beta|^2\xi(q)
        }
    }
    &\le e^{-\Omega(n^{2\alpha})}.
\end{align}
\end{proposition}
For $\bm{W} \in \C^{n\times n}$, $\lambda \in \C$, $\sigma, \tau \in \mathcal{S}_{n}$, further define
\begin{align*}
    \bm{W}^{\lambda,\sigma,\tau} = \bm{W}
    + \frac{\lambda}{n} \sigma\sigma^\top
    + \frac{\bar\lambda}{n} \tau\tau^\top
\end{align*}
and
\begin{align*}
    f_\lambda(\bm{W})
    = \wrapp{
      \sum_{\Gamma \in \UC(n)}
      (-1)^{c(\Gamma)}
      \prod_{e\in \Gamma} (\lambda \bm{W}_e)
    } \wrapp{
      \sum_{\Gamma \in \UC(n)}
      (-1)^{c(\Gamma)}
      \prod_{e\in \Gamma} (\bar\lambda \bm{W}_e)
    }.
\end{align*}
For $q \in \supp\ov(\varrho)$, let $\varrho_2(q)$ be the distribution of $(\sigma,\tau) \sim \varrho^{\otimes 2}$ conditioned on $\langle \sigma, \tau \rangle / n = q$.
Let $\GOE(n)$ denote the law of the symmetric matrix $\bm{W}$ with independent centered Gaussian entries on and above the diagonal, with variance
\begin{align*}
    \E\wrapb{\bm{W}_{i,j}^2}
    = \begin{cases}
        2/n & i = j, \\
        1/n & i\neq j.
    \end{cases}
\end{align*}
Then let
\begin{align}\label{eq:def-Val}
    \Val(\lambda,q)
    &=
    |1-\lambda^2|^{-1}
    \exp\wrapp{
        \frac{n}{2}\re(\lambda^2)
        - \re(\lambda^2) - \frac{1}{4} \re(\lambda^4)
    } \\
    \nonumber
    &\times
    \E_{(\sigma,\tau) \sim \varrho_2(q)}
    \E_{\bm{W} \sim \GOE(n)} \wrapb{
      \exp\wrapp{
        - \re(\lambda) \tr (\bm{W}^{\lambda,\sigma,\tau})
        - \re(\lambda^2)
        \sum_{1\le i<j\le n}
        (\bm{W}^{\lambda,\sigma,\tau}_{i,j})^2
      } f_\lambda(\bm{W}^{\lambda,\sigma,\tau})
    }.
\end{align}

\begin{proposition}[Proved in \cref{sec:planted-cluster-expansion-moment}]\label{prop:planted-cluster-expansion-moment}
Let $\alpha = 0.01$.
Consider any $\lambda \in \overline{\D(0,1-\frac{\varepsilon}{2})}$ and $q\in \supp\ov(\varrho)$.
If $|q| \le n^{-1/2 + \alpha}$, then
\begin{align}\label{eq:planted-cluster-expansion-moment-orthogonal}
    \Val(\lambda,q) = (1 - |\lambda|^2)^{1/2} + o_n(1),
\end{align}
for an error $o_n(1)$ uniform over such $q$.
Moreover, there exists $C = C(\varepsilon)$ such that for all other $q$,
\begin{align}\label{eq:planted-cluster-expansion-moment-not-orthogonal}
    |\Val(\lambda,q)| \le C \cdot 2^{4K}.
\end{align}
\end{proposition}

\begin{proof}[Proof of \cref{prop:reweighted-second-moment}]
We calculate
\begin{align}
    \nonumber
    \E_{\bm{G}}\wrapb{
        \abs{\frac{X_{\bm{G}}(\beta)}{X_{\bm{G}}(0)}}^2
    }
    &=
    \E_{\bm{G}}\wrapb{
        Z_{\bm{G}}(\beta) Z_{\bm{G}}(\bar \beta)
        A_{\bm{G}}(\beta) A_{\bm{G}}(\bar \beta)
    } \\
    \label{eq:reweighted-second-moment-step1}
    &=
    \E_{\sigma,\tau \sim \varrho}
    \E_{\bm{G}} \wrapb{
        \exp\wrapp{
            \beta \mathcal{H}_{\bm{G}}(\sigma)
            + \bar\beta \mathcal{H}_{\bm{G}}(\tau)
        }
        A_{\bm{G}}(\beta)
        A_{\bm{G}}(\bar \beta)
    }.
\end{align}
We will evaluate the inner expectation by applying \cref{cor:planted-hamiltonian} with $f(\bm{G}) = A_{\bm{G}}(\beta) A_{\bm{G}}(\bar\beta)$.
To justify its use, we note that for some polynomial $p$, (recalling $\bm{M}, \zeta$ defined in \cref{eq:nabla2H-expansion,eq:def-bM,eq:def-zeta})
\begin{align*}
    f(\bm{G}) &= \exp\wrapp{
        -\re(\beta) \tr (\bm{M})
        -\re(\beta^2)
        \sum_{1\le i<j\le n}
        \bm{M}_{i,j}^2
    } p(\bm{G}) \\
    &= \exp \wrapp{
        -\re(\zeta) \sqrt{\frac{2}{n}} \sum_{i=1}^n \bm{G}_{i,i}
        -\frac{\re(\zeta^2)}{2n}
        \sum_{1\le i<j\le n}
        (\bm{G}_{i,j}+\bm{G}_{j,i})^2
    } p(\bm{G})
\end{align*}
In particular, if we treat $\bm{G}$ as a vector in $\C^m$ as in \cref{cor:planted-hamiltonian}, with $m$ defined therein, then $f(\bm{G})$ is of the form $\exp(\langle \bm{A} \bm{G}, \bm{G} \rangle + \langle \bm{w}, \bm{G} \rangle) p(\bm{G})$, where $\bm{A} \in \R^{m\times m}$ has entries
\begin{align*}
    \bm{A}_{(i_1,\ldots,i_p),(j_1,\ldots,j_{p'})} =
    \frac{-\re(\zeta^2)}{2n}
    \indic\wrapb{
        \text{
            $p=p'=2$, $i_1\neq i_2$, and
            ($(i_1,i_2)=(j_1,j_2)$ or $(i_1,i_2)=(j_2,j_1)$)
        }
    }.
\end{align*}
Thus all nonzero entries of $\bm{A}$ are in $2\times 2$ blocks with entries $\frac{-\re(\zeta^2)}{2n}$.
By \cref{eq:zeta-ub},
\begin{align*}
    \bm{A}
    \preceq \max\wrapp{
        \frac{-\re(\zeta^2)}{n}, 0
    } I
    \preceq
    \frac{1}{n} I.
\end{align*}
Thus \cref{cor:planted-hamiltonian} applies, and \cref{eq:reweighted-second-moment-step1} simplifies to
\begin{align}
    \nonumber
    \E_{\bm{G}} \wrapb{
        \abs{\frac{X_{\bm{G}}(\beta)}{X_{\bm{G}}(0)}}^2
    }
    &= \E_{\sigma,\tau \sim \varrho} \wrapb{
        \E_{\bm{G}} \wrapb{
            \exp\wrapp{
                \beta \mathcal{H}_{\bm{G}}(\sigma)
                + \bar\beta \mathcal{H}_{\bm{G}}(\tau)
            }
        }
        \E^{\pl,\beta,\sigma,\tau} \wrapb{
            A_{\bm{G}}(\beta) A_{\bm{G}}(\bar \beta)
        }
    } \\
    \label{eq:reweighted-second-moment-step2}
    &= \E_{q\sim \ov(\varrho)} \E_{(\sigma,\tau) \sim \varrho_2(q)} \wrapb{
        \E_{\bm{G}} \wrapb{
            \exp\wrapp{
                \beta \mathcal{H}_{\bm{G}}(\sigma)
                + \bar\beta \mathcal{H}_{\bm{G}}(\tau)
            }
        }
        \E^{\pl,\beta,\sigma,\tau} \wrapb{
            A_{\bm{G}}(\beta) A_{\bm{G}}(\bar \beta)
        }
    },
\end{align}
where $\E^{\pl,\beta,\sigma,\tau}$ denotes expectation with respect to $\mathcal{H}_{\bm{G}}$ drawn from the planted law $\mathcal{H}^{\pl,\beta,\sigma,\tau}_{\bm{G}}$ defined in \cref{eq:planted-law}.
The first of the two innermost expectations in \cref{eq:reweighted-second-moment-step2} evaluates as (similarly to the proof of \cref{lem:no-reweight-reduction-to-curie-weiss})
\begin{align*}
    \E_{\bm{G}} \wrapb{
        \exp\wrapp{
            \beta \mathcal{H}_{\bm{G}}(\sigma)
            + \bar\beta \mathcal{H}_{\bm{G}}(\tau)
        }
    }
    = \exp \wrapp{
        n\re(\beta^2) \xi(1) + n |\beta|^2 \xi(q)
    }.
\end{align*}
Note that this depends on $(\sigma,\tau)$ only through their overlap $q$.
So,
\begin{align}\label{eq:reweighted-second-moment-step3}
    \E_{\bm{G}} \wrapb{
        \abs{\frac{X_{\bm{G}}(\beta)}{X_{\bm{G}}(0)}}^2
    }
    = \E_{q\sim \ov(\varrho)} \wrapb{
        \exp \wrapp{
            n\re(\beta^2) \xi(1) + n |\beta|^2 \xi(q)
        }
        \E_{(\sigma,\tau) \sim \varrho_2(q)}
        \E^{\pl,\beta,\sigma,\tau} \wrapb{
            A_{\bm{G}}(\beta) A_{\bm{G}}(\bar \beta)
        }
    }.
\end{align}
For the remaining inner expectation, note from \cref{eq:nabla2H-expansion} that
\begin{align*}
    \bm{W} \coloneqq \xi''(0)^{-1/2} \nabla^2 \mathcal{H}^{\nul}_{\bm{G}}(0)
\end{align*}
has $\GOE(n)$ distribution.
Under $\E^{\pl,\beta,\sigma,\tau}$, $\mathcal{H}_{\bm{G}}$ is drawn from the planted law $\mathcal{H}^{\pl,\beta,\sigma,\tau}_{\bm{G}}$, so
\begin{align*}
    \bm{M}
    = \nabla^2 \mathcal{H}^{\pl,\beta,\sigma,\tau}_{\bm{G}}(0)
    = \xi''(0)^{1/2} \bm{W}
    + \frac{\beta \xi''(0)}{n} \sigma\sigma^\top
    + \frac{\bar\beta \xi''(0)}{n} \tau\tau^\top
    = \xi''(0)^{1/2} \bm{W}^{\zeta,\sigma,\tau}.
\end{align*}
Then, under $\E^{\pl,\beta,\sigma,\tau}$,
\begin{align*}
    A_{\bm{G}}(\beta) A_{\bm{G}}(\bar \beta)
    &= |1-\zeta^2|^{-1}
    \exp\wrapp{
        - n\re(\beta^2)\xi(1)
        + \frac{n}{2}\re(\zeta^2)
        - \re(\zeta^2) - \frac{1}{4} \re(\zeta^4)
    } \\
    &\qquad \times
    \exp\wrapp{
        - \re(\beta) \tr (\bm{M})
        - \re(\beta^2)
        \sum_{1\le i<j\le n}
        \bm{M}_{i,j}^2
    } \\
    &\qquad \times \wrapp{
        \sum_{\Gamma \in \UC(n)}
        (-1)^{c(\Gamma)}
        \prod_{e\in \Gamma} (\beta \bm{M}_e)
    } \wrapp{
        \sum_{\Gamma \in \UC(n)}
        (-1)^{c(\Gamma)}
        \prod_{e\in \Gamma} (\bar\beta \bm{M}_e)
    } \\
    &= |1-\zeta^2|^{-1}
    \exp\wrapp{
        - n\re(\beta^2)\xi(1)
        + \frac{n}{2}\re(\zeta^2)
        - \re(\zeta^2) - \frac{1}{4} \re(\zeta^4)
    } \\
    &\qquad \times
    \exp\wrapp{
        - \re(\zeta) \tr (\bm{W}^{\zeta,\sigma,\tau})
        - \re(\zeta^2)
        \sum_{1\le i<j\le n}
        (\bm{W}^{\zeta,\sigma,\tau}_{i,j})^2
    } f_\zeta(\bm{W}^{\zeta,\sigma,\tau}).
\end{align*}
So, for any $q \in \supp \ov(\varrho)$,
\begin{align*}
    \E_{(\sigma,\tau) \sim \varrho_2(q)}
    \E^{\pl,\beta,\sigma,\tau} \wrapb{
        A_{\bm{G}}(\beta) A_{\bm{G}}(\bar \beta)
    }
    = \exp\wrapp{
        -n\re(\beta^2)\xi(1)
    } \Val(\zeta,q).
\end{align*}
Thus \cref{eq:reweighted-second-moment-step3} simplifies to
\begin{align}
    \label{eq:reweighted-second-moment-step4}
    \E_{\bm{G}} \wrapb{
        \abs{\frac{X_{\bm{G}}(\beta)}{X_{\bm{G}}(0)}}^2
    }
    &= \E_{q\sim \ov(\varrho)} \wrapb{
        \exp \wrapp{
            n |\beta|^2 \xi(q)
        }
        \Val(\zeta,q)
    }
    = (\romI) + (\romII),
\end{align}
where (recall $\alpha = 0.01)$
\begin{align*}
    (\romI) &= \E_{q\sim \ov(\varrho)} \wrapb{
        \indic\wrapb{|q| \le n^{-1/2 + \alpha}}
        \exp \wrapp{
            n |\beta|^2 \xi(q)
        }
        \Val(\zeta,q)
    }, \\
    (\romII) &= \E_{q\sim \ov(\varrho)} \wrapb{
        \indic\wrapb{|q| > n^{-1/2 + \alpha}}
        \exp \wrapp{
            n |\beta|^2 \xi(q)
        }
        \Val(\zeta,q)
    }.
\end{align*}
Note that \cref{eq:zeta-ub} yields $|\zeta| \le 1 - \frac{\varepsilon}{2}$, so the hypothesis of \cref{prop:planted-cluster-expansion-moment} is satisfied with $\lambda = \zeta$. By \cref{prop:curie-weiss-bound,prop:planted-cluster-expansion-moment},
\begin{align*}
    (\romI)
    &\le \wrapp{\sqrt{1 - |\zeta|^2} + o_n(1)} \cdot
    \E_{q\sim \ov(\varrho)} \wrapb{
        \indic\wrapb{|q| \le n^{-1/2 + \alpha}}
        \exp \wrapp{
            n |\beta|^2 \xi(q)
        }
    } \\
    &\le \wrapp{\sqrt{1 - |\zeta|^2} + o_n(1)}
    \wrapp{\frac{1}{\sqrt{1 - |\beta|^2 \xi''(0)}} + o_n(1)}
    = 1 + o_n(1).
\end{align*}
Here we recall from \cref{eq:def-zeta} that $|\zeta| = |\beta| \xi''(0)^{1/2}$.
Moreover,
\begin{align*}
    (\romII)
    \le 2^{4K} \cdot
    \E_{q\sim \ov(\varrho)} \wrapb{
        \indic\wrapb{|q| > n^{-1/2 + \alpha}}
        \exp \wrapp{
            n |\beta|^2 \xi(q)
        }
    }
    \le 2^{4\log \log n} \cdot e^{-\Omega(n^{2\alpha})}
    = o_n(1).
\end{align*}
Plugging into \cref{eq:reweighted-second-moment-step4} concludes the proof.
\end{proof}

\subsection{Proof of Main Zero-Freeness Result}

We now derive \cref{thm:zeros-pspin-glass} from \cref{prop:reweighted-second-moment}.

\begin{proof}[Proof of \cref{thm:zeros-pspin-glass}]
As discussed just below \cref{eq:def-A-bG}, $A_{\bm{G}}$ is analytic in an open neighborhood $\Omega \supseteq \overline{\D(0,R)}$.
Therefore, $X_{\bm{G}}$ is also analytic in $\Omega$.
Setting $f_{\bm{G}}(\beta) = X_{\bm{G}}(\beta)$ in \cref{lem:jensen-formula-second-moment} and combining with \cref{prop:reweighted-second-moment} yields
\begin{align*}
    \E\wrapb{\#\{\omega \in \D(0,r) : X_{\bm{G}}(\omega) = 0\}}
    &\le \frac{1}{2\varepsilon} \E_{\theta} \log \E_{\bm{G}} \wrapb{
        \abs{\frac{X_{\bm{G}}(Re^{i\theta})}{X_{\bm{G}}(0)}}^2
    }
    \le \frac{1}{2\varepsilon} \E_{\theta} \log (1+o_n(1))
    = o_n(1).
\end{align*}
By Markov's inequality,
\begin{align*}
    \mathbb{P}\wrapb{\#\{\omega \in \D(0,r) : X_{\bm{G}}(\omega) = 0\} = 0}
    = 1 - o_n(1).
\end{align*}
The result follows because deterministically
\begin{align*}
    \#\{\omega \in \D(0,r) : Z_{\bm{G}}(\omega) = 0\}
    \le \#\{\omega \in \D(0,r) : X_{\bm{G}}(\omega) = 0\}.
\end{align*}
\end{proof}

\section{Bounds on Curie--Weiss Partition Functions}\label{sec:curie-weiss-bound}

In this section, we prove \cref{prop:curie-weiss-bound}. 
To provide intuition, we begin with a heuristic calculation justifying \cref{prop:curie-weiss-bound} under two simplifying assumptions.
\begin{itemize}
    \item A Gaussian approximation to the overlap distribution $\ov(\varrho)$. That is, replace $\ov(\varrho)$ by $\mathcal{N}(0,1/n)$. This is reasonable because in the spherical setting, the actual overlap distribution is a high dimensional sphere projected onto a one dimensional subspace, while in the Ising setting, the overlap distribution is a binomial distribution -- both of which are well-approximated by a Gaussian.
    \item A quadratic approximation to $\xi(q)$. That is, pretend that $\xi(q) = \frac{1}{2}\xi''(0) q^2$. This is essentially saying that the higher degree terms will only contribute terms which are lower order in $n$.
\end{itemize}
With these assumptions, the first part of the proposition follows via the Gaussian integral

\begin{align*}
    \E_{q\sim\ov(\varrho)} \wrapb{
        \indic\wrapb{|q| \le n^{-1/2 + \alpha}}
        \exp\wrapp{
            n|\beta|^2\xi(q)
        }
    }
    &\approx 
    \E_{q\sim \mathcal{N}(0,1/n)} \wrapb{
        \exp\wrapp{
            \frac{1}{2}n|\beta|^2 \xi''(0) q^2
        }
    } \\
    &= \int_{-\infty}^{\infty} \sqrt{\frac{n}{2\pi}} \exp\wrapp{-\frac{n}{2}\lprp{1 - |\beta|^2 \xi''(0)} q^2} \,dq \\
    &= \frac{1}{\sqrt{1 - |\beta|^2 \xi''(0)}}.
\end{align*}

The second part of the proposition is similarly easy with this simplification, because it just ends up being a Gaussian tail bound.

For the formal proof, we follow the approach in \cite{MSB21} and operate in the following slightly more general setting.
\begin{itemize}
    \item $\psi:[-1,1]\to \R$ is a smooth function (in our case, $\psi(m) = |\beta|^2 \xi(m)$)
    \item $\varrho$ is either $\mathsf{Unif}(\mathcal{C}_n)$ or $\mathsf{Unif}(\mathcal{S}_n)$
    \item $h:[-1,1]\to \R$ is the entropy term defined in \cref{def:2nd-moment-regime} (depending on the choice of $\varrho$)
    \item $\varphi(m)\defeq \psi(m) + h(m)$ satisfying
    \begin{itemize}
        \item $\varphi$ is uniquely maximized at $m^*=0$.
        \item $\varphi''(0)=-\epsilon$ for some $\epsilon>0$
    \end{itemize}
    \item $\alpha = 0.01$. We use this to define the interval from which we expect the dominant contribution to come from, as $\mathcal{I}_n = [-n^{-1/2 + \alpha}, n^{-1/2 + \alpha}]$.
\end{itemize}

Given these assumptions, the following are true.
\begin{restatable}{lemma}{CWIntegral}\label{lem:CW-integral}
    Define the following partition functions depending on $\psi$.
    \begin{align*}
        Z_\psi   &\defeq \E_{m\sim \ov(\varrho)}\wrapb{\exp\wrapp{n\psi(m)}}, \\[4pt]
        Z^0_\psi &\defeq \E_{m\sim \ov(\varrho)}\wrapb{\indic\wrapb{|m| \le n^{-1/2 + \alpha}} \cdot \exp\wrapp{n\psi(m)}}, \\[4pt]
        Z^1_\psi &\defeq \E_{m\sim \ov(\varrho)}\wrapb{\indic\wrapb{|m| > n^{-1/2 + \alpha}} \cdot \exp\wrapp{n\psi(m)}},
    \end{align*}
    We have the following estimates:
    \begin{align*}
        Z^0_\psi &= (1 + o_{n}(1)) \cdot \sqrt{\frac{1}{\varepsilon}} \cdot \exp\wrapp{n \cdot \varphi\wrapp{0}}, \\
        Z^1_\psi &= e^{-\Omega\lprp{n^{2\alpha}}}.
    \end{align*}
\end{restatable}
We now use the lemma to give a proof of \cref{prop:curie-weiss-bound}, after which we prove \cref{lem:CW-integral}.

\begin{proof}[Proof of \cref{prop:curie-weiss-bound}]
    Set $\psi(m) = |\beta|^2 \xi(m)$, and observe that $\varphi(m) = \psi(m) + h(m)$ satisfies the conditions above with $\epsilon = 1 - |\beta|^2 \xi''(0)$. Also, $\varphi(0) = 0$. Plugging this into the second conclusion of \cref{lem:CW-integral} gives

    $$Z^1_\psi = \E_{m\sim \ov(\varrho)} \wrapb{
        \indic\wrapb{|m| > n^{-1/2+\alpha}}
        \exp\wrapp{
            n|\beta|^2\xi(m)
        }
    }
    = e^{-\Omega(n^{2\alpha})}$$

    Next, plugging into the first conclusion of \cref{lem:CW-integral} gives
    $$Z^0_\psi = \E_{m\sim \ov(\varrho)} \wrapb{
        \indic\wrapb{|m| \le n^{-\frac{1}{2} + \alpha}}
        \exp\wrapp{
            n|\beta|^2\xi(m)
        }
    }
    = (1 + o_{n}(1)) \cdot \sqrt{\frac{1}{1 - |\beta|^2 \xi''(0)}}$$
\end{proof}

We now turn to a proof of \cref{lem:CW-integral}, which formalizes the idea that in a small neighborhood of $m^*$, our integral is well-approximated by a Gaussian integral -- this uses the Laplace approximation of \cref{lem:laplace-approx}, which says that any distribution that looks like $\mu \propto e^{n\varphi(m)}$ for $\varphi$ with a unique maximum essentially only depends on the second derivative of $\varphi$ at its maximizer $m^*$.
\begin{proof}[Proof of \cref{lem:CW-integral}]
First we estimate $Z_\psi^0$ by considering  the spherical and Ising cases separately. As we will see the proof works for any $\alpha \in (0,1/6)$. Set $\delta_{n} = n^{-\frac{1}{2} + \alpha}$ and $\mathcal{I}_n = [-\delta_n,\delta_n]$.

\paragraph{The Spherical Case.} When $\varrho = \mathsf{Unif}(\mathcal{S}_{n})$, we have
\begin{align*}
    Z_\psi^0 &= (1 + o_{n}(1)) \cdot \sqrt{\frac{n}{2\pi}} \cdot \int_{-\delta_{n}}^{\delta_{n}} \frac{1}{(1 - m^{2})^{3/2}} \cdot \exp(n \cdot \varphi(m)) \,dm \tag{\cref{cor:stirling} \cref{item:stirling-interior-sphere}} \\
    &= (1 + o_{n}(1)) \cdot \sqrt{\frac{1}{\varepsilon}} \cdot \exp\wrapp{n \cdot \varphi(0)}, \tag{\cref{lem:laplace-approx} and $\delta_{n} = n^{-\frac{1}{2} + \alpha}$}
\end{align*}
as we claimed in the lemma.

\paragraph{The Ising Case.} Unlike the spherical case above, since the support of $\ov(\varrho)$ is discrete when $\varrho = \mathsf{Unif}(\mathcal{C}_{n})$, we cannot directly invoke the Laplace Approximation from \cref{lem:laplace-approx}. To overcome this, we approximate the sum by a corresponding integral, and bound the error. For this, we will need \cref{lem:riemann-approx} from the appendix about Riemann approximation.

For convenience, define the smooth function
\begin{align*}
    \zeta_{n}(m) \defeq \frac{1}{2^{n}} \cdot \binom{n}{\frac{1 + m}{2} \cdot n} \cdot \exp(n \cdot \psi(m)), \qquad \forall m \in [-1,1],
\end{align*}
where the domain of $m \mapsto \binom{n}{\frac{1+m}{2} \cdot n}$ is extended to all of $[-1,1]$ via the continuous binomial coefficient: $\binom{x}{y} \defeq \frac{\Gamma(x+1)}{\Gamma(y+1) \cdot \Gamma(x-y+1)}$. Note that $\ov(\varrho)(m) = \frac{1}{2^{n}}\binom{n}{\frac{1+m}{2} \cdot n}$ for all $m \in \supp(\ov(\varrho))$.

Since
\begin{align*}
    Z_\psi^0 = \sum_{m \in \supp\wrapp{\ov(\varrho)} \cap \mathcal{I}_n} \zeta_{n}(m),
\end{align*}
\cref{lem:riemann-approx} implies
\begin{align*}
    Z_\psi^0 \leq \underset{(\mathsf{Error \, Term})}{\underbrace{\frac{n}{2} \cdot \frac{\delta_{n}}{n} \cdot \sup_{m \in \mathcal{I}_n} \abs{\zeta_{n}'(m)}}} + \underset{(\mathsf{Main \, Term})}{\underbrace{\frac{n}{2} \cdot \int_{\mathcal{I}_n} \zeta_{n}(m) \,dm}}
\end{align*}
where the $\frac{n}{2}$ scaling comes from the difference between consecutive elements of $\supp\wrapp{\ov(\varrho)}$. We may bound the main term as follows, similar to the argument for the spherical case:
\begin{align*}
    (\mathsf{Main \, Term}) &\leq (1 + o_{n}(1)) \cdot \sqrt{\frac{n}{2\pi}} \cdot \int_{-\delta_{n}}^{\delta_{n}} \sqrt{\frac{1}{1 - m^{2}}} \cdot \exp(n \cdot \varphi(m)) \,dm \tag{\cref{cor:stirling} \cref{item:stirling-interior-cube}} \\
    &\leq (1 + o_{n}(1)) \cdot \sqrt{\frac{1}{\varepsilon}} \cdot \exp\wrapp{n \cdot \varphi(0)} \tag{\cref{lem:laplace-approx} and $\delta_{n} = n^{-\frac{1}{2} + \alpha}$}
\end{align*}
Hence, it suffices to show that $(\mathsf{Error \, Term}) \leq o_{n}(1) \cdot (\mathsf{Main \, Term})$. Towards this, we claim the following.
\begin{claim}[see e.g. Lemma B.8 in \cite{MSB21}]\label{claim:stirling-deriv}
For any $m \in (-1,1)$ bounded away from the endpoints, we have
\begin{align*}
    \zeta_{n}'(m) = \zeta_{n}(m) \cdot \wrapp{n \cdot \varphi'(m) + \frac{m}{1-m^{2}} + O(1/n)}.
\end{align*}
\end{claim}
We now use this to complete the proof of the cubical case. Since $\varphi'\wrapp{m^{*}} = 0$,
\begin{align*}
    \sup_{m \in \mathcal{I}_n} \abs{\varphi'(m)} \leq \delta_{n} \cdot \sup_{m \in \mathcal{I}_n} \abs{\varphi''(m)} \leq O(\varepsilon) \cdot \delta_{n} = O\wrapp{n^{-\frac{1}{2} + \alpha}}
\end{align*}
by the Mean Value Theorem. Since $m \in \mathcal{I}_n$ is bounded away from $\pm1$, it follows that
\begin{align*}
    (\mathsf{Error \, Term}) &= \frac{\delta_{n}}{2} \cdot \sup_{m \in \mathcal{I}_n} \abs{\zeta_{n}(m) \cdot \wrapp{n \cdot \varphi'(m) + \frac{m}{1 - m^{2}} + O(1/n)}} \\
    &\leq O(\varepsilon n) \cdot \delta_{n}^{2} \cdot \zeta_{n}\wrapp{m^{*}} \\
    &\leq O\wrapp{\varepsilon n^{-\frac{1}{2} + 2\alpha}} \cdot \sqrt{\frac{1}{1 - (m^{*})^{2}}} \cdot \exp\wrapp{n \cdot \varphi\wrapp{m^{*}}} \tag{\cref{cor:stirling} \cref{item:stirling-interior-cube}} \\
    &\leq O\wrapp{n^{-1/6}} \cdot (\mathsf{Main \, Term}). \tag{Using $\alpha < 1/6$}
\end{align*}
Combined with the above upper bound on $(\mathsf{Main \, Term})$, we are done.
\,\\\\
We move on to the bound on $Z^1_\psi$. This formalizes the idea that our integral should have tail bounds similar to that of a Gaussian, using the strong concavity of $\varphi$ around its maximizer.

\begin{align*}
    Z^1_\psi
        &= \int_{[-1,1]\setminus I_n}\ov(\varrho)(m) \cdot \exp(n \cdot \psi(m)) dm \\
        &\leq \poly(n) \cdot \int_{[-1,1] \setminus \mathcal{I}_{n}} \exp(n \cdot \varphi(m)) \,dm \tag{\cref{cor:stirling} \cref{eq:stirling-polyn-loss}} \\
        &\leq \poly(n) \cdot \sup_{m \in \supp\wrapp{\ov(\varrho)} \setminus \mathcal{I}_{n}} \exp(n \cdot \varphi(m)) \\
        &\leq \poly(n) \cdot \exp\wrapp{-C\varepsilon \delta_{n}^{2}n} \cdot \int_{\mathcal{I}_{n} / 2} \exp(n \cdot \varphi(m)) \,dm \tag{\cref{lem:local-concave}} \\
        &\leq \poly(n) \cdot \exp\wrapp{-C\varepsilon \delta_{n}^{2}n} \cdot \int_{\mathcal{I}_{n}} \varrho(m) \cdot \exp(n \cdot \psi(m)) \,dm \tag{\cref{cor:stirling} \cref{eq:stirling-polyn-loss}} \\
        &\leq e^{-\Omega \lprp{n^{2\alpha}}}\cdot Z_\psi^0 \tag{Absorbing the $\poly(n)$} \\
        &\leq e^{-\Omega \lprp{n^{2\alpha}}}. \tag{Using $Z_\psi^0=O(1)$ by \cref{lem:CW-integral}}
\end{align*}
\end{proof}

\section{Expectation of Reweighting Term in Planted Model}\label{sec:planted-cluster-expansion-moment}

In this section, we will prove \cref{prop:planted-cluster-expansion-moment}. Throughout this section, we assume the setting of this proposition: we let $\alpha = 0.01$, $\lambda \in \overline{\D(0,1-\frac{\varepsilon}{2})}$, and $q\in \supp \ov(\varrho)$.
Although the first part of this proposition assumes $|q| \le n^{-1/2 + \alpha}$, we will always state this assumption when it is used.
This section is structured as follows.
\begin{itemize}
    \item The main term in $\Val(\lambda,q)$ is the expectation
    \begin{align*}
        \E_{(\sigma,\tau) \sim \varrho_2(q)}
        \E_{\bm{W} \sim \GOE(n)} \wrapb{
          \exp\wrapp{
            - \re(\lambda) \tr (\bm{W}^{\lambda,\sigma,\tau})
            - \re(\lambda^2)
            \sum_{1\le i<j\le n}
            (\bm{W}^{\lambda,\sigma,\tau}_{i,j})^2
          } f_\lambda(\bm{W}^{\lambda,\sigma,\tau})
        }.
    \end{align*}
    In \cref{subsec:cluster-exp-terms}, we address the exponential factor using a standard Gaussian change of measure calculation.
    This results in a simplified formula \cref{eq:val-step1} for $\Val(\lambda,q)$, where it remains to evaluate
    \begin{align}
        \label{eq:val-overview}
        \E_{(\sigma,\tau) \sim \rho_2(q)}
        \wrapb{
            \exp\wrapp{
                \kappa(\sigma,\tau)
            }
            \E_{\bm{W}} \wrapb{
                f_\tlambda(\bm{W}^{\tlambda,\sigma,\tau})
            }
        }.
    \end{align}
    Here $\tlambda$ is a slight rescaling of $\lambda$, see \cref{eq:def-tchi}, while $\kappa(\sigma,\tau)$ is an error term small enough that its exponential will not contribute significantly to this expectation; thus the interesting term in \cref{eq:val-overview} is the cluster expansion $f_\tlambda(\bm{W}^{\tlambda,\sigma,\tau})$.
    \item In \cref{subsec:cluster-gaussian}, we evaluate the inner expectation $\E_{\bm{W}} \wrapb{f_\tlambda(\bm{W}^{\tlambda,\sigma,\tau})}$ in \cref{eq:val-overview}.
    This is a relatively simple calculation, as $f_\tlambda(\bm{W}^{\tlambda,\sigma,\tau})$ is a product of two multilinear functions in the entries of $\bm{W}$.
    The result is a polynomial in the entries of $\sigma,\tau$, whose monomials are indexed by graphs consisting of cycles (as in $\UC(n)$) \emph{and paths}.
    \item This simplifies \cref{eq:val-overview} to a combinatorial sum, whose evaluation is the main task of this section.
    In \cref{subsec:cluster-overview}, we provide a heuristic overview of this calculation in the case $|q| \le n^{-1/2 + \alpha}$ and explain how the estimate $(1-|\lambda|^2)^{1/2}$ in \cref{prop:planted-cluster-expansion-moment} appears.
    In \cref{subsec:cluster-rigorous} we carry out this calculation rigorously and prove both parts of \cref{prop:planted-cluster-expansion-moment}.
    Some proofs are deferred to \cref{subsec:sph-mmt}.
\end{itemize}

\subsection{Exponential Terms}
\label{subsec:cluster-exp-terms}

The following fact will be useful.
\begin{fact}
  \label{fac:real-squared}
  For all $z\in \C$, $\re(z)^2 = \frac{1}{2} (|z|^2 + \re(z^2))$.
\end{fact}
\begin{proof}
  Write $z = u+iv$, and note that $|z|^2 = u^2+v^2$ while $\re(z^2) = u^2-v^2$.
\end{proof}
We next evaluate the contribution from the trace term in $\Val(\lambda,q)$, which is independent of everything else (because the term $f_\lambda(\bm{W}^{\lambda,\sigma,\tau})$ does not depend on the diagonal entries of $\bm{W}^{\lambda,\sigma,\tau}$).
\begin{proposition}
  \label{prop:trace-term}
  For any $\sigma,\tau \in \mathcal{S}_n$, $\E_{\bm{W}} \exp(-\re(\lambda) \tr(\bm{W}^{\lambda,\sigma,\tau})) = \exp(-\frac{1}{2} |\lambda|^2 - \frac{1}{2} \re(\lambda^2))$.
\end{proposition}
\begin{proof}
  Note that
  \[
    \tr(\bm{W}^{\lambda,\sigma,\tau})
    = \tr(\bm{W}) + \lambda + \bar \lambda
    \stackrel{d}{=} \sqrt{2} Z + 2\re(\lambda),
  \]
  where $Z \sim \mathcal{N}(0,1)$.
  Thus, by \cref{fac:real-squared},
  \begin{align*}
    \E_{\bm{W}} \wrapb{
        \exp(-\re(\lambda) \tr(\bm{W}^{\lambda,\sigma,\tau}))
    }
    &= \E \wrapb{
        \exp(-\sqrt{2}\re(\lambda)Z - 2\re(\lambda)^2)
    }
    = \exp(-\re(\lambda)^2) \\
    &= \exp\wrapp{-\frac{1}{2} |\lambda|^2 - \frac{1}{2} \re(\lambda^2)}.
  \end{align*}
\end{proof}
\noindent The remaining expectation can be evaluated with the help of the following lemma.
\begin{lemma}
  \label{lem:gaussian-quadratic-rewt}
  Let $f : \C \to \C$ be a polynomial and $s\in \R$ satisfy $|s| \le 1$.
  Further let $W \sim \mathcal{N}(0,1/n)$ and $\Delta \in \C$ be arbitrary and fixed. Then,
  \[
    \E \wrapb{
      \exp\wrapp{-s(W+\Delta)^2} f(W+\Delta)
    }
    = \sqrt{\frac{n}{n+2s}} \exp\wrapp{
      - \frac{n}{n+2s} \cdot s\Delta^2
    }
    \E \wrapb{
        f\wrapp{
          \sqrt{\frac{n}{n+2s}}W + \frac{n}{n+2s}\Delta
        }
    }
  \]
\end{lemma}
\begin{remark}
  The proof below works for any measurable $f$ of suitably bounded growth and all $s\in (-\infty, n/2)$.
\end{remark}
\begin{proof}
  This is a routine gaussian calculation:
  \begin{align*}
    \E \wrapb{
      \exp\wrapp{-s(W+\Delta)^2} f(W+\Delta)
    }
    &= \sqrt{\frac{n}{2\pi}}
    \int_\R \exp\wrapp{
      -\frac{n}{2} x^2 - s(x+\Delta)^2
    } f(x+\Delta)
    \,d x \\
    &= \sqrt{\frac{n}{2\pi}}
    \int_\R \exp\wrapp{
      -\frac{n+2s}{2} \wrapp{x + \frac{2s\Delta}{n+2s}}^2 - \frac{n}{n+2s} \cdot s\Delta^2
    } f(x+\Delta)
    \,d x \\
    &= \sqrt{\frac{n}{n+2s}} \exp\wrapp{
      - \frac{n}{n+2s} \cdot s\Delta^2
    }
    \E \wrapb{
        f\wrapp{
          \wrapp{-\frac{2s\Delta}{n+2s} + \sqrt{\frac{n}{n+2s}}W}
          +\Delta
        }
    }.
  \end{align*}
\end{proof}

\begin{proposition}
  \label{prop:exp-quadratic-term}
  For any $\sigma,\tau \in \mathcal{S}_n$ with $\langle \sigma,\tau \rangle / n = q$,
  \begin{align*}
    &\E_{\bm{W}} \wrapb{
      \exp\wrapp{
        - \re(\lambda^2)
        \sum_{1\le i<j\le n}
        (W^{\lambda,\sigma,\tau}_{i,j})^2
      }
      f_\lambda(\bm{W}^{\lambda,\sigma,\tau})
    } \\
    &= \exp\wrapp{
      -\frac{n}{2} \re(\lambda^2) + \frac{1}{2} \re(\lambda^2) - \frac{1}{4} |\lambda|^4 - \frac{1}{4} \re(\lambda^4)
      - \re(\lambda^2)|\lambda|^2 q^2
      + \kappa(\sigma,\tau) + O(n^{-1})
    }
    \E_{\bm{W}} \wrapb{f_\tlambda(\bm{W}^{\tlambda,\sigma,\tau})},
  \end{align*}
  where $\kappa(\sigma,\tau)$ is an error term satisfying
  \begin{align}
    \label{eq:kappa-bd}
    |\kappa(\sigma,\tau)|
    \le \frac{1}{n^2} (\|\sigma\|_4^4 + \|\tau\|_4^4)
  \end{align}
  and
  \begin{align}
    \label{eq:def-tchi}
    \tlambda = \sqrt{\frac{n}{n+2\re(\lambda^2)}} \lambda.
  \end{align}
\end{proposition}
\begin{proof}
  It will be helpful to abbreviate
  \begin{align*}
    s &= \re(\lambda^2), &
    \bm{\Delta} &= \frac{\lambda}{n} \sigma\sigma^\top + \frac{\bar\lambda}{n} \tau\tau^\top, &
    \bm{\Delta}_{i,j} &= \frac{\lambda}{n} \sigma_i\sigma_j + \frac{\bar\lambda}{n} \tau_i\tau_j.
  \end{align*}
  We apply \cref{lem:gaussian-quadratic-rewt} (whose hypothesis $|s| \le 1$ applies because $\lambda \in \overline{\D(0,1-\frac{\varepsilon}{2})}$) to each entry $\bm{W}_{i,j}$ of $W$.
  This shows
  \begin{align}
    \nonumber
    &\E_{\bm{W}} \wrapb{
      \exp\wrapp{
        -s
        \sum_{1\le i<j\le n}
        (W^{\lambda,\sigma,\tau}_{i,j})^2
      }
      f_\lambda(\bm{W}^{\lambda,\sigma,\tau})
    } \\
    \label{eq:exp-quadratic-term-step1}
    &= \wrapp{\frac{n}{n+2s}}^{n(n-1)/4} \exp\wrapp{
      - \frac{ns}{n+2s} \sum_{1\le i<j\le n} \bm{\Delta}_{i,j}^2
    } \E_{\bm{W}} \wrapb{
        f_\lambda\wrapp{
          \sqrt{\frac{n}{n+2s}}\bm{W}
          + \frac{n}{n+2s}\bm{\Delta}
        }
    }.
  \end{align}
  It is easy to check that
  \begin{align*}
    \lambda \wrapp{
      \sqrt{\frac{n}{n+2s}}\bm{W}
      + \frac{n}{n+2s}\bm{\Delta}
    }
    &= \tlambda\bm{W}^{\tlambda,\sigma,\tau}, &
    \bar\lambda \wrapp{
      \sqrt{\frac{n}{n+2s}}\bm{W}
      + \frac{n}{n+2s}\bm{\Delta}
    }
    &= \bar\tlambda\bm{W}^{\tlambda,\sigma,\tau},
  \end{align*}
  and thus
  \[
    f_\lambda\wrapp{
      \sqrt{\frac{n}{n+2s}}\bm{W}
      + \frac{n}{n+2s}\bm{\Delta}
    }
    = f_\tlambda(\bm{W}^{\tlambda,\sigma,\tau}).
  \]
  It remains to evaluate the first two factors on the right-hand side of \cref{eq:exp-quadratic-term-step1}.
  We have
  \begin{align*}
    \wrapp{\frac{n}{n+2s}}^{n(n-1)/4}
    &= \wrapp{1 + \frac{2s}{n}}^{-n(n-1)/4}
    = \exp\wrapp{
      \frac{2s}{n} - \frac{2s^2}{n^2} + O(n^{-3})
    }^{-n(n-1)/4} \\
    &= \exp\wrapp{
      -\frac{n}{2} s + \frac{1}{2} s + \frac{1}{2} s^2 + O(n^{-1})
    }
  \end{align*}
  For use below, define
  \begin{align*}
    \tkappa(\sigma,\tau) &= \sum_{i=1}^n (\lambda \sigma_i^2 + \bar\lambda \tau_i^2)^2, &
    \kappa(\sigma,\tau) &= \frac{s}{2n^2} \tkappa(\sigma,\tau).
  \end{align*}
  Then,
  \[
    |\tkappa(\sigma,\tau)|
    \le \sum_{i=1}^n |\lambda \sigma_i^2 + \bar\lambda \tau_i^2|^2
    \le 2\sum_{i=1}^n \wrapp{|\lambda \sigma_i^2|^2 + |\bar\lambda \tau_i^2|^2}
    \le 2(\|\sigma\|_4^4 + \|\tau\|_4^4),
  \]
  where the final inequality follows from $\lambda \in \overline{\D(0,1-\frac{\varepsilon}{2})}$.
  Consequently $\kappa(\sigma,\tau)$ satisfies \cref{eq:kappa-bd}.
  Then,
  \begin{align*}
    \sum_{1\le i<j\le n} (\bm{\Delta}_{i,j})^2
    &= \frac{1}{n^2} \sum_{1\le i<j\le n} (\lambda \sigma_i\sigma_j + \bar\lambda \tau_i\tau_j)^2 \\
    &= \frac{1}{2n^2} \wrapp{
      \sum_{i,j=1}^n (\lambda \sigma_i\sigma_j + \bar\lambda \tau_i\tau_j)^2
      - \tkappa(\sigma,\tau)
    } \\
    &= \frac{1}{2n^2} \wrapp{
      n^2 (\lambda^2 + \bar\lambda^2 + 2|\lambda|^2 q^2)
      - \tkappa(\sigma,\tau)
    } \\
    &= \re(\lambda^2) + |\lambda|^2 q^2
    - \frac{1}{2n^2} \tkappa(\sigma,\tau)
  \end{align*}
  Thus,
  \begin{align*}
    &\wrapp{\frac{n}{n+2s}}^{n(n-1)/4} \exp\wrapp{
      - \frac{ns}{n+2s} \sum_{1\le i<j\le n} \bm{\Delta}_{i,j}^2
    } \\
    &= \exp\wrapp{
      -\frac{n}{2} s + \frac{1}{2} s + \frac{1}{2} s^2
      - s(\re(\lambda^2) + |\lambda|^2 q^2)
      + \frac{n}{n+2s} \kappa(\sigma,\tau)
      + O(n^{-1})
    } \\
    &= \exp\wrapp{
      -\frac{n}{2} \re(\lambda^2) + \frac{1}{2} \re(\lambda^2) - \frac{1}{2} \re(\lambda^2)^2
      - \re(\lambda^2) |\lambda|^2 q^2
      + \kappa(\sigma,\tau)
      + O(n^{-1})
    }.
  \end{align*}
  \cref{fac:real-squared} yields $\re(\lambda^2)^2 = \frac{1}{2} |\lambda|^4 + \frac{1}{2} \re(\lambda^4)$, which completes the proof.
\end{proof}
\noindent \cref{prop:trace-term,prop:exp-quadratic-term} imply the following simplified formula for $\Val(\lambda,q)$.
For $\tlambda$ defined in \cref{eq:def-tchi},
\begin{align}
  \nonumber
  \Val(\lambda,q)
  &=
  |1-\lambda^2|^{-1}
  \exp\wrapp{
    - \re(\lambda^2)
    - \frac{1}{2} \re(\lambda^4)
    - \frac{1}{2} |\lambda|^2
    - \frac{1}{4} |\lambda|^4
    - \re(\lambda^2)|\lambda|^2 q^2
    + O(n^{-1})
  } \\
  \label{eq:val-step1}
  &\qquad \times \E_{(\sigma,\tau) \sim \rho_2(q)}
  \wrapb{
    \exp\wrapp{
      \kappa(\sigma,\tau)
    }
    \E_{\bm{W}} \wrapb{
      f_\tlambda(\bm{W}^{\tlambda,\sigma,\tau})
    }
  }.
\end{align}

\subsection{Cluster Expansion Term: Expectation over Gaussian Matrix}
\label{subsec:cluster-gaussian}

It remains to evaluate the expectation in the right-hand side of \cref{eq:val-step1}.
In this subsection, we will evaluate the inner expectation over $\bm{W}$, which will yield a combinatorial sum over subgraphs on $[n]$.
The rest of this section will evaluate this combinatorial sum.
In \cref{subsec:cluster-overview}, we provide a non-rigorous outline of how this is achieved, followed by a rigorous evaluation in \cref{subsec:cluster-rigorous}.
For convenience, recall that
\[
  f_\tlambda(\bm{W}^{\tlambda,\sigma,\tau})
  = p_1(\bm{W},\sigma,\tau)p_2(\bm{W},\sigma,\tau),
\]
where
\begin{align*}
  p_1(\bm{W},\sigma,\tau)
  &=
  \sum_{\Gamma \in \UC(n)}
  (-1)^{c(\Gamma)}
  \prod_{\{i,j\} \in \Gamma} \wrapp{
    \tlambda \bm{W}_{i,j}
    + \frac{\tlambda^2}{n} \sigma_i\sigma_j
    + \frac{|\tlambda|^2}{n} \tau_i\tau_j
  }, \\
  p_2(\bm{W},\sigma,\tau)
  &=
  \sum_{\Gamma \in \UC(n)}
  (-1)^{c(\Gamma)}
  \prod_{\{i,j\} \in \Gamma} \wrapp{
    \bar\tlambda \bm{W}_{i,j}
    + \frac{|\tlambda|^2}{n} \sigma_i\sigma_j
    + \frac{\bar\tlambda^2}{n} \tau_i\tau_j
  }.
\end{align*}
We can expand $p_1,p_2$ as polynomials in $\bm{W}$.
These polynomials contain monomials in $\bm{W}$ which are indexed by the following set:
\begin{align}
  \label{eq:P}
  P(n) = \wrapc{
    \Gamma_0 \subseteq \binom{[n]}{2} :
    \begin{array}{c}
      \text{
        $\Gamma_0$ is the edge set of a vertex-disjoint union} \\
      \text{of paths and cycles
        and $|\Gamma_0| \le K$
      }
    \end{array}
  }.
\end{align}
Furthermore, for $\Gamma_0 \in P(n)$, let
\begin{align}
  \label{eq:uc-completion}
  \UC(n,\Gamma_0) = \wrapc{
    \Lambda \subseteq \binom{[n]}{2} :
    \Gamma_0 \cap \Lambda = \emptyset\,\,\text{and}\,\,
    \Gamma_0 \cup \Lambda \in \UC(n)
  }.
\end{align}
Then, we can write
\begin{align*}
  p_1(\bm{W},\sigma,\tau)
  &= \sum_{\Gamma_0 \in P(n)}
  q_1(\Gamma_0,\sigma,\tau)
  \prod_{e\in \Gamma_0} (\tlambda \bm{W}_e), &
  p_2(\bm{W},\sigma,\tau)
  &= \sum_{\Gamma_0 \in P(n)}
  q_2(\Gamma_0,\sigma,\tau)
  \prod_{e\in \Gamma_0} (\bar\tlambda \bm{W}_e),
\end{align*}
where
\begin{align*}
  q_1(\Gamma_0,\sigma,\tau)
  &=
  \sum_{\Lambda \in \UC(n,\Gamma_0)}
  (-1)^{c(\Gamma_0\cup\Lambda)}
  \prod_{\{i,j\} \in \Lambda} \wrapp{
    \frac{\tlambda^2}{n} \sigma_i\sigma_j
    + \frac{|\tlambda|^2}{n} \tau_i\tau_j
  } \\
  q_2(\Gamma_0,\sigma,\tau)
  &=
  \sum_{\Lambda \in \UC(n,\Gamma_0)}
  (-1)^{c(\Gamma_0\cup\Lambda)}
  \prod_{\{i,j\} \in \Lambda} \wrapp{
    \frac{|\tlambda|^2}{n} \sigma_i\sigma_j
    + \frac{\bar\tlambda^2}{n} \tau_i\tau_j
  }
\end{align*}
Consequently
\[
  \E_{\bm{W}} \wrapb{f_\tlambda(\bm{W}^{\tlambda,\sigma,\tau})}
  = \sum_{\Gamma_0\in P(n)}
  \wrapp{\frac{|\tlambda|^2}{n}}^{|\Gamma_0|}
  q_1(\Gamma_0,\sigma,\tau)
  q_2(\Gamma_0,\sigma,\tau).
\]
This implies
\begin{align}
  \nonumber
  \Val(\lambda,q)
  &=
  |1-\lambda^2|^{-1}
  \exp\wrapp{
    - \re(\lambda^2)
    - \frac{1}{2} \re(\lambda^4)
    - \frac{1}{2} |\lambda|^2
    - \frac{1}{4} |\lambda|^4
    - \re(\lambda^2)|\lambda|^2 q^2
    + O(n^{-1})
  } \\
  \label{eq:val-step2}
  &\qquad \times
  \wrapb{
    \sum_{\Gamma_0\in P(n)}
    \wrapp{\frac{|\tlambda|^2}{n}}^{|\Gamma_0|}
    \E_{(\sigma,\tau) \sim \rho_2(q)}
    \wrapb{
      \exp\wrapp{
        \kappa(\sigma,\tau)
      }
      q_1(\Gamma_0,\sigma,\tau)
      q_2(\Gamma_0,\sigma,\tau)
    }
  }.
\end{align}

\subsection{Cluster Expansion Term: Technical Overview of Combinatorial Sum}
\label{subsec:cluster-overview}

\cref{eq:planted-cluster-expansion-moment-not-orthogonal} in \cref{prop:planted-cluster-expansion-moment} follows from relatively crude estimates on $q_1,q_2$.
In this subsection, we non-rigorously explain how to obtain \cref{eq:planted-cluster-expansion-moment-orthogonal} in \cref{prop:planted-cluster-expansion-moment}, that for $|q| \le n^{-1/2+\alpha}$ the quantity \cref{eq:val-step2} is $(1-|\lambda|^2)^{1/2}$ to leading order.

First, recall from \cref{eq:kappa-bd} that $|\kappa(\sigma,\tau)| \le n^{-2} (\|\sigma\|_4^4 + \|\tau\|_4^4)$.
In the Ising setting $\rho = \unif(C_n)$, this is deterministically $o_n(1)$.
In the spherical setting $\rho = \unif(\mathcal{S}_n)$, this is $o_n(1)$ unless $\sigma$ or $\tau$ is delocalized (which occurs with low probability), and deterministically bounded by $2$.
Thus the factor $\exp(\kappa(\sigma,\tau))$ should not affect the expectation in \cref{eq:val-step2}.

We will see that the contribution to \cref{eq:val-step2} from $\Gamma_0 \in P(n) \setminus \UC(n)$ (i.e. those $\Gamma_0$ that contain at least one path) is also $o_n(1)$. 
Roughly speaking, this comes from degree of freedom considerations.
Consider expanding $q_1, q_2$ in \cref{eq:val-step2} to obtain a combinatorial sum over $\Gamma_0 \in P_n$ and $\Lambda_1,\Lambda_2 \in \UC(n,\Gamma_0)$.
The term indexed by $(\Gamma_0,\Lambda_1,\Lambda_2)$ is of order $n^{-E(\Gamma_0,\Lambda_1,\Lambda_2)}$, where $E(\Gamma_0,\Lambda_1,\Lambda_2) = |\Gamma_0| + |\Lambda_1| + |\Lambda_2|$ is is the total number of edges in these subgraphs.
The number of $(\Gamma_0,\Lambda_1,\Lambda_2)$ of a certain shape is of order $n^{V(\Gamma_0,\Lambda_1,\Lambda_2)}$, where $V(\Gamma_0,\Lambda_1,\Lambda_2)$ is the number of independent vertex choices required to specify $(\Gamma_0,\Lambda_1,\Lambda_2)$.
If $\Gamma_0 \in \UC(n)$, then $\Lambda_1,\Lambda_2$ are also unions of disjoint cycles, so the total degrees of freedom $V(\Gamma_0,\Lambda_1,\Lambda_2) - E(\Gamma_0,\Lambda_1,\Lambda_2)$ is zero; this is why we expect the total contribution of such graphs to be $\Theta(1)$.
If $\Gamma_0$ contains paths, each path in increases the degrees of freedom of $\Gamma_0$ by $1$ (it has one more vertex than edges), but decreases the degrees of freedom of both $\Lambda_1$ and $\Lambda_2$ by $1$, as their paths must complete all paths in $\Gamma_0$ to cycles.
Thus each path in $\Gamma_0$ decrements $V(\Gamma_0,\Lambda_1,\Lambda_2) - E(\Gamma_0,\Lambda_1,\Lambda_2)$, and we expect \cref{eq:val-step2} to be dominated by $\Gamma_0$ consisting of only cycles.

So we focus on the following sum, which approximates the last line of \cref{eq:val-step2}:
\begin{align}
  \label{eq:heuristic-expr1}
  \sum_{\Gamma_0\in \UC(n)}
  \wrapp{\frac{|\tlambda|^2}{n}}^{|\Gamma_0|}
  \E_{(\sigma,\tau) \sim \rho_2(q)}
  \wrapb{
    q_1(\Gamma_0,\sigma,\tau)
    q_2(\Gamma_0,\sigma,\tau)
  }
\end{align}

Fix any $\Gamma_0 \in \UC(n)$.
Then each $\Lambda \in \UC(n,\Gamma_0)$ is a vertex-disjoint union of cycles which shares no vertices with $\Gamma_0$.
We can describe each such union of cycles with a sequence $(A_3,A_4,\ldots)$, where $A_k$ is the number of cycles of length $k$.
Let
\begin{align}
  \label{eq:def-Ct}
  \Ct(K) = \wrapc{
    \Omega = (A_3,A_4,\ldots) :
    A_k \ge 0, \sum_{k=3}^\infty kA_k \le K
  }.
\end{align}
(The sequence $(A_3,A_4,\ldots)$ is infinite, but eventually zero due to the constraint $\sum_{k=3}^\infty kA_k \le K$.)
For $\Omega = (A_3,A_4,\ldots) \in \Ct(K)$ (and $\Gamma_0 \in \UC(n)$ as above), let
\begin{align*}
  \UC(n,\Omega) &= \wrapc{
    \Gamma \in \UC(n) :
    \text{$\Gamma$ has $A_k$ cycles of length $k$}
  }, \\
  \UC(n,\Gamma_0,\Omega) &= \wrapc{
    \Lambda \in \UC(n,\Gamma_0) :
    \text{$\Lambda$ has $A_k$ cycles of length $k$}
  }.
\end{align*}
Also, let
\begin{align*}
  |\Omega| &= \sum_{k=3}^\infty kA_k, &
  c(\Omega) &= \sum_{k=3}^\infty A_k, &
  \sym(\Omega) &= \prod_{k=3}^\infty \wrapp{(2k)^{A_k} A_k!}.
\end{align*}
Note that $|\Omega|$ and $c(\Omega)$ equal $|\Gamma|$ and $c(\Gamma)$ for any $\Gamma \in \UC(n,\Omega)$.
Also, $\sym(\Omega)$ is the number of automorphisms of a $\Gamma \in \UC(n,\Omega)$: each $k$-cycle has $2k$ symmetries, and the $A_k$ $k$-cycles can be permuted in $A_k!$ ways.
Continuing,
\begin{align*}
  q_1(\Gamma_0,\sigma,\tau)
  &= (-1)^{c(\Gamma_0)}
  \sum_{\Lambda \in \UC(n,\Gamma_0)}
  (-1)^{c(\Lambda)}
  \prod_{\{i,j\} \in \Lambda} \wrapp{
    \frac{\tlambda^2}{n} \sigma_i\sigma_j
    + \frac{|\tlambda|^2}{n} \tau_i\tau_j
  } \\
  &= (-1)^{c(\Gamma_0)}
  \sum_{\Omega \in \Ct(K)}
  (-1)^{c(\Omega)}
  \sum_{\Lambda \in \UC(n,\Gamma_0,\Omega)}
  \prod_{\{i,j\} \in \Lambda} \wrapp{
    \frac{\tlambda^2}{n} \sigma_i\sigma_j
    + \frac{|\tlambda|^2}{n} \tau_i\tau_j
  }.
\end{align*}
The key estimate is
\begin{align}
  \label{eq:overview-approx1}
  \sym(\Omega) \cdot \sum_{\Lambda \in \UC(n,\Gamma_0,\Omega)}
  \prod_{\{i,j\} \in \Lambda} \wrapp{
    \frac{\tlambda^2}{n} \sigma_i\sigma_j
    + \frac{|\tlambda|^2}{n} \tau_i\tau_j
  }
  &\approx \prod_{k=3}^\infty \tr\wrapp{\wrapp{
    \frac{\tlambda^2}{n} \sigma\sigma^\top
    + \frac{|\tlambda|^2}{n} \tau\tau^\top
  }^k}^{A_k} \\
  \label{eq:overview-approx2}
  &\approx \prod_{k=3}^\infty \wrapp{\tlambda^{2k} + |\tlambda|^{2k}}^{A_k}
  \equiv V(\Omega).
\end{align}
In the step \cref{eq:overview-approx1}, both sides are sums of the form
\[
  \sum_{E\in S}
  \prod_{\{i,j\} \in E} \wrapp{
    \frac{\tlambda^2}{n} \sigma_i\sigma_j
    + \frac{|\tlambda|^2}{n} \tau_i\tau_j
  },
\]
where $E$ is a (multi-)set of pairs $\{i,j\}$ and $S$ is a set of such $E$.
The left-hand side of \cref{eq:overview-approx1} sums over $E$ which form the edge set of a union of cycles containing $A_k$ cycles of length $k$, which avoids the vertices of $\Gamma_0$ and does not repeat vertices.
The right-hand side sums over $E$ which form the edge set of a union of cycles containing $A_k$ cycles of length $k$, but which are allowed to use vertices of $\Gamma_0$ or repeat vertices.
(We need a factor of $\sym(\Omega)$ on the left, because the left-hand side sums over edge sets while the right-hand side sums over sequences of vertices.)
The approximation \cref{eq:overview-approx1} heuristically holds because the difference between these sets of $E$ contributes sub-leading order (since both $|\Gamma_0|$ and $|E|$ are bounded by $K$, which is small).
The approximation \cref{eq:overview-approx2} holds because we assumed $q = \langle \sigma,\tau \rangle / n$ is close to $0$.
Thus
\begin{align}
  \nonumber
  (-1)^{c(\Gamma_0)}
  q_1(\Gamma_0,\sigma,\tau)
  &\approx
  \sum_{\Omega \in \Ct(K)}
  \frac{(-1)^{c(\Omega)} V(\Omega)}{\sym(\Omega)}
  = \sum_{\Omega = (A_3,A_4,\ldots) \in \Ct(K)}
  \prod_{k=3}^\infty \wrapp{- \frac{\tlambda^{2k} + |\tlambda|^{2k}}{2k} }^{A_k} \cdot \frac{1}{A_k!} \\
  \label{eq:overview-approx3}
  &\approx
  \sum_{A_3,A_4,\ldots \ge 0}
  \prod_{k=3}^\infty \wrapp{- \frac{\tlambda^{2k} + |\tlambda|^{2k}}{2k} }^{A_k} \cdot \frac{1}{A_k!} \\
  \nonumber
  &=
  \prod_{k=3}^\infty \exp\wrapp{
    - \frac{\tlambda^{2k} + |\tlambda|^{2k}}{2k}
  }
  = \exp\wrapp{
    - \sum_{k=1}^\infty \frac{\tlambda^{2k} + |\tlambda|^{2k}}{2k}
  } \exp\wrapp{
    \frac{\tlambda^2 + |\tlambda|^2}{2} + \frac{\tlambda^4 + |\tlambda|^4}{4}
  } \\
  \nonumber
  &= (1-\tlambda^2)^{1/2} (1-|\tlambda|^2)^{1/2} \exp\wrapp{
    \frac{\tlambda^2 + |\tlambda|^2}{2} + \frac{\tlambda^4 + |\tlambda|^4}{4}
  }
\end{align}
Analogously,
\[
  (-1)^{c(\Gamma_0)} q_2(\Gamma_0,\sigma,\tau)
  \approx (1-\bar\tlambda^2)^{1/2} (1-|\tlambda|^2)^{1/2} \exp\wrapp{
    \frac{\bar\tlambda^2 + |\tlambda|^2}{2} + \frac{\bar\tlambda^4 + |\tlambda|^4}{4}
  }
\]
so
\[
  q_1(\Gamma_0,\sigma,\tau)q_2(\Gamma_0,\sigma,\tau)
  \approx |1-\tlambda^2| (1-|\tlambda|^2) \exp\wrapp{
    \re(\tlambda^2) + |\tlambda|^2 + \frac{\re(\tlambda^4) + |\tlambda|^4}{2}
  }.
\]
Thus \cref{eq:heuristic-expr1} is approximately
\begin{align}
  \label{eq:heuristic-expr2}
  |1-\tlambda^2| (1-|\tlambda|^2) \exp\wrapp{
    \re(\tlambda^2) + |\tlambda|^2 + \frac{\re(\tlambda^4) + |\tlambda|^4}{2}
  }
  \sum_{\Gamma_0\in \UC(n)}
  \wrapp{\frac{|\tlambda|^2}{n}}^{|\Gamma_0|}.
\end{align}
The remaining sum can be evaluated by a similar argument.
Note that for $\Omega \in \Ct(K)$,
\[
  |\UC(n,\Omega)|
  = \frac{n(n-1)\cdots(n-|\Omega|+1)}{\sym(\Omega)}
  \approx \frac{n^{|\Omega|}}{\sym(\Omega)}.
\]
Thus,
\begin{align}
  \nonumber
  \sum_{\Gamma_0\in \UC(n)}
  \wrapp{\frac{|\tlambda|^2}{n}}^{|\Gamma_0|}
  &=
  \sum_{\Omega \in \Ct(K)}
  |\UC(n,\Omega)|
  \wrapp{\frac{|\tlambda|^2}{n}}^{|\Omega|}
  \approx
  \sum_{\Omega \in \Ct(K)}
  \frac{|\tlambda|^{2|\Omega|}}{\sym(\Omega)} \\
  \label{eq:overview-approx4}
  &\approx
  \sum_{A_3,A_4,\ldots \ge 0}
  \prod_{k=3}^{\infty} \wrapp{\frac{|\tlambda|^{2k}}{2k}}^{A_k} \cdot \frac{1}{A_k!}
  = \prod_{k=3}^{\infty} \exp\wrapp{\frac{|\tlambda|^{2k}}{2k}} \\
  \nonumber
  &=(1-|\tlambda|^2)^{-1/2} \cdot \exp\wrapp{-\frac{|\tlambda|^2}{2}-\frac{|\tlambda|^4}{4}}.
\end{align}
Then \cref{eq:heuristic-expr2} is approximately
\begin{align}
  \label{eq:heuristic-expr3}
  |1-\tlambda^2| (1-|\tlambda|^2)^{1/2} \exp\wrapp{
    \re(\tlambda^2) + \frac{1}{2} |\tlambda|^2 + \frac{1}{2} \re(\tlambda^4) + \frac{1}{4} |\tlambda|^4
  }.
\end{align}
Recall that $\tlambda = \lambda + O(n^{-1})$.
Plugging this into \cref{eq:val-step2} then shows $\Val(\lambda,q) \approx (1-|\lambda|^2)^{1/2}$.
In the next subsection we make this argument rigorous by justifying the approximations made above:
\begin{itemize}
  \item The contribution to \cref{eq:val-step2} from delocalized $\sigma,\tau$ is small, so we may ignore the $\kappa(\sigma,\tau)$ factor. (In the Ising setting this is automatic because $\sigma,\tau \in \{\pm 1\}^n$.)
  \item The contribution to \cref{eq:val-step2} from $\Gamma_0 \in P(n) \setminus \UC(n)$ is small.
  \item The approximation \cref{eq:overview-approx1} of a sum over vertex-disjoint cycles by a product of traces holds, with small error in expectation over $(\sigma,\tau)$.
  \item The approximation \cref{eq:overview-approx2} of the trace holds for all $|q| \le n^{-1/2 + \alpha}$.
  \item The approximations \cref{eq:overview-approx3} and \cref{eq:overview-approx4} of sums over $\Omega \in \Ct(K)$ by infinite sums hold.
  \item $\Val(\lambda,q)$ can be bounded crudely to prove \cref{eq:planted-cluster-expansion-moment-not-orthogonal} of \cref{prop:planted-cluster-expansion-moment}.
\end{itemize}

\subsection{Cluster Expansion Term: Rigorous Evaluation of Combinatorial Sum}
\label{subsec:cluster-rigorous}

For the proof of \cref{prop:planted-cluster-expansion-moment}, we will define a sequence $X_1,\ldots,X_7$ below, where $X_1$ is the combinatorial sum in \cref{eq:val-step2} and $X_7$ is its approximation \cref{eq:heuristic-expr3} that we heuristically computed in the previous subsection.
We will show that consecutive elements of this sequence are approximately equal.
We remind the reader here that we allow $q$ to take any value in $\supp\ov(\rho)$, unless otherwise specified. In particular, only the proof of \cref{prop:x5-x6} below will use the assumption $|q| \le n^{-1/2 + \alpha}$ from the first part of \cref{prop:planted-cluster-expansion-moment}.

For $\sigma,\tau$ with $\langle \sigma,\tau \rangle / n = q$ and $\Omega = (A_3, A_4, \ldots) \in \Ct(K)$, define
\begin{align*}
  U(q,\Omega)
  &= \prod_{k=3}^\infty
  \tr\wrapp{\wrapp{
    \frac{\tlambda^2}{n} \sigma\sigma^\top
    + \frac{|\tlambda|^2}{n} \tau\tau^\top
  }^k}^{A_k}, &
  V(\Omega)
  &= \prod_{k=3}^\infty
  \wrapp{\tlambda^{2k} + |\tlambda|^{2k}}^{A_k}.
\end{align*}
Note that $U(q,\Omega)$ is well-defined because it depends on $(\sigma,\tau)$ only through their overlap $q$.
Also,
\begin{align*}
  \overline{U(q,\Omega)}
  &= \prod_{k=3}^\infty
  \tr\wrapp{\wrapp{
    \frac{|\tlambda|^2}{n} \sigma\sigma^\top
    + \frac{\bar\tlambda^2}{n} \tau\tau^\top
  }^k}^{A_k}, &
  \overline{V(\Omega)}
  &= \prod_{k=3}^\infty
  \wrapp{|\tlambda|^{2k} + \bar\tlambda^{2k}}^{A_k}.
\end{align*}
Throughout this subsection, all expectations $\E$ are over $(\sigma,\tau) \sim \rho_2(q)$ unless denoted otherwise.
Define
\begin{align*}
  X_1 &= \sum_{\Gamma_0\in P(n)}
  \wrapp{\frac{|\tlambda|^2}{n}}^{|\Gamma_0|}
  \E
  \wrapb{
    \exp\wrapp{
      \kappa(\sigma,\tau)
    }
    q_1(\Gamma_0,\sigma,\tau)
    q_2(\Gamma_0,\sigma,\tau)
  } \\
  X_2 &= \sum_{\Gamma_0\in \UC(n)}
  \wrapp{\frac{|\tlambda|^2}{n}}^{|\Gamma_0|}
  \E
  \wrapb{
    \exp\wrapp{
      \kappa(\sigma,\tau)
    }
    q_1(\Gamma_0,\sigma,\tau)
    q_2(\Gamma_0,\sigma,\tau)
  } \\
  X_3 &= \sum_{\Gamma_0\in \UC(n)}
  \wrapp{\frac{|\tlambda|^2}{n}}^{|\Gamma_0|}
  \wrapp{
    \sum_{\Omega \in \Ct(K - |\Gamma_0|)}
    \frac{(-1)^{c(\Omega)} U(q,\Omega)}{\sym(\Omega)}
  }
  \wrapp{
    \sum_{\Omega \in \Ct(K - |\Gamma_0|)}
    \frac{(-1)^{c(\Omega)} \overline{U(q,\Omega)}}{\sym(\Omega)}
  } \\
  &\qquad \times
  \E
  \wrapb{
    \exp\wrapp{
      \kappa(\sigma,\tau)
    }
  } \\
  X_4 &= \sum_{\Gamma_0\in \UC(n)}
  \wrapp{\frac{|\tlambda|^2}{n}}^{|\Gamma_0|}
  \wrapp{
    \sum_{\Omega \in \Ct(K - |\Gamma_0|)}
    \frac{(-1)^{c(\Omega)} U(q,\Omega)}{\sym(\Omega)}
  }
  \wrapp{
    \sum_{\Omega \in \Ct(K - |\Gamma_0|)}
    \frac{(-1)^{c(\Omega)} \overline{U(q,\Omega)}}{\sym(\Omega)}
  } \\
  X_5 &= \sum_{\Omega_0\in \Ct(K)}
  \frac{|\tlambda|^{2|\Omega_0|}}{\sym(\Omega_0)}
  \wrapp{
    \sum_{\Omega \in \Ct(K - |\Omega_0|)}
    \frac{(-1)^{c(\Omega)} U(q,\Omega)}{\sym(\Omega)}
  }
  \wrapp{
    \sum_{\Omega \in \Ct(K - |\Omega_0|)}
    \frac{(-1)^{c(\Omega)} \overline{U(q,\Omega)}}{\sym(\Omega)}
  } \\
  X_6 &= \sum_{\Omega_0\in \Ct(K)}
  \frac{|\tlambda|^{2|\Omega_0|}}{\sym(\Omega_0)}
  \wrapp{
    \sum_{\Omega \in \Ct(K - |\Omega_0|)}
    \frac{(-1)^{c(\Omega)} V(\Omega)}{\sym(\Omega)}
  }
  \wrapp{
    \sum_{\Omega \in \Ct(K - |\Omega_0|)}
    \frac{(-1)^{c(\Omega)} \overline{V(\Omega)}}{\sym(\Omega)}
  } \\
  X_7 &= |1-\tlambda^2| (1-|\tlambda|^2)^{1/2} \exp\wrapp{
    \re(\tlambda^2) + \frac{1}{2} |\tlambda|^2 + \frac{1}{2} \re(\tlambda^4) + \frac{1}{4} |\tlambda|^4
  }.
\end{align*}
Note that $X_1,X_2,X_3$ implicitly depend on $q$ through $\E$, while $X_4,X_5$ depend on $q$ through $U(q,\Omega)$, and $X_6,X_7$ do not depend on $q$.
For use below, let
\[
  \mathcal{Q}_1(n) = \wrapc{
    \text{multi-sets consisting of two elements of $[n]$}
  }
\]
and
\[
  \mathcal{Q}_2(n) = \wrapc{
    \text{multi-sets consisting of elements of $\mathcal{Q}_1$}
  }.
\]
For example, $\{1,1\}$ and $\{2,3\}$ are elements of $\mathcal{Q}_1$, and $\{\{1,2\},\{1,2\},\{3,4\},\{5,5\}\}$ is an element of $\mathcal{Q}_2$.
We will use \cref{lem:error-term-bounding} below to bound several error terms in the approximation argument.
\begin{fact}[Spherical moments are dominated by gaussian moments; proved in \cref{subsec:sph-mmt}]
  \label{fac:sph-mmt}
  For any fixed $i\in [n]$ and integer $k\ge 0$, we have $\E_{\sigma\sim\unif(\mathcal{S}_n)} [|\sigma_i|^{2k}] \le (2k-1)!!$.
\end{fact}
\begin{lemma}
  \label{lem:error-term-bounding}
  Let $S \subseteq \mathcal{Q}_2$ be such that
  \[
    \max_{\Gamma \in S} |\Gamma| \le K.
  \]
  Then, for any $q \in \supp\ov(\rho)$,
  \begin{align}
    \label{eq:error-term-bounding}
    \E \wrapb{
      \sum_{\Gamma \in S}
      \prod_{\{i,j\} \in \Gamma}
      \abs{
        \frac{\tlambda^2}{n} \sigma_i\sigma_j
        + \frac{|\tlambda|^2}{n} \tau_i\tau_j
      }
    }^2
    \le (4K)^{2K} \wrapp{
      \sum_{\Gamma \in S}
      \wrapp{\frac{|\tlambda|^2}{n}}^{|\Gamma|}
    }^2.
  \end{align}
  The same estimate holds with $\frac{|\tlambda|^2}{n} \sigma_i\sigma_j + \frac{\bar\tlambda^2}{n} \tau_i\tau_j$ in place of $\frac{\tlambda^2}{n} \sigma_i\sigma_j + \frac{|\tlambda|^2}{n} \tau_i\tau_j$.
\end{lemma}
\begin{proof}
  Let $Y$ denote the left-hand side of \cref{eq:error-term-bounding}.
  Then, (with $\cup$ denoting multi-set union)
  \begin{align*}
    Y&\le \E \wrapb{
      \sum_{\Gamma \in S}
      \prod_{\{i,j\} \in \Gamma}
      \wrapp{
        \frac{|\tlambda|^2}{n} |\sigma_i||\sigma_j|
        + \frac{|\tlambda|^2}{n} |\tau_i||\tau_j|
      }
    }^2 \\
    &= \sum_{\Gamma_1,\Gamma_2 \in S}
    \wrapp{\frac{|\tlambda|^2}{n}}^{|\Gamma_1 \cup \Gamma_2|}
    \E \wrapb{
      \prod_{\{i,j\} \in \Gamma_1 \cup \Gamma_2}
      \wrapp{
        |\sigma_i||\sigma_j|
        + |\tau_i||\tau_j|
      }
    } \\
    &\le \sum_{\Gamma_1,\Gamma_2 \in S}
    \wrapp{\frac{|\tlambda|^2}{n}}^{|\Gamma_1 \cup \Gamma_2|}
    \prod_{\{i,j\} \in \Gamma_1 \cup \Gamma_2}
    \E \wrapb{
      \wrapp{
        |\sigma_i||\sigma_j|
        + |\tau_i||\tau_j|
      }^{|\Gamma_1 \cup \Gamma_2|}
    }^{\frac{1}{|\Gamma_1 \cup \Gamma_2|}},
  \end{align*}
  where the final estimate is by H\"older's inequality.
  Abbreviate $k = |\Gamma_1 \cup \Gamma_2|$, and note that $k\le 2K$ for all $\Gamma_1,\Gamma_2 \in S$.
  The inner expectation can be bounded by
  \begin{align*}
    2^{k-1}\E \wrapb{
      |\sigma_i|^k|\sigma_j|^k
      + |\tau_i|^k|\tau_j|^k
    }
    &= 2^k \E \wrapb{
      |\sigma_i|^k|\sigma_j|^k
    }
    \le 2^k \E [ |\sigma_i|^{2k} ] \\
    &\le 2^k \cdot (2k-1)!!
    \le (2k)^k
    \le (4K)^{2K}.
  \end{align*}
The first inequality is by the Cauchy--Schwarz inequality, combined with the fact that $\sigma_i$ and $\sigma_j$ have the same distribution; the second uses \cref{fac:sph-mmt} if we are in the spherical setting $\rho = \unif(\mathcal{S}_n)$ (and is trival in the Ising setting $\rho = \unif(C_n)$).
  Thus
  \[
    Y \le (4K)^{2K}
    \sum_{\Gamma_1,\Gamma_2 \in S}
    \wrapp{\frac{|\tlambda|^2}{n}}^{|\Gamma_1 \cup \Gamma_2|}
    = (4K)^{2K} \wrapp{
      \sum_{\Gamma \in S}
      \wrapp{\frac{|\tlambda|^2}{n}}^{|\Gamma|}
    }^2.
  \]
\end{proof}
\noindent We now proceed to bound the differences $|X_t - X_{t+1}|$.
\begin{proposition}
  \label{prop:x1-x2}
  We have $|X_1 - X_2| \le n^{-1} \cdot e^2 |\tlambda|^6 (1-|\tlambda|^2)^{-5} (2K)^{4K}$.
\end{proposition}
\begin{proof}
  Note that
  \[
    X_1 - X_2
    = \sum_{\Gamma_0\in P(n) \setminus \UC(n)}
    \wrapp{\frac{|\tlambda|^2}{n}}^{|\Gamma_0|}
    \E
    \wrapb{
      \exp\wrapp{
        \kappa(\sigma,\tau)
      }
      q_1(\Gamma_0,\sigma,\tau)
      q_2(\Gamma_0,\sigma,\tau)
    }
  \]
  By \cref{eq:kappa-bd}, we have deterministically $|\kappa(\sigma,\tau)| \le 2$.
  So,
  \begin{align}
    \nonumber
    |X_1 - X_2|
    &\le e^2 \sum_{\Gamma_0\in P(n) \setminus \UC(n)}
    \wrapp{\frac{|\tlambda|^2}{n}}^{|\Gamma_0|}
    \E
    \wrapb{
      |q_1(\Gamma_0,\sigma,\tau)|
      |q_2(\Gamma_0,\sigma,\tau)|
    } \\
    \label{eq:x1-x2-step1}
    &\le e^2 \sum_{\Gamma_0\in P(n) \setminus \UC(n)}
    \wrapp{\frac{|\tlambda|^2}{n}}^{|\Gamma_0|}
    \E
    [|q_1(\Gamma_0,\sigma,\tau)|^2]^{1/2}
    \E
    [|q_2(\Gamma_0,\sigma,\tau)|^2]^{1/2},
  \end{align}
  by the Cauchy-Schwarz inequality.
  Similarly to how each element of $\UC(n)$ can be described by its cycle counts, which are an element of $\Ct(K)$, we will describe each element of $P(n) \setminus \UC(n)$ by its cycle and path counts, which are an element of the following set:
  \[
    \Ct_1(K)
    = \wrapc{
      \hOmega = \wrapp{(A_3,A_4,\ldots),(B_1,B_2,\ldots)} :
      A_k, B_k \ge 0,
      \sum_{k=3}^\infty kA_k
      + \sum_{k=1}^\infty kB_k \le K,
      \sum_{k=1}^\infty B_k > 0
    }.
  \]
  For $\hOmega = ((A_3,A_4,\ldots),(B_1,B_2,\ldots)) \in \Ct_1(K)$, define
  \begin{align*}
    && |\hOmega| &= \sum_{k=3}^\infty kA_k + \sum_{k=1}^\infty kB_k, &
    V(\hOmega) &= \sum_{k=3}^\infty kA_k + \sum_{k=1}^\infty (k+1) B_k, &
    p(\hOmega) &= \sum_{k=1}^\infty B_k,
  \end{align*}
  and
  \[
    \sym(\hOmega) = \prod_{k=3}^\infty \wrapp{(2k)^{A_k} A_k!} \prod_{k=1}^\infty \wrapp{2^{B_k} B_k!}.
  \]
  Then define
  \[
    P(n,\hOmega) = \wrapc{
      \Gamma \in P(n) :
      \text{$\Gamma$ has $A_k$ cycles of length $k$ and $B_k$ paths with $k$ edges}
    }.
  \]
  For $\Gamma \in P(n,\hOmega)$, also let $p(\Gamma) = p(\hOmega)$ be the number of paths in $\Gamma$.

  Recall $\UC(n,\Gamma_0)$ defined in \cref{eq:uc-completion}.
  If $\Gamma_0$ contains $p$ paths with endpoints $v(1),\ldots,v(2p)$, each $\Lambda \in \UC(n,\Gamma_0)$ consists of $p$ paths whose endpoint set is also $\{v(1),\ldots,v(2p)\}$, and some cycles, which are all mutually vertex-disjoint and vertex-disjoint with $\Gamma_0$ (except at $v(1),\ldots,v(2p)$).
  Elements of $\UC(n,\Gamma_0)$ can be described by elements of $\Ct_2(K,p)$, which we now define.
  For $p\ge 0$, let
  \begin{align*}
    \PM(2p) &= \wrapc{\text{
      perfect matchings $\pi = \wrapc{
        \{\pi^1_1,\pi^1_2\},\ldots,
        \{\pi^p_1,\pi^p_2\}
      }$ on $\{1,\ldots,2p\}$}
    }, \\
    L(p) &= \wrapc{
      \vec \ell = (\ell_1,\ldots,\ell_p) \in Z_{\ge 1}^p
    },
  \end{align*}
  and
  \[
    \Ct_2(K,p) = \Ct(K) \times \PM(2p) \times L(p).
  \]
  For $\tOmega = (\Omega,\pi,\vec\ell) \in \Ct_2(K,p)$, with $\Omega = (A_3,A_4,\ldots)$, $\pi = \{\{\pi^1_1,\pi^1_2\},\ldots,\{\pi^p_1,\pi^p_2\}\}$, and $\vec\ell = (\ell_1,\ldots,\ell_p)$, let
  \begin{align*}
    |\tOmega| &= \sum_{k=3}^\infty kA_k + \sum_{k=1}^p \ell_k, &
    V(\tOmega) &= \sum_{k=3}^\infty kA_k + \sum_{k=1}^p (\ell_k-1) = |\tOmega| - p.
  \end{align*}
  Then, let
  \[
    \UC(n,\Gamma_0,\tOmega)
    = \wrapc{
      \Lambda \in \UC(n,\Gamma_0) :
      \begin{array}{l}
      \text{$\Lambda$ contains a path of length $\ell_k$ connecting} \\
      \text{$v(\pi^k_1)$ and $v(\pi^k_2)$, and $A_k$ cycles of length $k$}
      \end{array}
    }.
  \]
  Note that all $\Lambda \in \UC(n,\Gamma_0,\tOmega)$ have $|\tOmega|$ edges and $V(\tOmega)$ vertices.
  Returning to \cref{eq:x1-x2-step1}, we estimate for any $\Gamma_0 \in P(n)$
  \begin{align}
    \nonumber
    \E [|q_1(\Gamma_0,\sigma,\tau)|^2]
    &\le \wrapp{
      \sum_{\Lambda \in \UC(n,\Gamma_0)}
      \prod_{\{i,j\} \in \Lambda} \abs{
        \frac{\tlambda^2}{n} \sigma_i\sigma_j
        + \frac{|\tlambda|^2}{n} \tau_i\tau_j
      }
    }^2 \\
    \nonumber
    &\le (4K)^{2K} \wrapp{
      \sum_{\Lambda \in \UC(n,\Gamma_0)}
      \wrapp{\frac{|\tlambda|^2}{n}}^{|\Lambda|}
    }^2 \\
    \label{eq:x1-x2-step2}
    &= (4K)^{2K} \wrapp{
      \sum_{\tOmega \in \Ct_2(K,p)}
      |\UC(n,\Gamma_0,\tOmega)|
      \wrapp{\frac{|\tlambda|^2}{n}}^{|\tOmega|}
    }^2,
  \end{align}
  where the second inequality is by \cref{lem:error-term-bounding}.
  We can bound
  \[
    |\UC(n,\Gamma_0,\tOmega)|
    = \frac{(n-V(\Gamma_0)) \cdots (n-V(\Gamma_0)-V(\tOmega)+1)}{\sym(\Omega)}
    \le \frac{n^{V(\tOmega)}}{\sym(\Omega)},
  \]
  where we recall $\Omega$ is the first coordinate of $\tOmega$.
  Thus (since $V(\tOmega) = |\tOmega| - p$) the inner sum in \cref{eq:x1-x2-step2} is bounded by
  \begin{align}
    \label{eq:x1-x2-step3}
    n^{-p}
    \sum_{\tOmega \in \Ct_2(K,p)}
    \frac{|\tlambda|^{2|\tOmega|}}{\sym(\Omega)}
    \le n^{-p} |\PM(2p)|
    \sum_{\Omega \in \Ct(K)}
    \frac{|\tlambda|^{2|\Omega|}}{\sym(\Omega)}
    \cdot \sum_{\vec\ell \in L(p)}
    |\tlambda|^{2\|\vec\ell\|_1}.
  \end{align}
  Recall that $p = p(\Gamma_0)$ is the number of paths in $\Gamma_0$, which is clearly bounded by $K$.
  Thus,
  \[
    |\PM(2p)| = (2p-1)!! \le (2K-1)!! \le K^K,
  \]
  while the remaining two sums in \cref{eq:x1-x2-step3} are bounded by
  \begin{align}
    \label{eq:cycle-term-bound}
    \sum_{\Omega \in \Ct(K)}
    \frac{|\tlambda|^{2|\Omega|}}{\sym(\Omega)}
    \le \sum_{A_3,A_4,\ldots \ge 0}
    \prod_{k=3}^\infty
    \wrapp{\frac{|\tlambda|^{2k}}{2k}}^{A_k}\cdot \frac{1}{A_k!}
    = \exp\wrapp{\sum_{k=3}^\infty \frac{|\tlambda|^{2k}}{2k}}
    \le (1-|\tlambda|^2)^{-1/2}
  \end{align}
  and
  \[
    \sum_{\vec\ell \in L(p)}
    |\tlambda|^{2\|\vec\ell\|_1}
    = \wrapp{
      \sum_{\ell \ge 1}
      |\tlambda|^{2\ell}
    }^p
    = \wrapp{
      \frac{|\tlambda|^2}{1-|\tlambda|^2}
    }^p.
  \]
  Altogether the inner sum of \cref{eq:x1-x2-step2} is bounded by
  \[
    K^K
    (1-|\tlambda|^2)^{-1/2}
    \wrapp{
      \frac{|\tlambda|^2}{n(1-|\tlambda|^2)}
    }^p,
  \]
  which implies
  \[
    \E [|q_1(\Gamma_0,\sigma,\tau)|^2]
    \le (2K)^{4K}
    (1-|\tlambda|^2)^{-1}
    \wrapp{
      \frac{|\tlambda|^2}{n(1-|\tlambda|^2)}
    }^{2p}.
  \]
  By an analogous argument, we can bound $\E [|q_2(\Gamma_0,\sigma,\tau)|^2]$ by the same quantity.
  Plugging back into \cref{eq:x1-x2-step1} yields
  \begin{align}
    \nonumber
    |X_1-X_2|
    &\le
    e^2 (1-|\tlambda|^2)^{-1} (2K)^{4K}
    \sum_{\Gamma_0\in P(n) \setminus \UC(n)}
    \wrapp{\frac{|\tlambda|^2}{n}}^{|\Gamma_0|}
    \wrapp{
      \frac{|\tlambda|^2}{n(1-|\tlambda|^2)}
    }^{2p(\Gamma_0)} \\
    \label{eq:x1-x2-step4}
    &= e^2 (1-|\tlambda|^2)^{-1} (2K)^{4K}
    \sum_{\hOmega \in \Ct_1(K)}
    |P(n,\hOmega)|
    \wrapp{\frac{|\tlambda|^2}{n}}^{|\hOmega|}
    \wrapp{
      \frac{|\tlambda|^2}{n(1-|\tlambda|^2)}
    }^{2p(\hOmega)}
  \end{align}
  Note that
  \[
    |P(n,\hOmega)|
    = \frac{n(n-1)\cdots(n-V(\hOmega)+1)}{\sym(\hOmega)}
    \le \frac{n^{V(\hOmega)}}{\sym(\hOmega)}.
  \]
  Since $V(\hOmega) = |\hOmega| + p(\hOmega)$, the sum in \cref{eq:x1-x2-step4} is bounded by
  \begin{align*}
    &\sum_{\hOmega \in \Ct_1(K)}
    \frac{n^{p(\hOmega)} |\tlambda|^{2|\hOmega|}}{\sym(\hOmega)}
    \wrapp{
      \frac{|\tlambda|^2}{n(1-|\tlambda|^2)}
    }^{2p(\hOmega)} \\
    &\le \sum_{A_3,A_4,\ldots \ge 0}
    \prod_{k=3}^\infty
    \wrapp{\frac{|\tlambda|^{2k}}{2k}}^{A_k}
    \cdot \frac{1}{A_k!}
    \cdot \sum_{\substack{B_1,B_2,\ldots \ge 0 \\ \text{not all $0$}}}
    \prod_{k=1}^\infty
    \wrapp{\frac{|\tlambda|^{2k}}{2} \cdot \frac{|\tlambda|^4}{n(1-|\tlambda|^2)^2}}^{B_k} \cdot \frac{1}{B_k!}.
  \end{align*}
  The first sum is bounded by
  \[
    \exp\wrapp{
      \sum_{k=3}^\infty
      \frac{|\tlambda|^{2k}}{2k}
    }
    \le (1-|\tlambda|^2)^{-1/2},
  \]
  while the second is bounded by
  \[
    \exp\wrapp{
      \sum_{k=1}^\infty
      \frac{|\tlambda|^{2k}}{2} \cdot \frac{|\tlambda|^4}{n(1-|\tlambda|^2)^2}
    } - 1
    = \exp\wrapp{
      \frac{|\tlambda|^6}{2n(1-|\tlambda|^2)^3}
    } - 1
    \le \frac{|\tlambda|^6}{n(1-|\tlambda|^2)^3}.
  \]
  In conclusion,
  \[
    |X_1-X_2|
    \le e^2 (1-|\tlambda|^2)^{-1} (2K)^{4K} \cdot (1-|\tlambda|^2)^{-1/2} \cdot \frac{|\tlambda|^6}{n(1-|\tlambda|^2)^3}
    \le \frac{e^2 |\tlambda|^6 (2K)^{4K}}{n(1-|\tlambda|^2)^5}.
  \]
\end{proof}
\begin{lemma}
  \label{lem:Z-Gamma}
  For $\Gamma_0 \in \UC(n)$, define
  \begin{align}
    \label{eq:def-Z-Gamma}
    Z(\Gamma_0) = \sum_{\Omega \in \Ct(K-|\Gamma_0|)}
    \frac{(-1)^{c(\Omega)} U(q,\Omega)}{\sym(\Omega)}.
  \end{align}
  Then, for all $q\in [-1,1]$, $|Z(\Gamma_0)| \le 4^K (1-|\tlambda|^2)^{-1/2}$.
\end{lemma}
\begin{proof}
  Note that
  \[
    \norm{
      \frac{\tlambda^2}{n} \sigma\sigma^\top
      + \frac{|\tlambda|^2}{n} \tau\tau^\top
    }_{\op}
    \le 2|\tlambda|^2.
  \]
  Thus (crudely)
  \[
    |U(q,\Omega)|
    \le \prod_{k=3}^\infty
    \wrapp{2\cdot (2|\tlambda|^2)^k}^{A_k}
    \le (2|\tlambda|)^{2|\Omega|}
    \le 4^K \cdot |\tlambda|^{2|\Omega|}.
  \]
  So,
  \[
    |Z(\Gamma_0)|
    \le 4^K \cdot \sum_{\Omega \in \Ct(K-|\Gamma_0|)}
    \frac{|\tlambda|^{2|\Omega|}}{\sym(\Omega)}.
  \]
  Identically to \cref{eq:cycle-term-bound}, the inner sum is bounded by $(1-|\tlambda|^2)^{-1/2}$.
\end{proof}
\begin{proposition}
  \label{prop:x2-x3}
  We have $|X_2-X_3| \le n^{-1} \cdot 3e^2 (4K)^{2K} K^4 (1-|\tlambda|^2)^{-2}$.
\end{proposition}
\begin{proof}
  For any $\Gamma_0 \in \UC(n)$, we can write
  \begin{align*}
    q_1(\Gamma_0,\sigma,\tau)
    &=
    (-1)^{c(\Gamma_0)}
    \sum_{\Lambda \in \UC(n,\Gamma_0)}
    (-1)^{c(\Lambda)}
    \prod_{\{i,j\} \in \Lambda} \wrapp{
      \frac{\tlambda^2}{n} \sigma_i\sigma_j
      + \frac{|\tlambda|^2}{n} \tau_i\tau_j
    } \\
    &=
    (-1)^{c(\Gamma_0)}
    \sum_{\Omega \in \Ct(K-|\Gamma_0|)}
    (-1)^{c(\Omega)}
    \sum_{\Lambda \in \UC(n,\Gamma_0,\Omega)}
    \prod_{\{i,j\} \in \Lambda} \wrapp{
      \frac{\tlambda^2}{n} \sigma_i\sigma_j
      + \frac{|\tlambda|^2}{n} \tau_i\tau_j
    }.
  \end{align*}
  For $\Omega \in \Ct(K-|\Gamma_0|)$, let
  \[
    T(\Omega) = [n]^{|\Omega|},
  \]
  and name the entries of $\vec t \in T(\Omega)$ by $\vec t = (t^k_{\ell,m} : k\ge 3, 1\le \ell \le A_k, 1\le m\le k)$.
  For such $\vec t \in T(\Omega)$, define the multiset
  \[
    E(\vec t)
    = \bigcup_{\substack{k\ge 3 \\ 1\le \ell \le A_k}}
    \wrapc{
      \{t^k_{\ell,1},t^k_{\ell,2}\},
      \{t^k_{\ell,2},t^k_{\ell,3}\},\ldots
      \{t^k_{\ell,k},t^k_{\ell,1}\}
    }.
  \]
  These index unions of cycles with vertices in $[n]$, which may repeat vertices or contain self-loops.
  Let $T_+(\Omega) \subseteq J(\Omega)$ consist of the $\vec t$ with disjoint entries and no entries in common with vertices of $\Gamma_0$, and $T_-(\Omega) = T(\Omega) \setminus T_+(\Omega)$.
  Then note that
  \[
    \sum_{\Lambda \in \UC(n,\Gamma_0,\Omega)}
    \prod_{\{i,j\} \in \Lambda} \wrapp{
      \frac{\tlambda^2}{n} \sigma_i\sigma_j
      + \frac{|\tlambda|^2}{n} \tau_i\tau_j
    }
    = \frac{1}{\sym(\Omega)} \sum_{\vec t \in T_+(\Omega)}
    \prod_{\{i,j\} \in E(\vec t)} \wrapp{
      \frac{\tlambda^2}{n} \sigma_i\sigma_j
      + \frac{|\tlambda|^2}{n} \tau_i\tau_j
    },
  \]
  while
  \[
    U(q,\Omega)
    = \sum_{\vec t \in T(\Omega)}
    \prod_{\{i,j\} \in E(\vec t)} \wrapp{
      \frac{\tlambda^2}{n} \sigma_i\sigma_j
      + \frac{|\tlambda|^2}{n} \tau_i\tau_j
    }.
  \]
  Define
  \begin{align*}
    \Delta_1(\Omega) &= \sum_{\vec t \in T_-(\Omega)}
    \prod_{\{i,j\} \in E(\vec t)} \wrapp{
      \frac{\tlambda^2}{n} \sigma_i\sigma_j
      + \frac{|\tlambda|^2}{n} \tau_i\tau_j
    }, &
    W_1(\Gamma_0) &= \sum_{\Omega \in \Ct(K-|\Gamma_0|)}
    \frac{(-1)^{c(\Omega)} \Delta_1(\Omega)}{\sym(\Omega)}.
  \end{align*}
  Then,
  \[
    \sum_{\Lambda \in \UC(n,\Gamma_0,\Omega)}
    \prod_{\{i,j\} \in \Lambda} \wrapp{
      \frac{\tlambda^2}{n} \sigma_i\sigma_j
      + \frac{|\tlambda|^2}{n} \tau_i\tau_j
    }
    = \frac{1}{\sym(\Omega)} (U(q,\Omega) - \Delta_1(\Omega)).
  \]
  So, for $Z(\Gamma_0)$ defined in \cref{eq:def-Z-Gamma},
  \[
    q_1(\Gamma_0,\sigma,\tau)
    =
    (-1)^{c(\Gamma_0)}
    \sum_{\Omega \in \Ct(K-|\Gamma_0|)}
    \frac{(-1)^{c(\Omega)}}{\sym(\Omega)}
    (U(q,\Omega) - \Delta_1(\Omega))
    = (-1)^{c(\Gamma_0)} (Z(\Gamma_0) - W_1(\Gamma_0)),
  \]
  Analogously, define
  \begin{align*}
    \Delta_2(\Omega) &= \sum_{\vec t \in T_-(\Omega)}
    \prod_{\{i,j\} \in E(\vec t)} \wrapp{
      \frac{|\tlambda|^2}{n} \sigma_i\sigma_j
      + \frac{\bar\tlambda^2}{n} \tau_i\tau_j
    }, &
    W_2(\Gamma_0) &= \sum_{\Omega \in \Ct(K-|\Gamma_0|)}
    \frac{(-1)^{c(\Omega)} \Delta_2(\Omega)}{\sym(\Omega)}.
  \end{align*}
  An identical argument shows
  \[
    q_2(\Gamma_0,\sigma,\tau)
    =
    (-1)^{c(\Gamma_0)} (\overline{Z(\Gamma_0)} - W_2(\Gamma_0)).
  \]
  Thus,
  \begin{align}
    \nonumber
    &|X_2-X_3| \\
    \nonumber
    &\le \sum_{\Gamma_0\in \UC(n)}
    \wrapp{\frac{|\tlambda|^2}{n}}^{|\Gamma_0|}
    \E
    \wrapb{
      \exp\wrapp{
        \kappa(\sigma,\tau)
      }
      |(Z(\Gamma_0)-W_1(\Gamma_0))(\overline{Z(\Gamma_0)}-W_2(\Gamma_0)) - Z(\Gamma_0)\overline{Z(\Gamma_0)}|
    } \\
    \nonumber
    &\le e^2 \sum_{\Gamma_0\in \UC(n)}
    \wrapp{\frac{|\tlambda|^2}{n}}^{|\Gamma_0|}
    \E
    \wrapb{
      |Z(\Gamma_0)||W_2(\Gamma_0)| + |Z(\Gamma_0)||W_1(\Gamma_0)| + |W_1(\Gamma_0)||W_2(\Gamma_0)|
    } \\
    \nonumber
    &\le e^2 \sum_{\Gamma_0\in \UC(n)}
    \wrapp{\frac{|\tlambda|^2}{n}}^{|\Gamma_0|}
    \Big[
      |Z(\Gamma_0)|\E[|W_2(\Gamma_0)|^2]^{1/2}
      + |Z(\Gamma_0)| \E[|W_1(\Gamma_0)|^2]^{1/2} \\
      \nonumber
      &\qquad\qquad\qquad\qquad\qquad
      + \E[|W_1(\Gamma_0)|^2]^{1/2} \E[|W_2(\Gamma_0)|^2]^{1/2}
    \Big] \\
    \label{eq:x2-x3-step1}
    &= e^2 \sum_{\Gamma_0\in \UC(n)}
    \wrapp{\frac{|\tlambda|^2}{n}}^{|\Gamma_0|}
    \wrapb{
      2|Z(\Gamma_0)| \E[|W_1(\Gamma_0)|^2]^{1/2}
      + \E[|W_1(\Gamma_0)|^2]
    }.
  \end{align}
  Here, we use that $Z(\Gamma_0)$ depends on $(\sigma,\tau)$ only through $q = \langle \sigma,\tau \rangle / n$, and thus does not depend on the realization of $(\sigma,\tau) \sim \rho_2(q)$.
  The final equality is by symmetry.
  We then estimate
  \begin{align*}
    \E[|W_1(\Gamma_0)|^2]
    &\le \sum_{\Omega_1,\Omega_2 \in \Ct(K-|\Gamma_0|)}
    \frac{\E[|\Delta_1(\Omega_1)||\Delta_1(\Omega_2)|]}{\sym(\Omega_1)\sym(\Omega_2)} \\
    &\le \sum_{\Omega_1,\Omega_2 \in \Ct(K-|\Gamma_0|)}
    \frac{\E[|\Delta_1(\Omega_1)|^2]^{1/2} \E[|\Delta_1(\Omega_2)|^2]^{1/2}}{\sym(\Omega_1)\sym(\Omega_2)}
    = \wrapp{
      \sum_{\Omega \in \Ct(K-|\Gamma_0|)}
      \frac{\E[|\Delta_1(\Omega)|^2]^{1/2}}{\sym(\Omega)}
    }^2.
  \end{align*}
  Furthermore, by \cref{lem:error-term-bounding},
  \[
    \E[|\Delta_1(\Omega)|^2]
    \le (4K)^{2K} \wrapp{
      \sum_{\vec t \in T_-(\Omega)}
      \wrapp{\frac{|\tlambda|^2}{n}}^{|\Omega|}
    }^2
    = (4K)^{2K} |T_-(\Omega)|^2 \wrapp{\frac{|\tlambda|^2}{n}}^{2|\Omega|}.
  \]
  However, we have
  \begin{align*}
    |T_-(\Omega)|
    = |T(\Omega)| - |T_+(\Omega)|
    &= n^{|\Omega|} - (n - |\Gamma_0|)\cdots (n - |\Gamma_0|-|\Omega|+1) \\
    &\le n^{|\Omega|-1} \sum_{k=1}^{|\Gamma_0|+|\Omega|-1} k
    \le K^2 n^{|\Omega|-1},
  \end{align*}
  where we use that $|\Gamma_0|+|\Omega| \le K$.
  Thus,
  \[
    \E[|\Delta_1(\Omega)|^2]
    \le n^{-2} (4K)^{2K} K^4 |\tlambda|^{4|\Omega|},
  \]
  which implies
  \[
    \E[|W_1(\Gamma_0)|^2]
    \le n^{-2} (4K)^{2K} K^4 \wrapp{
      \sum_{\Omega \in \Ct(K-|\Gamma_0|)}
      \frac{|\tlambda|^{2|\Omega|}}{\sym(\Omega)}
    }^2.
  \]
  Identically to \cref{eq:cycle-term-bound}, the inner sum is bounded by $(1-|\tlambda|^2)^{-1/2}$.
  So,
  \[
    \E[|W_1(\Gamma_0)|^2]
    \le n^{-2} (4K)^{2K} K^4 (1-|\tlambda|^2)^{-1}.
  \]
  Combining with \cref{lem:Z-Gamma} and plugging into \cref{eq:x2-x3-step1} yields
  \begin{align}
    \nonumber
    |X_2-X_3|
    &\le e^2 \wrapp{
      2n^{-1} (16K)^{K} K^2 (1-|\tlambda|^2)^{-1}
      + n^{-2} (4K)^{2K} K^4 (1-|\tlambda|^2)^{-1}
    }
    \sum_{\Gamma_0\in \UC(n)}
    \wrapp{\frac{|\tlambda|^2}{n}}^{|\Gamma_0|} \\
    \label{eq:x2-x3-step2}
    &\le 3e^2 n^{-1} (4K)^{2K} K^4 (1-|\tlambda|^2)^{-1}
    \sum_{\Gamma_0\in \UC(n)}
    \wrapp{\frac{|\tlambda|^2}{n}}^{|\Gamma_0|}.
  \end{align}
  This final sum is bounded by
  \[
    \sum_{\Omega \in \Ct(K)}
    |\UC(n,\Omega)|
    \wrapp{\frac{|\tlambda|^2}{n}}^{|\Omega|}
    \le \sum_{\Omega \in \Ct(K)}
    \frac{|\tlambda|^{2|\Omega|}}{\sym(\Omega)}
    \le (1-|\tlambda|^2)^{-1/2},
  \]
  where the last estimate is identical to \cref{eq:cycle-term-bound}.
  The result follows.
\end{proof}
\begin{lemma}
  \label{lem:kappa-estimate}
  We have $|\E[\exp(\kappa(\sigma,\tau))] - 1| \le 4n^{-1} \log^4 n$ for sufficiently large $n$.
\end{lemma}

\begin{proof}
  In the Ising setting this is trivial from \cref{eq:kappa-bd}, so assume $\sigma,\tau$ are on the sphere.
  For any fixed $i\in [n]$, it is standard that $\mathbb{P}(|\sigma_i| \ge \log n) \le e^{-\Omega(\log^2 n)}$.
  Let $\mathcal{E}$ be the event that
  \[
    \max_{1\le i\le n} \max(|\sigma_i|,|\tau_i|) \le \log n.
  \]
  By a union bound, $\mathbb{P}(\mathcal{E}^c) \le e^{-\Omega(\log^2 n)}$.
  We estimate
  \begin{align*}
    |\E[\exp(\kappa(\sigma,\tau))] - 1|
    &\le \E[|\exp(\kappa(\sigma,\tau))- 1|] \\
    &= \E[\indic\{\mathcal{E}\}|\exp(\kappa(\sigma,\tau))- 1|]
    + \E[\indic\{\mathcal{E}^c\}|\exp(\kappa(\sigma,\tau))- 1|].
  \end{align*}
  On $\mathcal{E}$, we have
  \[
    \kappa(\sigma,\tau)
    \stackrel{\eqref{eq:kappa-bd}}{\le}
    \frac{1}{n^2}(\|\sigma\|_4^4 + \|\tau\|_4^4)
    \le 2n^{-1} \log^4 n.
  \]
  So,
  \[
    \E[\indic\{\mathcal{E}\}|\exp(\kappa(\sigma,\tau))- 1|] \le 3n^{-1} \log^4 n.
  \]
  Furthermore, since $|\kappa(\sigma,\tau)|$ is bounded deterministically by $2$,
  \[
    \E[\indic\{\mathcal{E}^c\}|\exp(\kappa(\sigma,\tau))- 1|]
    \le (e^2+1) \mathbb{P}(\mathcal{E}^c)
    \le e^{-\Omega(\log^2 n)}
    \le n^{-1} \log^4 n.
  \]
\end{proof}
\begin{proposition}
  \label{prop:x4}
  We have $|X_4| \le 2^{4K} (1-|\tlambda|^2)^{-2}$.
\end{proposition}
\begin{proof}
  For $Z(\Gamma_0)$ defined in \cref{eq:def-Z-Gamma},
  \[
    X_4 = \sum_{\Gamma_0 \in \UC(n)}
    \wrapp{\frac{|\tlambda|^2}{n}}^{|\Gamma_0|}
    |Z(\Gamma_0)|^2.
  \]
  By \cref{lem:Z-Gamma},
  \[
    |X_4| \le 2^{4K} (1-|\tlambda|^2)^{-1}
    \sum_{\Gamma_0 \in \UC(n)}
    \wrapp{\frac{|\tlambda|^2}{n}}^{|\Gamma_0|}.
  \]
  The last term is bounded by $(1-|\tlambda|^2)^{-1/2} \le (1-|\tlambda|^2)^{-1}$ identically to \cref{eq:cycle-term-bound}.
\end{proof}
\begin{proposition}
  \label{prop:x3-x4}
  We have $|X_3 - X_4| \le n^{-1} \log^4 n \cdot 2^{4K+2} (1-|\tlambda|^2)^{-2}$.
\end{proposition}
\begin{proof}
  Immediate from \cref{lem:kappa-estimate} and \cref{prop:x4}.
\end{proof}
At this point we are ready to prove the part of \cref{prop:planted-cluster-expansion-moment} pertaining to general overlaps $q$.
\begin{proof}[Proof of \cref{eq:planted-cluster-expansion-moment-not-orthogonal} of \cref{prop:planted-cluster-expansion-moment}]
  We can estimate
  \[
    |X_1|
    \le |X_1-X_2|
    + |X_2-X_3|
    + |X_3-X_4|
    + |X_4|.
  \]
  These are bounded by \cref{prop:x1-x2,prop:x2-x3,prop:x3-x4,prop:x4}.
  In particular, the bound on $|X_4|$ from \cref{prop:x4} gives the main contribution, while the remaining terms are of lower order.
  So,
  \[
    |X_1| \le 2(1-|\tlambda|^2)^{-2} \cdot 2^{4K}.
  \]
  Recall that $X_1$ is the combinatorial sum in \cref{eq:val-step2}
  Thus,
  \begin{align*}
    |\Val(q,\lambda)|
    &\le
    |1-\lambda^2|^{-1}
    \exp\wrapp{
      - \re(\lambda^2)
      - \frac{1}{2} \re(\lambda^4)
      - \frac{1}{2} |\lambda|^2
      - \frac{1}{4} |\lambda|^4
      - \re(\lambda^2)|\lambda|^2 q^2
      + O(n^{-1})
    } \\
    &\qquad \times
    2(1-|\tlambda|^2)^{-2} \cdot 2^{4K}.
  \end{align*}
  Recall from \cref{eq:def-tchi} that $|\tlambda - \lambda| = O(n^{-1})$.
  Since $\lambda \in \overline{\D(0,1-\frac{\varepsilon}{2})}$, the last display is bounded by $C(\varepsilon) \cdot 2^{4K}$ for some $C(\varepsilon)$ depending only on $\varepsilon$.
\end{proof}
\begin{proposition}
  \label{prop:x4-x5}
  We have $|X_4-X_5| \le n^{-1} \cdot 2^{4K} K^2 (1-|\tlambda|^2)^{-2}$.
\end{proposition}
\begin{proof}
  Let $Z(\Gamma_0)$ is defined in \cref{eq:def-Z-Gamma}.
  Note that this quantity depends on $\Gamma_0$ only through $|\Gamma_0|$, so we may denote it $Z(|\Gamma_0|)$, i.e.
  \[
    Z(k) = \sum_{\Omega \in \Ct(K-k)}
    \frac{(-1)^{c(\Omega)} U(q,\Omega)}{\sym(\Omega)}.
  \]
  We can write
  \begin{align*}
    X_4 = \sum_{\Gamma_0\in \UC(n)}
    \wrapp{\frac{|\tlambda|^2}{n}}^{|\Gamma_0|}
    |Z(|\Gamma_0|)|^2
    &= \sum_{\Omega_0\in \Ct(K)}
    \frac{|\UC(n,\Omega_0)|}{n^{|\Omega_0|}}
    |\tlambda|^{2|\Omega_0|}
    |Z(|\Omega_0|)|^2 \\
    X_5 &= \sum_{\Omega_0\in \Ct(K)}
    \frac{|\tlambda|^{2|\Omega_0|}}{\sym(\Omega_0)}
    |Z(|\Omega_0|)|^2,
  \end{align*}
  Note that for any $\Omega_0 \in \Ct(K)$,
  \[
    |\UC(n,\Omega_0)|
    = \frac{n(n-1)\cdots(n-|\Omega_0|+1)}{\sym(\Omega_0)}.
  \]
  So,
  \[
    \abs{\frac{|\UC(n,\Omega_0)|}{n^{|\Omega_0|}} - \frac{1}{\sym(\Omega_0)}}
    \le \frac{(1 + \cdots + (|\Omega_0|-1))}{n\sym(\Omega_0)}
    \le \frac{K^2}{n\sym(\Omega_0)}.
  \]
  Thus,
  \begin{align*}
    |X_4-X_5|
    &\le \frac{K^2}{n}
    \sum_{\Omega_0\in \Ct(K)}
    \frac{|\tlambda|^{2|\Omega_0|}}{\sym(\Omega_0)}
    |Z(|\Omega_0|)|^2
    \stackrel{Lem.~\ref{lem:Z-Gamma}}{\le}
    \frac{2^{4K} K^2 (1-|\tlambda|^2)^{-1}}{n}
    \sum_{\Omega_0\in \Ct(K)}
    \frac{|\tlambda|^{2|\Omega_0|}}{\sym(\Omega_0)} \\
    &\!\!\stackrel{\eqref{eq:cycle-term-bound}}{\le}
    \frac{2^{4K} K^2 (1-|\tlambda|^2)^{-2}}{n}.
  \end{align*}
\end{proof}
\begin{proposition}
  \label{prop:x5-x6}
  If $|q| \le n^{-1/2 + \alpha}$, we have $|X_5-X_6| \le 3n^{-1+2\alpha} \cdot 2^{8K} (1-|\tlambda|^2)^{-2}$.
\end{proposition}
\begin{proof}
  For $k\ge 3$, let $E_k = \tlambda^{2k} + |\tlambda|^{2k}$ and, for $\sigma,\tau \in \mathcal{S}_n$ with $\langle \sigma,\tau \rangle / n = q$,
  \begin{align*}
    \Delta_k &= \tr\wrapp{\wrapp{
      \frac{\tlambda^2}{n} \sigma\sigma^\top
      + \frac{|\tlambda|^2}{n} \tau\tau^\top
    }^k} - E_k \\
    &= \tr\wrapp{
      \wrapp{
        \frac{\tlambda^2}{n} \sigma\sigma^\top
        + \frac{|\tlambda|^2}{n} \tau\tau^\top
      }^k
      - \wrapp{
        \frac{\tlambda^2}{n} \sigma\sigma^\top
      }^k
      - \wrapp{
        \frac{|\tlambda|^2}{n} \tau\tau^\top
      }^k
    }.
  \end{align*}
  This can be expanded as the sum of $2^k-2$ terms, each of which is bounded in absolute value by
  \[
    |\tlambda|^{2k} \cdot q^2 \le |\tlambda|^{2k} \cdot n^{-1+2\alpha}.
  \]
  Thus, for all $k\le K$,
  \[
    |\Delta_k| \le |\tlambda|^{2k} \cdot n^{-1+2\alpha} \cdot (2^k-2) \le n^{-1+2\alpha} \cdot 2^K.
  \]
  The estimate $|\tlambda|\le 1$ also implies
  \[
    |E_k| \le 2.
  \]
  Then, for $\Omega = (A_3,A_4,\ldots) \in \Ct(K)$,
  \begin{align*}
    U(q,\Omega) &= \prod_{k=3}^\infty (E_k + \Delta_k)^{A_k}, &
    V(\Omega) &= \prod_{k=3}^\infty (E_k)^{A_k}.
  \end{align*}
  Let $T = \sum_{k=3}^\infty A_k$, and note $T \le |\Omega| \le K$.
  Then, $U(q,\Omega) - V(\Omega)$ consists of $2^T-1$ terms, each of which is bounded in absolute value by
  \[
    |\tlambda|^{2|\Omega|} \cdot
    2^{T-1} \cdot (n^{-1+2\alpha} \cdot 2^K)
    \le |\tlambda|^{2|\Omega|} \cdot
    n^{-1+2\alpha} \cdot 2^{2K}.
  \]
  Thus
  \[
    |U(q,\Omega) - V(\Omega)|
    \le (2^T-1) \cdot |\tlambda|^{2|\Omega|} \cdot n^{-1+2\alpha} \cdot 2^{2K}
    \le |\tlambda|^{2|\Omega|} \cdot n^{-1+2\alpha} \cdot 2^{3K}.
  \]
  Let
  \[
    Y(k)
    = \sum_{\Omega \in \Ct(K - k)}
    \frac{(-1)^{c(\Omega)} V(\Omega)}{\sym(\Omega)}.
  \]
  Note that for $k\le K$,
  \begin{align*}
    |Y(k) - Z(k)|
    &\le \sum_{\Omega \in \Ct(K - k)}
    \frac{|U(q,\Omega) - V(\Omega)|}{\sym(\Omega)} \\
    &\le n^{-1+2\alpha} \cdot 2^{3K}
    \sum_{\Omega \in \Ct(K - k)}
    \frac{|\tlambda|^{2|\Omega|}}{\sym(\Omega)}
    \stackrel{\eqref{eq:cycle-term-bound}}{\le}
    n^{-1+2\alpha} \cdot 2^{3K} (1-|\tlambda|^2)^{-1/2}.
  \end{align*}
  Thus,
  \begin{align*}
    \abs{
      |Z(k)|^2 - |Y(k)|^2
    }
    &= \abs{
      Z(k)\overline{Z(k)} -
      \wrapp{Z(k) + (Y(k) - Z(k))}
      \wrapp{\overline{Z(k)} + (\overline{Y(k) - Z(k)})}
    } \\
    &\le 2|Z(k)||Y(k)-Z(k)|
    + |Y(k)-Z(k)|^2 \\
    &\!\!\!\!\!\!\stackrel{Lem.~\ref{lem:Z-Gamma}}{\le}
    2n^{-1+2\alpha} \cdot 2^{5K+1} (1-|\tlambda|^2)^{-1}
    + n^{-2} \log^4 n \cdot 2^{8K} (1-|\tlambda|^2)^{-1} \\
    &\le
    3n^{-1+2\alpha} \cdot 2^{8K} (1-|\tlambda|^2)^{-1}.
  \end{align*}
  Finally,
  \begin{align*}
    X_5 &= \sum_{\Omega_0\in \Ct(K)}
    \frac{|\tlambda|^{2|\Omega_0|}}{\sym(\Omega_0)}
    |Z(|\Omega_0|)|^2, &
    X_6 &= \sum_{\Omega_0\in \Ct(K)}
    \frac{|\tlambda|^{2|\Omega_0|}}{\sym(\Omega_0)}
    |Y(|\Omega_0|)|^2,
  \end{align*}
  so
  \[
    |X_5-X_6|
    \le 3n^{-1+2\alpha} \cdot 2^{8K} (1-|\tlambda|^2)^{-1} \cdot
    \sum_{\Omega_0\in \Ct(K)}
    \frac{|\tlambda|^{2|\Omega_0|}}{\sym(\Omega_0)}
    \stackrel{\eqref{eq:cycle-term-bound}}{\le}
    3n^{-1+2\alpha} \cdot 2^{8K} (1-|\tlambda|^2)^{-2}.
  \]
\end{proof}
\begin{proposition}
  \label{prop:x6-x7}
  We have $|X_6-X_7| = o_n(1)$.
\end{proposition}
\begin{proof}
  Let
  \[
    \widetilde X_7
    =
    \sum_{\Omega_0,\Omega_1,\Omega_2 \in \Ct(\infty)}
    \frac{|\tlambda|^{2|\Omega_0|}}{\sym(\Omega_0)}
    \cdot
    \frac{(-1)^{c(\Omega_1)} V(\Omega_1)}{\sym(\Omega_1)}
    \cdot
    \frac{(-1)^{c(\Omega_2)} \overline{V(\Omega_2)}}{\sym(\Omega_2)}.
  \]
  Since $X_6$ consists of a sub-sum of this infinite sum, which grows to include all terms as $K\to\infty$, it suffices to show this infinite sum equals $X_7$ and converges absolutely.
  The absolute sum of $\widetilde X_7$ is
  \[
    \sum_{\Omega_0,\Omega_1,\Omega_2 \in \Ct(\infty)}
    \frac{|\tlambda|^{2|\Omega_0|}}{\sym(\Omega_0)}
    \cdot
    \frac{|V(\Omega_1)|}{\sym(\Omega_1)}
    \cdot
    \frac{|V(\Omega_2)|}{\sym(\Omega_2)}
    = \wrapp{
      \sum_{\Omega \in \Ct(\infty)}
      \frac{|\tlambda|^{2|\Omega|}}{\sym(\Omega)}
    }
    \wrapp{
      \sum_{\Omega \in \Ct(\infty)}
      \frac{|V(\Omega)|}{\sym(\Omega)}
    }^2.
  \]
  The first factor is bounded by $(1-|\tlambda|^2)^{-1/2}$ by \cref{eq:cycle-term-bound}.
  For the second,
  \[
    \sum_{\Omega \in \Ct(\infty)}
    \frac{|V(\Omega)|}{\sym(\Omega)}
    = \sum_{A_3,A_4,\ldots \ge 0}
    \prod_{k=3}^\infty \wrapp{\frac{\abs{\tlambda^{2k} + |\tlambda|^{2k}}}{2k}}^{A_k} \cdot \frac{1}{A_k!}
    \le \exp\wrapp{
      \sum_{k=3}^\infty
      \frac{|\tlambda|^{2k}}{k}
    }
    \le (1-|\tlambda|^2)^{-1}.
  \]
  Thus $\widetilde X_7$ is absolutely summable.
  We can now evaluate
  \[
    \widetilde X_7
    = \wrapp{
      \sum_{\Omega \in \Ct(\infty)}
      \frac{|\tlambda|^{2|\Omega|}}{\sym(\Omega)}
    }
    \abs{
      \sum_{\Omega \in \Ct(\infty)}
      \frac{(-1)^{c(\Omega)} V(\Omega)}{\sym(\Omega)}
    }^2.
  \]
  The first sum evaluates as
  \begin{align*}
    \sum_{\Omega \in \Ct(\infty)}
    \frac{|\tlambda|^{2|\Omega|}}{\sym(\Omega)}
    &= \sum_{A_3,A_4,\ldots \ge 0}
    \prod_{k=3}^\infty \wrapp{\frac{|\tlambda|^{2k}}{2k}}^{A_k} \cdot \frac{1}{A_k!}
    = \exp\wrapp{
      \sum_{k=3}^\infty
      \frac{|\tlambda|^{2k}}{2k}
    } \\
    &= (1-|\tlambda|^2)^{-1/2} \exp\wrapp{
      - \frac{|\tlambda|^2}{2}
      - \frac{|\tlambda|^4}{4}
    }.
  \end{align*}
  The second sum evaluates as
  \begin{align*}
    \sum_{\Omega \in \Ct(\infty)}
    \frac{(-1)^{c(\Omega)} V(\Omega)}{\sym(\Omega)}
    &= \sum_{A_3,A_4,\ldots \ge 0}
    \prod_{k=3}^\infty \wrapp{-\frac{\tlambda^{2k} + |\tlambda|^{2k}}{2k}}^{A_k} \cdot \frac{1}{A_k!}
    = \exp\wrapp{
      - \sum_{k=3}^\infty
      \frac{\tlambda^{2k} + |\tlambda|^{2k}}{2k}
    } \\
    &= (1-\tlambda^2)^{1/2} (1-|\tlambda|^2)^{1/2}
    \exp\wrapp{
      \frac{\tlambda^2 + |\tlambda|^2}{2}
      + \frac{\tlambda^4 + |\tlambda|^4}{4}
    }.
  \end{align*}
  Hence
  \[
    \abs{
      \sum_{\Omega \in \Ct(\infty)}
      \frac{(-1)^{c(\Omega)} V(\Omega)}{\sym(\Omega)}
    }^2
    = |1-\tlambda^2| (1-|\tlambda|^2)
    \exp\wrapp{
      \re(\tlambda^2) + |\tlambda|^2
      + \frac{\re(\tlambda^4) + |\tlambda|^4}{2}
    },
  \]
  and combining shows
  \[
    \widetilde X_7
    = |1-\tlambda^2| (1-|\tlambda|^2)^{1/2}
    \exp\wrapp{
      \re(\tlambda^2) + \frac{1}{2} |\tlambda|^2
      + \frac{1}{2} \re(\tlambda^4)
      + \frac{1}{4} |\tlambda|^4
    }
    = X_7.
  \]
\end{proof}
\begin{proof}[Proof of \cref{eq:planted-cluster-expansion-moment-orthogonal} of \cref{prop:planted-cluster-expansion-moment}]
  Since $K = \lfloor \log \log n \rfloor$, the bounds in \cref{prop:x1-x2,prop:x2-x3,prop:x3-x4,prop:x4-x5,prop:x5-x6} are all $o_n(1)$.
  Combining with \cref{prop:x6-x7}, we have
  \[
    |X_1 - X_7|
    \le \sum_{i=1}^6 |X_i-X_{i+1}|
    = o_n(1).
  \]
  Since \cref{eq:def-tchi} implies $|\tlambda - \lambda| = O(n^{-1})$, we further have
  \[
    \abs{
      X_1 - |1-\lambda^2| (1-|\lambda|^2)^{1/2} \exp\wrapp{
        \re(\lambda^2) + \frac{1}{2} |\lambda|^2 + \frac{1}{2} \re(\lambda^4) + \frac{1}{4} |\lambda|^4
      }
    } = o_n(1).
  \]
  Plugging this estimate in for the combinatorial sum in \cref{eq:val-step2} yields the result.
\end{proof}

\subsection{Bound on Spherical Moments}
\label{subsec:sph-mmt}

Finally, we present the deferred proof of \cref{fac:sph-mmt}.
We will use the following continuous version of the rearrangement inequality.
\begin{fact}
    \label{fac:rearrangement}
    Suppose $Z$ is a real-valued random variable and $f : \R \rightarrow \R$ is nonincreasing and measurable, with $\E [|Z|], \E[|f(Z)|], \E[|Zf(Z)|]< \infty$.
    Then $\E[Zf(Z)] \le \E[Z]\E[f(Z)]$.
\end{fact}
\begin{proof}
    Let $Z'$ be an independent copy of $Z$.
    Since $f$ is nonincreasing, we have for all $x,y \in \R$
    \[
        (x-y)(f(x)-f(y)) \le 0.
    \]
    Thus,
    \begin{align*}
        0 \ge \E[(Z-Z')(f(Z)-f(Z'))]
        &= \E[Zf(Z)] - \E[Zf(Z')] - \E[Z'f(Z)] + \E[Z'f(Z')] \\
        &= 2\wrapp{\E[Zf(Z)] - \E[Z]\E[f(Z)]}.
    \end{align*}
\end{proof}

\begin{lemma}[Even moments of spherical marginals are negatively correlated]
  \label{lem:sph-neg-corr}
  For all $k,\ell\ge 0$ and $\sigma \sim \unif(\mathcal{S}_n)$,
  \[
    \E[\sigma_1^{2k}\sigma_2^{2\ell}]
    \le \E[\sigma_1^{2k}]\E[\sigma_2^{2\ell}]
  \]
\end{lemma}
\begin{proof}
  Note that $f(x) = \E[\sigma_2^{2\ell} | \sigma_1^{2k} = x]$ is nonincreasing.
  The result follows from \cref{fac:rearrangement}:
    \[
        \E[\sigma_1^{2k}\sigma_2^{2\ell}]
        = \E[\sigma_1^{2k} f(\sigma_1^{2k})]
        \le \E[\sigma_1^{2k}] \E[f(\sigma_1^{2k})]
        = \E[\sigma_1^{2k}]\E[\sigma_2^{2\ell}].
    \]
\end{proof}

\begin{proof}[Proof of \cref{fac:sph-mmt}]
  Let $Z_1 \sim \mathcal{N}(0,1)$.
  It suffices to prove that
  \begin{align}
    \label{eq:sph-mmt-goal}
    \E[\sigma_1^{2k}] \le \E[Z_1^{2k}].
  \end{align}
  We proceed by induction on $k$, with the base case $k=1$ trivial.
  Suppose we have proved \cref{eq:sph-mmt-goal} up to $k-1$.
  By rotational invariance,
  \[
    \E[\sigma_1^{2k}]
    = \E\wrapb{
      \wrapp{\frac{\sigma_1+\sigma_2}{\sqrt{2}}}^{2k}
    }
    = 2^{-k}
    \sum_{\ell=0}^k
    \binom{2k}{2\ell}
    \E[\sigma_1^{2\ell}\sigma_2^{2(k-\ell)}].
  \]
  Hence,
  \[
    (1-2^{-k+1}) \E[\sigma_1^{2k}]
    = 2^{-k}
    \sum_{\ell=1}^{k-1}
    \binom{2k}{2\ell}
    \E[\sigma_1^{2\ell}\sigma_2^{2(k-\ell)}].
  \]
  An identical argument shows that for $Z_2 \sim \mathcal{N}(0,1)$ independent of $Z_1$,
  \[
    (1-2^{-k+1}) \E[Z_1^{2k}]
    = 2^{-k}
    \sum_{\ell=1}^{k-1}
    \binom{2k}{2\ell}
    \E[Z_1^{2\ell}Z_2^{2(k-\ell)}].
  \]
  However, by \cref{lem:sph-neg-corr} and the inductive hypothesis, for all $1\le \ell \le k-1$,
  \[
    \E[\sigma_1^{2\ell}\sigma_2^{2(k-\ell)}]
    \le \E[\sigma_1^{2\ell}]\E[\sigma_2^{2(k-\ell)}]
    \le \E[Z_1^{2\ell}]\E[Z_2^{2(k-\ell)}]
    = \E[Z_1^{2\ell} Z_2^{2(k-\ell)}].
  \]
  This completes the induction.
\end{proof}

\section{Algorithms}\label{sec:algos}

In this section, we describe our algorithms for estimating $Z_{\bm{G}}\wrapp{\beta^{}}$ for a given $\beta^{} \in \D(0,(1 - \varepsilon) \cdot \betasecond)$. Following Barvinok (see e.g. \cite{Bar16book}), the basic idea behind \cref{thm:alg-pspin-glass} is to truncate the Taylor series of $\beta \mapsto \log Z_{\bm{G}}(\beta)$. With high probability, this series converges rapidly on the entirety of $\D(0,(1-\varepsilon) \cdot \betasecond)$ by the zero-freeness guarantee from \cref{thm:zeros-pspin-glass}.

More formally, define
\begin{align*}
    f(z) \defeq \log Z_{\bm{G}}(z)
\end{align*}
and note that with high probability this is a well-defined holomorphic function on  $\D(0,(1 - \varepsilon/2) \cdot \betasecond)$ with a convergent Taylor series expansion:\footnote{Note that since $Z_{\bm{G}}(0) = 1$, one can always take a small enough radius $r_{\mathsf{triv}}$ (possibly depending on the realization $\bm{G}$) such that $f$ is well-defined and holomorphic on $\D(0,r_{\mathsf{triv}})$. This, in particular, implies that we can always speak about the Taylor series expansion of $f$ in a neighborhood of $z = 0$.}
\begin{align}\label{eq:taylor}
    f(z) = \sum_{k=0}^{\infty} \frac{f^{(k)}(0)}{k!} \cdot z^{k}.
\end{align}
Note that $f(\beta) = \log Z_{\bm{G}}(\beta)$, the value that we wish to approximate up to $\eta$-additive error, and $f(0) = 0$. The key idea behind Barvinok's interpolation method is that if $Z_{\bm{G}}(z)$ is nonzero on $\D(0,(1-\varepsilon/2)\betasecond)$, then we can truncate the series at suitable depth and plug in $z = \beta$ to obtain the desired approximation.

To get a good approximation we need to know where we should truncate the series and we need an algorithm to compute the coefficients of the series. These ingredients will be addressed in the next two subsections. These will then be combined to give a formal proof of \cref{thm:alg-pspin-glass} in the final subsection.
\subsection{Truncating the Series}
The next proposition tells us the depth at which we should truncate. We state it as a generic tool, and only later in the proof of \cref{thm:alg-pspin-glass} do we specialize to the case where $F(z) = Z_{\bm{G}}(z)$.
\begin{remark}
In the case where $Z_{\bm{G}}(z)$ is a polynomial in $z$, it suffices to truncate the series at logarithmic depth in terms of the degree of that polynomial, as explained in~\cites[{Section 2.1}]{Bar16book}. However, we are dealing with an analytic function, and therefore need a different argument to bound the tail of the series \cref{eq:taylor}.
\end{remark}
\begin{proposition}\label{prop:guarantee Taylor}
Let $F : \Omega \to \C$ be an analytic function, where $\Omega \supseteq \overline{\D(0,R)}$ for some $R > 0$. Suppose $F$ is nonzero on $\D(0,R)$, and that there exists $L > 1$ such that $\left|\frac{F(z)}{F(0)}\right| \leq L $ uniformly for all $z \in \overline{\D(0,R)}$. Then for any $0 < r < R$, $z \in \D(0,r)$ and $\eta > 0$, there exists $C$ depending only on $r/R$ such that for $m = C\log\wrapp{\frac{\pi + \log L}{\eta}}$ and $f(z) \defeq \log F(z)$,
\begin{align*}
    \abs{\log F(z) - \sum_{k=0}^{m} \frac{f^{(k)}(0)}{k!} \cdot z^{k}} \leq \eta.
\end{align*}
\end{proposition}
\begin{proof}
Note that by absorbing the $k=0$ term in the Taylor expansion into $\log F(z)$, we may assume that $\log F(0) = 0$. By Cauchy's Integral Formula,
\begin{align}\label{eq:cauchy}
\begin{split}
    \abs{f^{(k)}(0)} &= \abs{\frac{k!}{2\pi i} \int_{\partial \D(0,R)} \frac{f(\omega)}{\omega^{k+1}} \,d\omega} = \abs{\frac{k!}{2\pi i} \int_{0}^{2\pi} \frac{f\wrapp{Re^{i\theta}}}{R^{k+1} e^{i\theta(k+1)}} \cdot iRe^{i\theta} \,d\theta} \\
    &\leq \frac{k!}{2\pi R^{k}} \int_{0}^{2\pi} \abs{f\wrapp{Re^{i\theta}}} \,d\theta.
\end{split}
\end{align}
To bound the integral, we decompose the interval $[0,2\pi]$ into two parts, both of which are measurable:
\begin{align*}
    A &\defeq \wrapc{\theta \in [0,2\pi] : \abs{F\wrapp{Re^{i\theta}}} < 1},
    \\
    B &\defeq \wrapc{\theta \in [0,2\pi] : \abs{F\wrapp{Re^{i\theta}}} \geq 1}.
\end{align*}
Recalling $f(z) = \log F(z)$ and the fact that $\abs{f(z)} \leq \abs{\log \abs{F(z)}} + 2\pi$, we see that
\begin{align*}
   \frac{1}{2\pi} \int_{0}^{2\pi} \abs{f\wrapp{Re^{i\theta}}} \,d\theta \leq -\frac{1}{2\pi}\int_{A} \log \abs{F\wrapp{Re^{i\theta}}} \,d\theta + \frac{1}{2\pi}\int_{B} \log \abs{F\wrapp{Re^{i\theta}}} \,d\theta + 2\pi.
\end{align*}
Now, Jensen's Formula (see \cref{thm:jensen-zeros}) tells us that
\begin{align*}
    -\frac{1}{2\pi}\int_{A} \log \abs{F\wrapp{Re^{i\theta}}} \,d\theta &= -\frac{1}{2\pi}\int_{0}^{2\pi} \log \abs{F\wrapp{Re^{i\theta}}} \,d\theta + \frac{1}{2\pi}\int_{B} \log \abs{F\wrapp{Re^{i\theta}}} \,d\theta \\
    &= -\log \abs{F(0)} + \frac{1}{2\pi}\int_{B} \log \abs{F\wrapp{Re^{i\theta}}} \,d\theta \\
    &\leq \log L. \tag{Boundedness of $F$ and $\log F(0) = 0$}
\end{align*}
This implies that
\begin{align*}
    \frac{1}{2\pi} \int_{0}^{2\pi} \abs{f\wrapp{Re^{i\theta}}} \,d\theta \leq 2\pi + 2 \log L
\end{align*}
from which it follows by~\cref{eq:cauchy} that
\begin{align*}
    \frac{\abs{f^{(k)}(0)}}{k!} \cdot \abs{z}^{k} \leq 2(\pi + \log L) \cdot \wrapp{\frac{r}{R}}^{k}.
\end{align*}
Hence, by truncating the series~\cref{eq:taylor} at depth $m$, the resulting error will be bounded by
\begin{align*}
    \sum_{k=m+1}^{\infty} \frac{\abs{f^{(k)}(0)}}{k!} \cdot \abs{z}^{k} \leq 2(\pi + \log L) \cdot \sum_{k=m+1}^{\infty} \wrapp{\frac{r}{R}}^{k} \leq 2 (\pi + \log L) \cdot \wrapp{\frac{r}{R}}^{m+1} \cdot \frac{1}{1 - \frac{r}{R}}.
\end{align*}
To obtain the desired approximation guarantee, it suffices to take $m=C\log\wrapp{\frac{\pi + \log L}{\eta}}$ for some constant $C$ depending only on $r/R$.
\end{proof}

\subsection{Computing the Taylor Coefficients}
To turn \cref{prop:guarantee Taylor} into an algorithm, we need to compute the coefficients $f^{(k)}(0)$ for any realization $\bm{G}$, where recall that $f(z) = \log Z_{\bm{G}}(z)$.
The next lemma tells us how to do so given access to the moments of $\mathcal{H}_{\bm{G}}(\sigma)$, where $\sigma$ is drawn from the Gibbs measure with inverse temperature $\beta = 0$ (i.e. $\sigma \sim \varrho$).
\begin{lemma}\label{lem:compute coefficients of f}
Given access to the first $m+1$ moments of $\mathcal{H}_{\bm{G}}(\sigma)$ where $\sigma \sim \varrho$, we can compute the first $m+1$ Taylor coefficients of $f(z) = \log Z_{\bm{G}}(z)$ at $0$ deterministically in $O(m^{2})$-time.
\end{lemma}
\begin{proof}
Since $Z_{\bm{G}}^{(j)}(0) = \E_{\sigma \sim \varrho}\wrapb{\mathcal{H}_{\bm{G}}(\sigma)^{j}}$, this lemma follows from standard tools; see e.g. \cite[Section 2.2.2]{Bar16book} (where we note that \cite[Section 2.2.2]{Bar16book} applies to the setting of functions $g$ such that $g(0)\neq 0$ that are analytic near $0$.)
\end{proof}

Since $\varrho$ is a very simple distribution ($\varrho = \mathsf{Unif}(\mathcal{C}_{n})$ or $\varrho = \mathsf{Unif}(\mathcal{S}_{n})$), we can compute the moments of $\mathcal{H}_{\bm{G}}(\sigma)$ for $\sigma \sim \varrho$ in a direct way using brute force enumeration.
\begin{lemma}\label{lem:compute-moments}
The first $m+1$ moments of $\mathcal{H}_{\bm{G}}(\sigma)$, where $\sigma \sim \varrho$ (and $\varrho = \mathsf{Unif}(\mathcal{C}_{n})$ or $\varrho = \mathsf{Unif}(\mathcal{S}_{n})$), can be computed deterministically in $n^{O(m)}$-time.
\end{lemma}
\begin{proof}
To expand the $k$th moment $\mathcal{H}_{\bm{G}}(\sigma)^{k}$, we set up some shorthand notation for convenience. For $2 \leq p \leq p_{\max}$ and $\alpha = (i_{1},\dots,i_{p}) \in [n]^{p}$, we write $\bm{G}_{\alpha} \defeq \bm{G}_{i_{1},\dots,i_{p}} \sim \mathcal{N}(0,1)$ and $\sigma^{\alpha} \defeq \prod_{j=1}^{p} \sigma_{i_{j}}$. We can then expand the $k$th moment using linearity of expectation as follows:
\begin{align}\label{eq:compute coefficients exp}
\begin{split}
     \E_{\sigma \sim \varrho}\wrapb{\mathcal{H}_{\bm{G}}(\sigma)^{k}} &= \sum_{\phi : [k] \to \{2,\dots,p_{\max}\}} \sum_{\substack{\alpha_{j} \in [n]^{\phi(j)} \\ \forall \, j=1,\dots,k}} \wrapp{\prod_{j=1}^{k} \frac{\upgamma_{\phi(j)}}{n^{\frac{\phi(j) - 1}{2}}} \cdot \bm{G}_{\alpha_{j}}} \cdot \E_{\sigma \sim \varrho}\wrapb{\prod_{j=1}^{k} \sigma^{\alpha_{j}}}.
\end{split}
\end{align}
For a given $\bm{\alpha} = (\alpha_{1},\dots,\alpha_{k})$ in the summation above, we have that $\prod_{j=1}^{k} \sigma^{\alpha_{j}}$ is a polynomial in entries of $\sigma$, i.e. it can be written as $\prod_{j=1}^{n} \sigma_{j}^{\hat{\alpha}_{j}}$ for a unique $\hat{\alpha} \in \mathbb{\Z}_{\geq0}^{n}$ depending on $\alpha$.
We will say that $\bm{\alpha}$ is \emph{even} if each $\hat{\alpha}_{j}$ is even and \emph{odd} otherwise.
Now, to compute the expectation in~\cref{eq:compute coefficients exp}, we distinguish between the discrete and continuous cases.

If $\varrho$ is the uniform measure on the hypercube $\mathcal{C}_{n}=\{\pm1\}^{n}$, then it follows by symmetry that
\begin{align}\label{eq:expectation discrete}
\E_{\sigma \sim \varrho}\wrapb{\prod_{j=1}^{k} \sigma^{\alpha_{j}}} = \begin{cases}
    1, &\quad\text{if } \bm{\alpha} \text{ is even} \\
    0, &\quad\text{if } \bm{\alpha} \text{ is odd}
\end{cases}.
\end{align}
If $\varrho$ is the uniform measure on the normalized sphere $\mathcal{S}_{n} = \wrapc{\sigma \in \R^{n} : \norm{\sigma}_{2}^{2} = n}$, then the above expectation is again $0$ when $\bm{\alpha}$ is odd. When $\bm{\alpha}$ is even, it is known that
\begin{align}\label{eq:expectation-continuous-even}
    \E_{\sigma \sim \varrho}\wrapb{\prod_{j=1}^{k} \sigma^{\alpha_{j}}} = n^{-\frac{n-1}{2} + \frac{1}{2}\sum_{j=1}^{n} \hat{\alpha}_{j}} \cdot \pi^{-n/2} \cdot \frac{2 \cdot \Gamma\wrapp{1 + \frac{n}{2}} \cdot \prod_{j=1}^{n} \Gamma\wrapp{\frac{1 + \hat{\alpha}_{j}}{2}}}{\Gamma\wrapp{\sum_{j=1}^{n} \frac{1 + \hat{\alpha}_{j}}{2}}},
\end{align}
where $\Gamma(\cdot)$ denotes the Gamma function. See~\cite{Fol01} for a nice explanation of this fact. All these intermediate quantities can be easily computed in polynomial time.

In either case, to compute the $k$th moment, we can simply enumerate all $\phi : [k] \to \{2,\dots,p_{\max}\}$ and all even $\bm{\alpha} \in [n]^{\phi(1)} \times \dotsb \times [n]^{\phi(k)}$, and compute \cref{eq:compute coefficients exp} via brute force. Since $p_{\max}$ is constant, this takes $n^{O(k)}$-time, completing the proof.
\end{proof}

\subsection{Proof of \texorpdfstring{\cref{thm:alg-pspin-glass}}{Theorem Algorithms}}
\begin{proof}[Proof of~\cref{thm:alg-pspin-glass}]
Assume $\beta\in \D(0,(1-\varepsilon) \cdot \betasecond)$, and set $F(z) = Z_{\bm{G}}(z)$, $r = (1 - \varepsilon) \cdot \betasecond$, $R = \wrapp{1 - \frac{\varepsilon}{2}} \cdot \betasecond$ and $m = C\log(C'n/\eta)$ where $C = \Theta(1/\varepsilon)$.

Applying \cref{prop:guarantee Taylor}, we wish to satisfy the assumptions that (with high probability), $Z_{\bm{G}}(z)$ is nonzero on $\D(0,R)$, and $\lvrv{\frac{Z_{\bm{G}}(z)}{Z_{\bm{G}}(0)}} \leq L = e^{O(n)}$ uniformly for all $z \in \overline{\D(0,R)}$. The first assumption follows from \cref{thm:zeros-pspin-glass} (with $\varepsilon$ replaced by $\varepsilon/2$), while the latter follows from \cref{prop:logZ-crude}.

This yields the first claim of \cref{thm:alg-pspin-glass}, i.e. that exponentiating the degree-$m$ Taylor approximation $\widehat{P}(\beta)$ of $\log Z_{\bm{G}}(\beta)$ yields a $e^{\pm\eta}$-multiplicative approximation to $Z_{\bm{G}}(\beta)$ with probability at least $1 - o_{n}(1)$. For the algorithmic part, note that we can compute $\widehat{P}(\beta)$ using \cref{lem:compute coefficients of f,lem:compute-moments} in time $n^{O(m)} = n^{O(\log(n/\eta))}$.

\end{proof}

\section{Discussion and Future Directions}\label{sec:conclusion}
We conclude this paper with some open problems. The first pertains to algorithms with better running time. 
\begin{openproblem}[Faster algorithms]\label{open:faster}
Design a genuine polynomial time algorithm for estimating the partition function $Z_{\bm{G}}(\beta)$ up to the second moment threshold.
\end{openproblem}
We note that for partition functions defined on bounded degree graphs one can design genuine polynomial time algorithms~\cite{PR17} based on appropriate zero-freeness. Thus far it is not clear how to improve the quasi-polynomial running time for mean field models. 

The next two open problems pertain to extending the range in which we obtain efficient and provably correct algorithms.

\begin{openproblem}[The Replica-Symmetric Phase]\label{open:RS}
Design an algorithm for estimating the partition function $Z_{\bm{G}}(\beta)$ in the entire replica-symmetric phase, as discussed in \cref{subsec:compare-thresholds}.
\end{openproblem}

\begin{openproblem}[External Fields]\label{open:fields}
For the Sherrington--Kirkpatrick model, design an algorithm for estimating the partition function in the presence of an external field (either in the direction $\allone$ or in a random direction $\bm{g} \sim \mathcal{N}(0,I)$), up to the \emph{de Almeida--Thouless line} \cite{dAT78}.
\end{openproblem}

While a direct application of the second moment method fails to resolve \cref{open:RS} and \cref{open:fields}, it is conceivable that a suitably truncated version of the argument still works, perhaps taking inspiration from e.g. \cite{Tal98, Bov06, HS24}.

Another natural problem is to establish a zero-free region for the partition function in terms of external fields, rather than the inverse temperature.
\begin{openproblem}[Lee--Yang Zeros]\label{open:LY}
Establish a zero-free region for the function
\begin{align*}
    Z_{\bm{G},\beta}(\lambda) \defeq \E_{\sigma \sim \varrho}\wrapb{\exp\wrapp{\beta \cdot \mathcal{H}_{\bm{G}}(\sigma) + \lambda \cdot \langle \sigma, \allone \rangle}},
\end{align*}
and its multivariate analog.
\end{openproblem}

{\sloppy
\printbibliography
}

\appendix

\section{Analytic Tools}
We collect here the analytic tools used in the proof of Lemma~\ref{lem:CW-integral}.

%We first state the analytic lemmas we will use.
\begin{lemma}[Local Concavity; see e.g. Lemma B.11 in \cite{MSB21}]\label{lem:local-concave}
Let $\varphi : [-1,1] \to \R$ be twice continuously differentiable, and have a unique global maximizer $x^{*}$. Further assume that $x^{*}$ lies in the interior $(-1,1)$, and that $\varphi''\wrapp{x^{*}} = -\varepsilon$ for some $\varepsilon > 0$. Then there exists $\delta = \delta(\varepsilon) > 0$ such that both of the following hold:
\begin{enumerate}
    \item\label{item:strong-concave} $\varphi''(x) \leq - \frac{\varepsilon}{2}$ on $\mathcal{I}_{\delta} \defeq [x^{*} - \delta, x^{*} + \delta]$.
    \item\label{item:perturb} $\sup_{x \in [-1,1] \setminus\mathcal{I}_{s}} \varphi(x)$ is attained at one of the endpoints $x^{*} \pm s$ of $\mathcal{I}_{s}$ for every $0 \leq s \leq \delta$. In particular, for any $0 \leq s < t \leq \delta$,
    \begin{align*}
        \sup_{x \in [-1,1] \setminus\mathcal{I}_{t}} \varphi(x) \leq \sup_{x \in [-1,1] \setminus\mathcal{I}_{s}} \varphi(x) - C\varepsilon (t - s)^{2}
    \end{align*}
    for some $C > 0$, by $\frac{\varepsilon}{2}$-strong concavity of $\varphi$ in $\mathcal{I}_{\delta}$.
\end{enumerate}
\end{lemma}

\begin{lemma}[Stirling's Approximation for the Gamma Function]\label{lem:stirling-gamma}
For any $x > 0$, we have
\begin{align*}
    \Gamma(x) = \sqrt{\frac{2\pi}{x}} \cdot \wrapp{\frac{x}{e}}^{x} \cdot \wrapp{1 + O\wrapp{\frac{1}{x}}}.
\end{align*}
\end{lemma}
Now recall that when $\varrho = \mathsf{Unif}(\mathcal{C}_{n})$, $\ov(\varrho)$ is given by the probability mass function
\begin{align*}
    \ov(\varrho)(m) = \frac{1}{2^{n}} \cdot \binom{n}{\frac{1 + m}{2} \cdot n} = \frac{1}{2^{n}} \cdot \frac{\Gamma(n+1)}{\Gamma\wrapp{\frac{1+m}{2} \cdot n + 1} \cdot \Gamma\wrapp{\frac{1-m}{2} \cdot n + 1}}
\end{align*}
for all $m \in \wrapc{-1, -1+\frac{2}{n},\dots,1-\frac{2}{n},1}$. Similarly, when $\varrho = \mathsf{Unif}(\mathcal{S}_{n})$, $\ov(\varrho)$ is given by the probability density function
\begin{align*}
    \ov(\varrho)(m) = \frac{1}{\sqrt{\pi}} \cdot \frac{\Gamma\wrapp{\frac{n}{2}}}{\Gamma\wrapp{\frac{n-1}{2}}} \cdot (1 - m^{2})^{\frac{n-3}{2}}
\end{align*}
for all $m \in [-1,1]$. Note that in the latter case, by spherical symmetry and rescaling, $\ov(\varrho)$ may be equivalently viewed as describing the law of the first coordinate of a uniformly random unit vector in $\R^{n}$. The following inequalities can be derived from \cref{lem:stirling-gamma} in a straightforward manner.
\begin{corollary}[Slice Measure Approximations]\label{cor:stirling}
Let $\varrho = \mathsf{Unif}(\mathcal{C}_{n})$ (resp. $\varrho = \mathsf{Unif}(\mathcal{S}_{n})$), and let $h: [-1,1] \to \R$ be defined as in \cref{def:2nd-moment-regime}. Then for all $m \in \supp\wrapp{\ov(\varrho)}$, we have the following pointwise estimate:
\begin{align}\label{eq:stirling-polyn-loss}
    \frac{1}{\poly(n)} \cdot \exp(n \cdot h(m)) \leq \ov(\varrho)(m) \leq \poly(n) \cdot \exp(n \cdot h(m)).
\end{align}
Moreover, if $m$ is bounded away from the endpoints $\pm1$, then the following hold.
\begin{enumerate}
    \item\label{item:stirling-interior-cube} In the case $\varrho = \mathsf{Unif}(\mathcal{C}_{n})$, we have
    \begin{align*}
        \ov(\varrho)(m) = (1 \pm o_{n}(1)) \cdot \sqrt{\frac{2}{\pi n \cdot (1 - m^{2})}} \cdot \exp(n \cdot h(m)).
    \end{align*}
    \item\label{item:stirling-interior-sphere} In the case $\varrho = \mathsf{Unif}(\mathcal{S}_{n})$, we have
    \begin{align*}
        \ov(\varrho)(m) = (1 \pm o_{n}(1)) \cdot \sqrt{\frac{n}{2\pi}} \cdot \frac{1}{(1 - m^{2})^{3/2}} \cdot \exp(n \cdot h(m)).
    \end{align*}
\end{enumerate}
\end{corollary}

\begin{lemma}[Laplace Approximation; Simplified Version of Lemma A.3 \cite{MSB21}]\label{lem:laplace-approx}
Fix an interval $[a,b]$ and a point $x^{*}$ bounded away from the endpoints $a,b$. Let $f,g : [a,b] \to \R$ be thrice continuously differentiable on $(a,b)$, and assume that $g'\wrapp{x^{*}} = 0$ and $g''\wrapp{x^{*}} < 0$. Then for all $\alpha \in (0,1/6)$,
\begin{align*}
    \sqrt{\frac{n}{2\pi}} \cdot \int_{x^{*} - n^{-\frac{1}{2} + \alpha}}^{x^{*} + n^{-\frac{1}{2} + \alpha}} f(x) \cdot \exp(n \cdot g(x)) \,dx = (1 + o_{n}(1)) \cdot \sqrt{\frac{1}{\abs{g''\wrapp{x^{*}}}}} \cdot f\wrapp{x^{*}} \cdot \exp\wrapp{n \cdot g\wrapp{x^{*}}}.
\end{align*}
\end{lemma}

\begin{lemma}[Riemann Approximation; see e.g. Lemma A.2 in \cite{MSB21}]\label{lem:riemann-approx}
Let $f : [a,b] \to \R$ be a differentiable function, and let $a = x_{0} < x_{1} < \dotsb < x_{n} = b$. Let $x_{k}^{*} \in [x_{k-1},x_{k}]$ for each $1 \leq k \leq n$. Then
\begin{align*}
    \abs{\int_{a}^{b} f(x) \,dx - \sum_{k=1}^{n} (x_{k} - x_{k-1}) \cdot f(x_{k}^{*})} \leq \frac{b-a}{2} \cdot \max_{1 \leq k \leq n} \{x_{k} - x_{k-1}\} \cdot \sup_{x \in [a,b]} \abs{f'(x)}
\end{align*}
\end{lemma}

\section{Miscellaneous Calculations and Proofs}\label{sec:misc-proof}
\begin{lemma}[Second Moment Interpretation]\label{lem:2nd-moment-interpret}
For every $\beta \geq 0$, we have
\begin{align*}
    Z_{\CW}\wrapp{\beta^{2}} = \frac{\E_{\bm{G}}\wrapb{Z_{\bm{G}}(\beta)^{2}}}{\E_{\bm{G}}\wrapb{Z_{\bm{G}}(\beta)}^{2}}.
\end{align*}
\end{lemma}
\begin{proof}
For the numerator, we have that
\begin{align*}
    \E_{\bm{G}}\wrapb{Z_{\bm{G}}(\beta)^{2}} &= \E_{\tau,\sigma \sim \varrho}\wrapb{\E_{\bm{G}}\wrapb{\exp\wrapp{\beta \cdot \mathcal{H}_{\bm{G}}(\tau) + \beta \cdot \mathcal{H}_{\bm{G}}(\sigma)}}} \\
    &= \E_{\tau,\sigma \sim \varrho}\wrapb{\exp\wrapp{\frac{\beta^{2}}{2} \cdot \E_{\bm{G}}\wrapb{\wrapp{\mathcal{H}_{\bm{G}}(\tau) + \mathcal{H}_{\bm{G}}(\sigma)}^{2}}}} \\
    &= \E_{\tau,\sigma \sim \varrho}\wrapb{\exp\wrapp{n \cdot \beta^{2} \cdot \xi(1) + n \cdot \beta^{2} \cdot \xi\wrapp{\frac{\langle\tau,\sigma\rangle}{n}}}} \\
    &= \exp\wrapp{n \cdot \xi(1) \cdot \beta^{2}} \cdot Z_{\CW}\wrapp{\beta^{2}}.
\end{align*}
On the other hand,
\begin{align*}
    \E_{\bm{G}}\wrapb{Z_{\bm{G}}(\beta)}^{2} &= \E_{\sigma \sim \varrho}\wrapb{\E_{\bm{G}}\wrapb{\exp\wrapp{\beta \cdot \mathcal{H}_{\bm{G}}(\sigma)}}}^{2} \\
    &= \E_{\sigma \sim \varrho}\wrapb{\exp\wrapp{\frac{\beta^{2}}{2} \cdot \E_{\bm{G}}\wrapb{\mathcal{H}_{\bm{G}}(\sigma)^{2}}}}^{2} \\
    &= \exp\wrapp{n \cdot \xi(1) \cdot \beta^{2}}.
\end{align*}
Taking a ratio concludes the proof.
\end{proof}

%\subsection{A Simple Argument for the Sherrington--Kirkpatrick Model}
\begin{lemma}[see e.g. Lemma 2.1 in \cite{Tal98}]\label{lem:CW-bound}
Consider the special case $\xi(s) = \frac{1}{2}s^{2}$ and $\varrho = \mathsf{Unif}(\mathcal{C}_{n})$. Then for every $0 \leq \beta < \betasecond = 1$, we have
\begin{align*}
    \log Z_{\CW}\wrapp{\beta^{2}} \leq -\frac{1}{2} \log\wrapp{1-\beta^{2}}.
\end{align*}
\end{lemma}
\begin{proof}
We follow the proof in \cite{Tal98}. Observe that
\begin{align*}
    Z_{\CW}\wrapp{\beta^{2}} &= \E_{\sigma \sim \varrho}\wrapb{\exp\wrapp{\frac{\beta^{2}}{2n} \langle \sigma, \allone \rangle^{2}}} \\
    &= \E_{\sigma \sim \varrho}\wrapb{\E_{g \sim \mathcal{N}(0,1)}\wrapb{\exp\wrapp{\frac{\beta}{\sqrt{n}} \langle \sigma, \allone \rangle \cdot g}}} \tag{Moment Generating Function of a Gaussian} \\
    &= \E_{g \sim \mathcal{N}(0,1)}\wrapb{\E_{\sigma \sim \varrho}\wrapb{\prod_{i=1}^{n} \exp\wrapp{\frac{\beta}{\sqrt{n}} \cdot \sigma_{i} \cdot g}}} \\
    &= \E_{g \sim \mathcal{N}(0,1)}\wrapb{\E_{s \sim \mathsf{Unif}\{\pm1\}}\wrapb{\exp\wrapp{\frac{\beta}{\sqrt{n}} \cdot s \cdot g}}^{n}} \tag{$\sigma_{1},\dots,\sigma_{n} \sim \mathsf{Unif}\{\pm1\}$ independent} \\
    &\leq \E_{g \sim \mathcal{N}(0,1)}\wrapb{\exp\wrapp{\frac{\beta^{2}}{2} \cdot g^{2}}} \tag{Subgaussianity of $\mathsf{Unif}\{\pm1\}$, i.e. $\cosh(x) \leq \exp(x^{2}/2)$} \\
    &= \frac{1}{\sqrt{1-\beta^{2}}}. \tag{Moment Generating Function for a $\chi^{2}$ Distribution}
\end{align*}
Taking logarithms completes the proof.
\end{proof}

\section{Motivation of Reweighting factor}\label{sec:reweight-factor-motivation}

In this appendix, we explain how to guess the reweighting factor $A_{\bm{G}}(\beta)$ in \cref{eq:def-A-bG}.
\textbf{The following calculation is for motivation early, and does not enter into the rigorous part of our argument}.
We give this heuristic derivation in the setting of Ising spins $\sigma \in \{\pm 1\}^n$.
For the spherical setting, we will use the same reweighting factor, as we expect the sphere and cube partition functions to match up to a stochastic $1+o_n(1)$ factor in the second moment phase.
For convenience, recall that
\[
  Z_{\bm{G}}(\beta) = \frac{1}{2^n} \sum_{\sigma \in \{\pm 1\}^n} \exp(\beta \mathcal{H}_{\bm{G}}(\sigma)).
\]
We will aim to guess a function $A_{\bm{G}}(\beta)$ which is analytic in $\beta \in \D(0,R)$ and $\bm{G} \in \R^{n^2 \times \cdots \times n^{p_{\max}}}$, and which stochastically approximates $Z_{\bm{G}}(\beta)^{-1}$, i.e. satisfies
\[
  Z_{\bm{G}}(\beta) A_{\bm{G}}(\beta) = 1+o_n(1) \qquad \text{whp}.
\]
Below, let $\approx$ denote equality up to a (in-probability) multiplicative factor of $1+o_n(1)$.
First, write
\[
  \mathcal{H}_{\bm{G}}(\sigma)
  = \frac{1}{2} \langle \nabla^2 \mathcal{H}_{\bm{G}}(\bm{0}) \sigma, \sigma \rangle
  + \mathcal{H}_{\bm{G},\ge3}(\sigma)
  = \frac{1}{2} \langle \bm{M}\sigma, \sigma \rangle
  + \mathcal{H}_{\bm{G},\ge3}(\sigma),
\]
where $\bm{M}$ is defined in \cref{eq:nabla2H-expansion,eq:def-bM} and $\mathcal{H}_{\bm{G},\ge3}$ consists of the parts of $\mathcal{H}_{\bm{G}}$ with degree $\ge3$.
Note that $\mathcal{H}_{\bm{G},\ge3}$ is a spin glass with mixture function
\[
  \xi_{\ge3}(s) = \sum_{p=3}^{p_{\max}} \gamma_p^2 s^p
  = \xi(s) - \frac{1}{2} \xi''(0) s^2.
\]
We expect that the degree $\ge3$ parts of $\mathcal{H}_{\bm{G}}$ contribute a multiplicative fluctuation of $1+o_n(1)$.
That is, for $\E_{\ge3}$ denoting expectation over just $\mathcal{H}_{\bm{G},\ge3}$,
\begin{align}
  \nonumber
  Z_{\bm{G}}(\beta)
  \approx \E_{\ge3} Z_{\bm{G}}(\beta)
  &= \exp\wrapp{
    \frac{n\beta^2}{2} \xi_{\ge3}(1)
  }
  \cdot \frac{1}{2^n}
  \sum_{\sigma \in \{\pm 1\}^n} \exp\wrapp{\frac{\beta}{2} \langle \bm{M}\sigma,\sigma \rangle} \\
  \label{eq:alr-first-step}
  &= \exp\wrapp{
    \frac{n\beta^2}{2} \xi(1)
    - \frac{n\zeta^2}{4}
  }
  \cdot \frac{1}{2^n}
  \sum_{\sigma \in \{\pm 1\}^n}
  \exp\wrapp{\frac{\beta}{2} \langle \bm{M}\sigma,\sigma \rangle}.
\end{align}
We expand the remaining partition function of the degree $2$ interaction following the approach of \cite{ALR87}.
Recall $\UC(n)$ defined in \cref{eq:uc}, and define analogously
\begin{align*}
  \Ev(n) &= \Big\{\Gamma \subseteq \binom{[n]}{2} : \text{$\Gamma$ is the edge set of a graph with all degrees even}\Big\}, \\
  \Cy(n) &= \Big\{\gamma \subseteq \binom{[n]}{2} : \text{$\gamma$ is the edge set of a cycle}\Big\}.
\end{align*}
Then,
\begin{align}
  \nonumber
  &\frac{1}{2^n}
  \sum_{\sigma \in \{\pm 1\}^n}
  \exp\wrapp{\frac{\beta}{2} \langle \bm{M}\sigma,\sigma \rangle} \\
  \nonumber
  &= \prod_{i=1}^n \exp\wrapp{\frac{\beta}{2} \bm{M}_{i,i}}
  \prod_{1\le i<j\le n} \cosh(\beta \bm{M}_{i,j}) \cdot
  \frac{1}{2^n}
  \sum_{\sigma \in \{\pm 1\}^n}
  \prod_{1\le i<j\le n}
  (1 + \sigma_i\sigma_j \tanh(\beta \bm{M}_{i,j})) \\
  \label{eq:alr-second-step}
  &= \exp\wrapp{\frac{\beta}{2} \tr(\bm{M}) }
  \wrapp{
    \prod_{1\le i<j\le n} \cosh(\beta \bm{M}_{i,j})
  }
  \wrapp{
    \sum_{\Gamma \in \Ev(n)}
    \prod_{e\in \Gamma} \tanh(\beta \bm{M}_e)
  }.
\end{align}
We aim to guess analytic approximations to the reciprocals of the last two factors.
First, note that
\[
  \cosh(x) = \exp\wrapp{\frac{1}{2} x^2 - \frac{1}{12} x^4 + O(x^6)},
\]
and thus
\begin{align*}
  \wrapp{\prod_{1\le i<j\le n} \cosh(\beta \bm{M}_{i,j})}^{-1}
  &= \exp\wrapp{
    - \frac{1}{2}
    \sum_{1\le i<j\le n}
    (\beta \bm{M}_{i,j})^2
    + \frac{1}{12}
    \sum_{1\le i<j\le n}
    (\beta \bm{M}_{i,j})^4
    + \sum_{1\le i<j\le n}
    O(\bm{M}_{i,j}^6)
  } \\
  &\approx \exp\wrapp{
    -\frac{\beta^2}{2}
    \sum_{1\le i<j\le n}
    \bm{M}_{i,j}^2
    + \frac{\beta^4}{12}
    \sum_{1\le i<j\le n}
    \bm{M}_{i,j}^4
  }.
\end{align*}
Furthermore, recalling \cref{eq:nabla2H-expansion,eq:def-bM}, we expect
\[
  \frac{\beta^4}{12}
  \sum_{1\le i<j\le n}
  \bm{M}_{i,j}^4
  = \frac{\beta^4}{12}
  \sum_{1\le i<j\le n}
  \E[\bm{M}_{i,j}^4] + o_n(1)
  = \frac{\beta^4}{12}
  \sum_{1\le i<j\le n}
  \frac{3\xi''(0)^2}{n^2} + o_n(1)
  = \frac{\zeta^4}{8} + o_n(1).
\]
Thus
\begin{align}
  \label{eq:alr-cosh-term}
  \wrapp{\prod_{1\le i<j\le n} \cosh(\beta \bm{M}_{i,j})}^{-1}
  \approx \exp\wrapp{
    -\frac{\beta^2}{2}
    \sum_{1\le i<j\le n}
    \bm{M}_{i,j}^2
  }
  \exp(\zeta^4/8).
\end{align}
We turn to the final term in \cref{eq:alr-second-step}.
As argued in \cite{ALR87},
\[
  \sum_{\Gamma \in \Ev(n)}
  \prod_{e\in \Gamma} \tanh(\beta \bm{M}_e)
  \approx
  \sum_{\Gamma \in \Ev(n)}
  \prod_{e\in \Gamma} (\beta \bm{M}_e)
  \approx
  \sum_{\Gamma \in \UC(n)}
  \prod_{e\in \Gamma} (\beta \bm{M}_e)
  \approx
  \prod_{\gamma \in \Cy(n)}
  \wrapp{
    1 + \prod_{e\in \gamma} (\beta \bm{M}_e)
  }.
\]
The intuition is that for a fixed $|\Gamma| \le K$, there are far more $\Gamma \in \UC(n)$ than $\Gamma \in \Ev(n) \setminus \UC(n)$ because the latter graphs have to reuse a vertex.
This results in a sub-leading order contribution.
Moreover, in the high-temperature phase the cluster expansion decays geometrically in $|\Gamma|$, so the subgraphs $\Gamma \in \Ev(n)$ with $|\Gamma| > K$ also contribute sub-leading order.
Similarly when we expand the $\prod_{\gamma \in \Cy(n)}$ we get a sum of monomials corresponding to multi-graphs where some edges might be repeated, but those that aren't both simple and in $\UC(n)$ contribute sub-leading order.
Analogously, we expect
\[
  \sum_{\Gamma \in \UC(n)}
  (-1)^{c(\Gamma)}
  \prod_{e\in \Gamma} (\beta \bm{M}_e)
  \approx
  \prod_{\gamma \in \Cy(n)}
  \wrapp{
    1 - \prod_{e\in \gamma} (\beta \bm{M}_e)
  }.
\]
Thus
\begin{align*}
  \wrapp{
    \sum_{\Gamma \in \Ev(n)}
    \prod_{e\in \Gamma} \tanh(\beta \bm{M}_e)
  }^{-1}
  &\approx
  \prod_{\gamma \in \Cy(n)}
  \wrapp{
    1 + \prod_{e\in \gamma} (\beta \bm{M}_e)
  }^{-1} \\
  &=
  \prod_{\gamma \in \Cy(n)}
  \wrapp{
    1 - \prod_{e\in \gamma} (\beta \bm{M}_e)
  }
  \prod_{\gamma \in \Cy(n)}
  \wrapp{
    1 - \prod_{e\in \gamma} (\beta \bm{M}_e)^2
  }^{-1} \\
  &\approx
  \wrapp{
    \sum_{\Gamma \in \UC(n)}
    (-1)^{c(\Gamma)}
    \prod_{e\in \Gamma} (\beta \bm{M}_e)
  }
  \prod_{\gamma \in \Cy(n)}
  \wrapp{
    1 - \prod_{e\in \gamma} (\beta \bm{M}_e)^2
  }^{-1}.
\end{align*}
The last factor can be estimated as
\begin{align*}
  \prod_{\gamma \in \Cy(n)}
  \wrapp{
    1 - \prod_{e\in \gamma} (\beta \bm{M}_e)^2
  }^{-1}
  &\approx
  \exp\wrapp{
    \sum_{\gamma \in \Cy(n)}
    \prod_{e\in \gamma} (\beta \bm{M}_e)^2
    - \sum_{\gamma \in \Cy(n)}
    O\wrapp{
      \prod_{e\in \gamma}
      \bm{M}_e^4
    }
  } \\
  &\approx
  \exp\wrapp{
    \sum_{\gamma \in \Cy(n)}
    (\zeta^2 / n)^{|\gamma|}
  } \\
  &= \exp\wrapp{
    \sum_{k\ge 3}
    \frac{n(n-1)\cdots(n-k+1)}{2k}
    \cdot
    (\zeta^2 / n)^{k}
  } \\
  &\approx \exp\wrapp{
    \sum_{k\ge 3}
    \frac{\zeta^{2k}}{2k}
  }
  = (1-\zeta^2)^{-1/2} \exp\wrapp{
    - \frac{\zeta^2}{2} - \frac{\zeta^4}{4}
  }
\end{align*}
Thus,
\begin{align}
  \label{eq:alr-cluster-term}
  \wrapp{
    \sum_{\Gamma \in \Ev(n)}
    \prod_{e\in \Gamma} \tanh(\beta \bm{M}_e)
  }^{-1}
  \approx
  (1-\zeta^2)^{-1/2} \exp\wrapp{
    - \frac{\zeta^2}{2} - \frac{\zeta^4}{4}
  }
  \wrapp{
    \sum_{\Gamma \in \UC(n)}
    (-1)^{c(\Gamma)}
    \prod_{e\in \Gamma} (\beta \bm{M}_e)
  }.
\end{align}
Combining \cref{eq:alr-first-step,eq:alr-second-step,eq:alr-cosh-term,eq:alr-cluster-term} yields that $Z_{\bm{G}}(\beta) \approx A_{\bm{G}}(\beta)^{-1}$ for $A_{\bm{G}}(\beta)$ as in \cref{eq:def-A-bG}.

\end{document}